\documentclass[12pt,a4paper]{amsart}
\usepackage{amsmath,amssymb,amsthm}
\usepackage{latexsym}
\usepackage{amsfonts}
\usepackage{bbm,bm,dsfont} 
\usepackage{color} 
\usepackage{graphicx}

\numberwithin{equation}{section} 

\newtheorem{proposition}{Proposition}
\newtheorem{theorem}{Theorem}
\newtheorem{lemma}{Lemma}
\newtheorem{corollary}{Corollary}

\theoremstyle{definition}
\newtheorem{definition}{Definition}
\newtheorem{example}{Example}
\newtheorem{remark}{Remark}

\newcommand{\hI}{\mathcal{I}} 

\newcommand{\sfa}{\mathsf{A}}
\newcommand{\sfb}{\mathsf{B}}

\newcommand{\sfg}{\mathsf{G}}

\newcommand{\N}{\mathbb N} 
\newcommand{\R}{\mathbb R} 
\newcommand{\C}{\mathbb C} 
\newcommand{\T}{\mathbb T} 

\newcommand{\fii}{\varphi} 

\newcommand{\hil}{{\mathcal H}} 
\newcommand{\ki}{\mathcal{K}} 

\newcommand{\id}{\mathds1} 

\newcommand{\tr}[1]{\mathrm{tr}\left[#1\right]} 
\def\<{\langle} 
\def\>{\rangle} 
\newcommand{\ket}[1]{|#1\rangle} 
\newcommand{\kb}[2]{|#1 \rangle\langle #2|} 

\newcommand{\Ao}{\mathsf{A}} 
\newcommand{\Qo}{\mathsf{Q}} 
\newcommand{\Mo}{\mathsf{M}} 

\newcommand{\ov}{\overline} 

\newcommand{\mc}[1]{\mathcal{#1}}

\newcommand{\msf}[1]{\mathsf{#1}}

\newcommand{\tj}{\vartheta}

\newcommand{\Set}{\mathbb{X}}
\newcommand{\Orb}{\mathcal{O}} 
\newcommand{\orb}{\Omega}      

\begin{document}

\title[Optimal covariant quantum measurements]{Optimal covariant quantum measurements}

\author{Erkka Haapasalo}
\email{erkkath@gmail.com}
\address{Department of Physics, Fudan University, 200433 Shanghai, China (PRoC)}

\author{Juha-Pekka Pellonp\"a\"a}
\email{juhpello@utu.fi}
\address{Turku Centre for Quantum Physics, Department of Physics and Astronomy, University of Turku, FI-20014 Turku, Finland}

\begin{abstract}
%
We discuss symmetric quantum measurements and the associated covariant observables modelled, respectively, as instruments and positive-operator-valued measures. The emphasis of this work are the optimality properties of the measurements, namely, extremality, informational completeness, and the rank-1 property which contrast the complementary class of (rank-1) projection-valued measures. The first half of this work concentrates solely on finite-outcome measurements symmetric w.r.t.\ finite groups where we derive exhaustive characterizations for the pointwise Kraus-operators of covariant instruments and necessary and sufficient extremality conditions using these Kraus-operators. We motivate the use of covariance methods by showing that observables covariant with respect to symmetric groups contain a family of representatives from both of the complementary optimality classes of observables and show that even a slight deviation from a rank-1 projection-valued measure can yield an extreme informationally complete rank-1 observable. The latter half of this work derives similar results for continuous measurements in (possibly) infinite dimensions. As an example we study covariant phase space instruments, their structure, and extremality properties.
\end{abstract}
\maketitle


\section{Introduction}

Let us concentrate on a quantum system described by the Hilbert space $\hil$. In quantum mechanical description, observables are represented as (normalized) positive operator valued measures (POVMs) and states are density operators, i.e.\ trace-1 positive operators. If the outcome space $\Set$ is finite, one can see a POVM $\Mo$ as a collection of positive operators $\Mo_x$, $x\in\Set$, summing up to the identity $\id=\id_\hil$ (normalization). The number $\tr{\rho\Mo_x}$ is interpreted as a probability to get $x$ in the measurement of $\Mo$ when the system is prepared in the state $\rho$. In the complete description of a measurement, we need to specify how the detection of an outcome $x$ affects the input state $\rho$ and this is done by an instrument $\hI$, originally introduced by Davies and Lewis \cite{QTOS76,DaviesLewis}, which is a collection of completely positive linear maps $\rho\mapsto\hI_x(\rho)$, where $\hI_x(\rho)$ is a (non-normalized) output state conditioned by $x$, and $\sum_{x\in\Set}\tr{\hI_x(\rho)}\equiv1$ so that $\sum_{x\in\Set}\hI_x(\rho)$ is the unconditioned total state. Note that the output states $\hI_x(\rho)$ may reside in a different Hilbert space $\ki$. Moreover, we say that $\hI$ measures a POVM $\Mo$, or is an $\Mo$--instrument, if $\tr{\hI_x(\rho)}\equiv \tr{\rho\Mo_x}$. For more details on quantum measurement theory, we refer to \cite{kirja}.

Observables are charaterized by symmetries. For example, position observables transform covariantly under the position shifts (translations) generated by the momentum operator. It is well known that, in addition to the sharp position (i.e.\ the spectral measure of the position operator), there are infinitely many unsharp position POVMs which all are smearings of the sharp one. To define a symmetric or covariant POVM, one must start by fixing a symmetry of the outcome space. For this, we need an appropriate (finite) symmetry group $G$ which acts on $\Set$, i.e.\ any $g\in G$ `transforms' or `shifts' an outcome $x$ into $gx\in\Set$. The neutral element $e\in G$ does nothing: $ex=x$ (and $ge=g=eg$). Moreover, we let $\Set$ be a $G$-space, i.e.,\ $g(hx)=(gh)x$ for all $g,\,h\in G$ and $x\in\Set$. Also, $G$ is assumed to act on the operator space of the system: any operator $A$ in the Heisenberg picture transforms into $\alpha_g(A):=U(g)AU(g)^*$ where $U(g)$ is a unitary operator\footnote{In the Schr\"odinger picture, any density operator $\rho$ transforms to $U(g)^*\rho U(g)$ under the action of $g\in G$.} and $g\mapsto\alpha_g$ is a group homomorphism of $G$ into the automorphism group of the operator algebra. This means that we may choose $g\mapsto U(g)$ to be a projective unitary representation, i.e.,\ there is a multiplier or 2-cocycle $m:G\times G\to\T$ such that $U(gh)=m(g,h)U(g)U(h)$ for all $g,\,h\in G$. The 2-cocycle conditions read
$$
m(e,g)=m(g,e)=1,\qquad m(g,h)m(gh,k)=m(g,hk)m(h,k)
$$
for all $g,\,h,\,k\in G$. Now, by definition, a covariant POVM $\Mo$ satisfies the following covariance (or equivariance) condition:
\begin{equation}\label{kovarianssiehto}
\Mo_{gx}=U(g) \Mo_x U(g)^*,\qquad g\in G,\;x\in\Set,
\end{equation}
that is, for any unit vector $\psi\in\hil$, the shifted probability distribution
$x\mapsto \<\psi|\Mo_{gx}\psi\>$ is the same as $x\mapsto \<\psi_g|\Mo_{x}\psi_g\>$  where $\psi_g=U(g)^*\psi$ is the symmetrically transformed input state. Thus, changing the initial state should only move the probability distribution without deforming its shape. One can see the condition \eqref{kovarianssiehto} as a generalization of canonical quantization of the classical variable $x$ \cite{Holevo}, or as the definition of the generalized imprimitivity system \cite{Cattaneo, QTOS76, Mackeykirja, varadarajan}. 


Entire measurement settings can be symmetric in the sense that applying symmetry transformations on input states is the same as registering transformed values and obtaining conditional output states which are symmetrically transformed. We keep the above finite $G$-space $\Set$ and the input representation $U$ fixed and introduce output system symmetries via a projective unitary representation $g\mapsto V(g)$ operating on the output system Hilbert space $\mc K$. Symmetry of the measurement described by an instrument $\mc I$ now means that
\begin{equation}\label{eq:InstrKovEhto}
\mc I_{gx}\big(U(g)\rho U(g)^*\big)=V(g)\mc I_{x}(\rho)V(g)^*
\end{equation}
for all $x\in\Set$, $g\in G$, and all input states $\rho$. In this case we say that $\mc I$ is {\it $(\Set,U,V)$--covariant}. It easily follows that the observable measured by an $(\Set,U,V)$--covariant instrument is covariant w.r.t.\ $U$, i.e.,\ {\it $(\Set,U)$--covariant}.


Let us make some further definitions on our $G$-space $\Set$. Let $\Orb$ be the set of the orbits $Gx=\{gx\,|\,g\in G\}\subseteq\Set$. Thus, the outcome space is the disjoint union of the orbits, $\Set=\biguplus_{\orb\in\Orb}\orb$. From now on, for any $\orb\in\Orb$ and $x\in\orb$, we first {\it fix} $x_\orb\in\orb$ (i.e.\ $\orb=Gx_\orb$) and then $g_x\in G$ such that $x=g_x x_\orb$. By defining stability subgroups $H_\orb=\{h\in G\,|\,hx_\orb=x_\orb\}$ we see that $g_x$ is not necessarily unique. Indeed, for all $h\in H_\orb$, one gets $(g_xh)x_\orb=g_xx_\orb=x\in\orb$ and, from \eqref{kovarianssiehto},
$$
\Mo_{x}=\Mo_{g_xx_\orb}=U(g_x) \Mo_{x_\orb} U(g_x)^*=\Mo_{(g_xh)x_\orb}=U(g_x)U(h) \Mo_{x_\orb}U(h)^* U(g_x)^*.
$$
so that
$
U(h) \Mo_{x_\orb} = \Mo_{x_\orb}U(h).
$
By denoting $K_\orb:=\Mo_{x_\orb}$ we have proven the following theorem:
\begin{theorem}\label{theor:CovObsBasic}
A POVM $\Mo$ is covariant if and only if
$
\Mo_x=U(g_x)K_\orb U(g_x)^*
$ for all $x\in\orb\in\Orb$ where $K_\orb$ is a positive operator such that
$K_\orb U(h)=U(h)K_\orb$, $h\in H_\orb$. Now $\Mo$ is normalized exactly when 
$K:=\sum_{\orb\in\Orb}\sum_{x\in \orb}U(g_x)K_\orb U(g_x)^*=\id$. 
\end{theorem}
\noindent
Note that if $\Mo$ is not normalized (i.e.\ $K\ne\id$) but $K$ is invertible, one can define a normalized covariant POVM as the collection of effects $K^{-1/2}\Mo_x K^{-1/2}$, $x\in\Set$.\footnote{Indeed, $U(g)KU(g)^*=K$ so that $K$ and thus $K^{-1/2}$
 commutes with any $U(g)$. Note that the eigenvalues of $K$ (and $K^{-1/2}$) are positive.}
Moreover, we note that there is necessarily no nontrivial solution $\Mo$ for \eqref{kovarianssiehto}. 
For example, if there is only one orbit, $\Orb=\{\Set\}$, and the subrepresentation $h\mapsto U(h)$ of $H_\Set$ is irreducible, then $K_\Set=k\id$, $k\ge 0$  (by Schur's lemma). Thus, $\Mo_x=k\id$ for all $x\in\Set$.

We also obtain a similar preliminary characterization for covariant instruments which we will further refine later in this work.

\begin{theorem}\label{theor:CovInstrBasic}
An instrument $\mc I$ is $(\Set,U,V)$--covariant if and only if
$$
\mc I_x(\rho)=V(g_x)\Lambda_\orb\big(U(g_x)^*\rho U(g_x)\big)V(g_x)^*
$$
for all $x\in\orb\in\Orb$ where $\Lambda_\orb$ is a completely positive linear map such that $\Lambda_\orb\big(U(h)\rho U(h)^*\big)=V(h)\Lambda_\orb(\rho)V(h)^*$ for all $h\in H_\orb$ and all input states $\rho$.
Clearly, the normalization condition $\sum_{x\in\Set}\tr{\hI_x(\rho)}\equiv1$ holds if and only if  
$
\sum_{\orb\in\Orb}\sum_{x\in\orb}\tr{\Lambda_\orb\big(U(g_x)^*\rho U(g_x)\big)}\equiv1.
$
\end{theorem}
\noindent One easily sees that we may choose $\Lambda_\orb=\mc I_{x_\orb}$ for any $\orb\in\Orb$, and the theorem immediately follows using Equation \eqref{eq:InstrKovEhto}.
Furthermore, the normalization condition above simplifies to 
$
\sum_{\orb\in\Orb}(\#H_\orb)^{-1}\sum_{g\in G}\tr{\Lambda_\orb\big(U(g)^*\rho U(g)\big)}\equiv1
$
where $\#S$ is the number of elements in a set $S$.

It is important to note that, if $\Mo$ is an $(\Set,U)$--covariant POVM (i.e.\ satisfies \eqref{kovarianssiehto})
and $V$ is a projective representation of the same group $G$ in any Hilbert space $\ki$ there exists an
$(\Set,U,V)$--covariant $\Mo$--instrument. Namely, for any $\orb\in\Orb$, choose a state $\sigma'_\orb$ of $\ki$ and define the $H_\orb$--invariant state $\sigma_\orb:=(\# H_\orb)^{-1}\sum_{h\in H_\orb}V(h)\sigma'_\orb V(h)^*$ and an instrument 
$$
\mc I^{\rm nuc}_x(\rho):=\tr{\rho\Mo_x}
V(g_x)\sigma_\orb V(g_x)^*
$$
for all $x\in\orb\in\Orb$. Note that, in the context of Theorem \ref{theor:CovInstrBasic}, $\Lambda_\orb(\rho)=\tr{\rho\Mo_{x_\orb}}\sigma_\orb$.
Operationally, in the measurement of $\Mo$ with $\mc I^{\rm nuc}$, if $x$ is obtained (with the probability $\tr{\rho\Mo_x}$) then the output state is $\sigma_x=V(g_x)\sigma_\orb V(g_x)^*$ which does not depend on the input state $\rho$.
Such an instrument is called measure-and-prepare or nuclear \cite{CyHe77}.

Typically there are infinitely many covariant observables so we can ask which are the optimal POVMs  $\Mo$ satisfying the condition \eqref{kovarianssiehto}.
In our previous paper \cite{OptObs}, we studied six optimality criteria of observables. First, we showed that two of these  properties are equivalent with the POVM $\Mo$ being of rank 1 (i.e.\ $\Mo_x=\kb{d_x}{d_x}$ or $\Mo_x=0$ for any $x$): the observable $\Mo$ determines the future of the system (i.e.\ 
any $\Mo$--instrument is nuclear)
 and $\Mo$ is free from the classical noise caused by post-processing of the measurement data. There is also another source of classical noise, namely, the mixing of POVMs. Extreme observables cannot be presented as convex mixtures of observables (`coin tossing between measurements') and, thus, they are free from this type of noise.\footnote{In this paper, we also discuss extremal elements of the smaller set of covariant POVMs or instruments. Such extreme points always exist but they are not necessarily extreme within the entire set of POVMs or instruments. } Projection valued POVMs (PVMs) are automatically extreme and they are also free from quantum noise of pre-processing (i.e.\ one cannot irreversibly send the input state through a channel and then measure another observable to get the same probabilities). The third important property of PVMs is that they determine their values with probabilistic certainty (i.e.\ for any $x$ one finds a state $\rho$ such that $\tr{\rho\Mo_x}=1$ if $\Mo_x\ne0$). The final optimality criterion we studied was the ability of the POVM to completely determine the initial state or the `past' of the system (i.e.\ $\tr{\rho\Mo_x}\equiv\tr{\rho'\Mo_x}$ implies $\rho'=\rho$). Such observables are called informationally complete (IC). See also \cite{BuKeDAPeWe2005,BuLa89} for earlier studies of these optimality properties.
Thus, we  essentially ended up to two mutually complementary classes of optimal POVMs: 
\begin{itemize}
\item[(a)] projection valued rank-1 POVMs and 
\item[(b)] informationally complete extreme (rank-1) POVMs. 
\end{itemize}
We emphasise that a covariance system characterised by  \eqref{kovarianssiehto} might not allow 
rank-1, extreme, PVM, or IC solutions. In the worst case, none such optimal solutions exist (e.g.\ a system with only a trivial solution, an example of which was given just after Theorem \ref{theor:CovObsBasic}).

In the $D$-dimensional Hilbert space $\hil$, any IC extreme observable (is rank-1 and) has exactly $D^2$ non-zero effects $\Mo_x=|d_x\>\<d_x|$ which form a linearly independent set. Similarly, any rank-1 PVM $\Mo_x=|d_x\>\<d_x|$ has $D$ (linearly independent) non-zero projections which form the usual `basis measurement.' Indeed, now $\<d_x|d_y\>=\delta_{xy}$ for non-zero vectors $d_x$ and $d_y$. Since our optimality classes (a) and (b) are clearly disjoint (i.e.\ the determination of the values and the past are complementary properties) we cannot force any observable to be optimal in all six ways above. What one can do is to assume that some optimality criteria hold only approximately and there are `continuous' transformation from one class to the other class of properties. We will exhibit examples of this kind of transformations which also preserve covariance.

The common criterion in both optimality classes (a) and (b) is the rank-1 property which we assume from now on. Clearly, a covariant POVM $\Mo$ is of rank 1 if and only if, for any orbit $\orb\in\Orb$, its `seed'  is of the form $K_\orb=\kb{d_\orb}{d_\orb}$ where $d_\orb$ is a common eigenvector\footnote{Since $U(h)\kb{d_\orb}{d_\orb}=\kb{d_\orb}{d_\orb}U(h)$ exactly when $U(h)d_\orb=c d_\orb$, $c\in\T:=\{c\in\C\,|\,|c|=1\}$. If $H_\orb\ni h\mapsto U(h)$ is irreducible then $d_\orb=0$ as otherwise $\C d_\orb$ would be a nontrivial invariant subspace.} for all unitary operators $U(h)$, $h\in H_\orb$, or $d_\orb=0$. Hence, we may choose $d_x=U(g_x)d_\orb$, $x\in\orb$. If $d_\orb=0$ then all operators $M_x=\kb{d_x}{d_x}$ vanish in the orbit $\orb$ so that the outcomes of that orbit are never registered in any measurement of $\Mo$. In this case, one can redefine $\Set$ to be the union of all orbits where $\Mo$ is not zero.

If $\Mo$ belongs to class (a) (i.e.\ is a PVM) then it has exactly $D$ non-zero (mutually orthogonal) unit vectors $d_x$. For example, if there is only one orbit $\orb=\Set$ and $H_\Set=\{e\}$ then both $G$ and $\Set$ has exactly $D$ elements (i.e.\ any $x=g_x x_\orb$ where $g_x$ is unique) we may take any orthonormal basis $\{d_x\}_{x\in\Set}$ of a $D$-dimensional Hilbert space and define a unitary representation $U(g):=\sum_{x\in\Set}|d_{gx}\>\<d_x|$ to get a covariant rank-1 PVM $\Mo_x:=\kb{d_x}{d_x}$. In this case, we see that \eqref{kovarianssiehto} cannot have a extreme IC solution (since we would need $D^2$ non-zero effects). However, one can extend the covariance structure in such a way that it may also admit an extreme IC solution: We extend the group action $G\times\Set\ni(g,x)\mapsto gx\in\Set$ to the Cartesian product $\Set^2:=\Set\times\Set$ into $G\times\Set^2\ni\big(g,(x,y)\big)\mapsto g(x,y):=(gx,gy)\in\Set^2$ and interpret any covariant POVM $\Mo_x$ as a covariant POVM $\overline\Mo_{(x,y)}:=\delta_{xy}\Mo_x$ with the value space $\Set^2$ of $D^2$ elements. Note that $U$ remains the same. Clearly, $\overline\Mo$ is supported on the diagonal $\{(x,x)\,|\,x\in\Set\}\cong\Set$ and it can be seen as a (trivial) joint measurement of $\Mo$ with itself.\footnote{A POVM $(\sfg_{(x,y)})$ is a joint observable for POVMs $(\sfa_x)$ and $(\sfb_y)$ if $\sum_y \sfg_{(x,y)}=\sfa_x$ and $\sum_x \sfg_{(x,y)}=\sfb_y$.}
A question is whether there is a covariant extreme IC solution for this enlarged system. 

Next we will see that any covariant rank-1 POVM is a projection (and postprocessing) of the above type covariant rank-1 PVM. Indeed, let $\Mo_x=U(g_x)\kb{d_\orb}{d_\orb}U(g_x)^*$, $x\in\orb\in\Orb$, be a covariant rank-1 POVM which need not be normalized since we can normalize it later (see Remark \ref{normalization}). To obtain rank-1 PVM as above we can take the following steps:
\begin{enumerate}
\item
Define a new (finite) outcome space $\Set':=\Orb\times G$ and a POVM
$$
\Mo'_{\orb,g}:=\frac1{\#H_\orb}U(g)\kb{d_\orb}{d_\orb}U(g)^*,\qquad\orb\in\Orb,\quad g\in G,
$$
Clearly, $\Mo'_{\orb,g_x}=\Mo'_{\orb,g_xh}=\Mo_x/\#H_\orb$, $x\in\orb$, $h\in H_\orb$, so that if $\Mo$ is normalized then $\Mo'$ is also normalized to $\id$ and $\Mo$ is a post-processing of $\Mo'$,
$$
\Mo_x=\sum_{h\in H_\orb} \Mo'_{\orb,g_xh},\qquad x\in\orb,
$$
that is, any measurement of $\Mo'$ can be viewed as a measurement of $\Mo$. Note that $\Mo'$ is also covariant when $\Set'$ is equipped with the $G$-action
$g(\orb,g'):=(\orb,gg')$ and the orbits are $\{\orb\}\times G$, $\orb\in\Orb$.

\item
Consider then a covariant Na\u{\i}mark dilation\footnote{The dilation is minimal if and only if $d_\orb\ne 0$ for all $\orb\in\Orb$ (i.e.\ $\Mo_x\ne 0$ for all $x\in\Set$).} of $\Mo'$: The $(\#\Orb\#G)$--dimensional dilation space is spanned by orthonormal vectors $|\orb,g\>$, $\orb\in\Orb$, $g\in G$. Now
$$
J:=\sum_{\orb\in\Orb}\frac{1}{\sqrt{\#H_\orb}}\sum_{g\in G}\kb{\orb,g}{d_\orb}U(g)^*
$$
and the canonical (rank-1) PVM $\Qo_{\orb,g}:=\kb{\orb,g}{\orb,g}$ are such that
$$
\Mo'_{\orb,g}\equiv J^*\Qo_{\orb,g}J.
$$
Clearly, $\Mo'$ is normalized if and only if $J$ is an isometry (i.e.\ $J^*J=\id$). Thus, any measurement of the normalized POVM $\Mo'$ can be seen as a measurement of $\Qo$ when the states are restricted to the range (subspace) of the Na\u{\i}mark projection $JJ^*$. Note that also $\Qo$ is covariant. Indeed, if $m$ is the Schur multiplier (2-cocycle) of the projective unitary representation $g\mapsto U(g)$ one can define a multiplier (left regular) representation
$$
V(g):=\sum_{\orb\in\Orb}\sum_{g'\in G}\overline{m(g,g')}\kb{\orb,gg'}{\orb,g'}
$$
such that $V(gg')=m(g,g')V(g)V(g')$, $V(g)J=JU(g)$ and 
$
\Qo_{g(\orb,g')}=V(g)\Qo_{\orb,g'}V(g)^*.
$
\item We can extend the group $G$ and assume that the multiplier $m(g,g')\equiv 1$. Indeed, as shown in  Appendix A, one can suppose that there exists a (minimal) positive integer $p\le \#G$ such that $m(g,g')^p=1$ for all $g,\,g'\in G$ and $m(e,e)=1$. Define then the (multiplicative) cyclic group $\<t\>=\{1,t,t^2,\ldots,t^{p-1}\}$ where $t:=\exp(2\pi i/p)$ so that $m(g,g')\in\<t\>$, i.e.\ $m(g,g')=t^{q(g,g')}$ where $q(g,g')\in\{0,1,\ldots,p-1\}$. Now a central extrension group (induced by $m$) is a finite set $G_m:=G\times\<t\>$ equipped with the multiplication $(g,t^k)(g',t^\ell):=\big(gg',\overline{m(g,g')}t^{k+\ell}\big)$. Since $m(g,e)=m(e,g)=m(e,e)=1$ one sees that $(e,1)$ is the identity element of $G_m$ and $(g,t^k)^{-1}=\Big(g^{-1},\overline{m(g,g^{-1})}^{-1}t^{-k}\Big)$. Defining unitary operators $\tilde U(g,t^k):=t^k U(g)$ one gets the unitary representation of $G_m$, i.e.\ $\tilde U\big((g,t^k)(g',t^\ell)\big)=\tilde U(g',t^\ell)\tilde U(g,t^k)$ with the constant cocycle. Futhermore, the action $gx$ extends trivially: $(g,t^k)x:=gx$ and we get
$$
\Mo_{(g,t^k)x}=\Mo_{gx}=U(g)\Mo_x U(g)^*= \tilde U(g,t^k)\Mo_x\tilde U(g,t^k)^*.
$$
Hence, $\Mo$ can be seen as a covariant POVM with respect to the larger group $G_m$. Note that if already $m(g,g')\equiv 1$ one has $p=1$, $\<t\>=\{1\}$ and $G_m\cong G$ via $(g,1)\mapsto g$. To conclude, one can replace $G$ with $G_m$ (and elements $g$ with pairs $(g,t^k)$) everywhere in items (1) and (2) and put $m(g,g')\equiv1$. 

\item If  $m(g,g')\equiv1$ then $V(g)=\sum_{\orb\in\Orb}\sum_{g'\in G}\kb{\orb,gg'}{\orb,g'}$ is just a permutation $\pi(g')=gg'$ acting on the basis vectors $|\orb,g'\>$ for a fixed $\orb$. Thus, one can view $G$ as a subgroup of the symmetric group ${\rm Sym}(G):=\{\pi:\,G\to G\,|\,\pi \text{ is bijective}\}$. Especially, $V$ extends to the unitary representation $\overline V(\pi):=\sum_{\orb\in\Orb}\sum_{g'\in G}\kb{\orb,\pi(g')}{\orb,g'}$, $\pi\in{\rm Sym}(G)$, which is a direct sum of the representations $\pi\mapsto \sum_{g'\in G}\kb{\orb,\pi(g')}{\orb,g'}$.
Note that the PVM $\Qo_{\orb,g}=\kb{\orb,g}{\orb,g}$ of item (2) is also covariant with respect to the larger group ${\rm Sym}(G)$:
$
\Qo_{\orb,\pi(g)}=V(\pi)\Qo_{\orb,g}V(\pi)^*.
$
Finally, we can simply number the elements of $G$, $G=\{g_1,g_2,\ldots,g_{\#G}\}$, and identify $G$ (respectively, ${\rm Sym}(G)$) with $\{1,2,\ldots,\#G\}$ (resp.\ the permutations of the integers in question).

\end{enumerate}

Above, we have a method for constructing optimal POVMs. Namely, one can start from item (4) and go backwards. As we have seen, the basic building block of a covariant POVM is a rank-1 PVM $\Qo^D_n:=\kb{n}{n}$, $n\in\Set_{D}:=\{1,\ldots,D\}$, which is covariant with respect to the symmetric group $S_D={\rm Sym}(\Set_D)$ which act in an $D$--dimensional Hilbert space with an orthonormal basis $\{|1\>,\,|2\>,\ldots,|D\>\}$ via the representation $U(\pi)=\sum_{n=1}^D\kb{\pi(n)}{n}$.\footnote{Note that, in item (4), $D=\#G$ and $\ket n=\ket{\orb,g_n}$.} To get an IC extreme POVM we first enlarge the outcome space $\Set_{D}$ to the Cartesian product $\Set_{D}^{2}=\{(n,m)\,|\,1\le n,\,m\le D\}$ where $S_D$ acts via $\pi(n,m):=\big(\pi(n),\pi(m)\big)$. Identify $\Set_{D}$ with the diagonal of $\Set_{D}^{2}$. In Example \ref{ex:genSym}, we define a continuous family of  covariant rank-1 IC extreme POVMs (with outcome space $\Set_{D}^{2}$) with the end point $\Qo^D$. We want to stress that the connective POVMs are also extreme and thus they are not (classical) convex mixtures. In dimension three ($D=3$) this is an easy exercise which we demonstrate next. 

\begin{example}\label{ex:Sym}
Consider the permutation group $S_3$ of a three element set $\Set_3=\{1,2,3\}$. Its generators are permutations $(12)$ and $(13)$. The other permutations are $e=(1)=(12)(12)$, $(123)=(13)(12)$, $(132)=(12)(13)$, and $(23)=(12)(13)(12)$. By definition, $S_3$ operates on $\{1,2,3\}$ by permuting its elements (e.g.\  $(23)1=1$, $(23)2=3$ ja $(23)3=2$). As before, $S_3$ operates also on the nine element set $\Set_3^2=\{1,2,3\}\times\{1,2,3\}$ [e.g.\ $(23)(1,3):=\big((23)1,(23)3\big)=(1,2)$]. Let the Hilbert space be three dimensional, fix its orthonormal basis $\{\ket1,\ket2,\ket3\}$ and define a unitary representation by $U(\pi)=\sum_{n=1}^3\kb{\pi(n)}{n}$, $\pi\in S_3$, that is,
\begin{eqnarray*}
U(12)=\kb21+\kb12+\kb33, & U(13)=\kb31+\kb22+\kb13, \\
 U(1)=\kb11+\kb22+\kb33, & U(123)=\kb21+\kb32+\kb13, \\ 
 U(132)=\kb31+\kb12+\kb23, & U(23)=\kb11+\kb32+\kb23.
\end{eqnarray*}

\begin{enumerate}
\item We have $\Set_3^2=\orb\uplus\orb'$ where the orbits are $\orb=\{(1,1),(2,2),(3,3)\}\cong\Set_3$ and $\orb'=\{(1,2),(2,1),(1,3),(3,1),(2,3),(3,2)\}$ from where we pick points $x_\orb=(1,1)$ ja $x_{\orb'}=(1,2)$.
\item Stability subgroups are $H_{\orb}=\{(1),(23)\}$ and $H_{\orb'}=\{(1)\}$.
\item Since $H_{\orb'}$ is trivial, its seed $K_{\orb'}$ can be an arbitrary positive operator. 
On the other hand, the seed $K_{\orb}\ge0$ must commute with 
$$
U(23)=\kb11+\kb32+\kb23=1\cdot(\kb11+\kb{\fii^{23}_+}{\fii^{23}_+})-1\cdot\kb{\fii^{23}_-}{\fii^{23}_-},$$ 
where the eigenvectors are of the form
$\fii^{ij}_\pm:=2^{-1/2}(\ket i\pm\ket j)$, $i,j\in\{1,2,3\}$, so that 
$$
K_{\orb}=
a\kb11+b\kb{1}{\fii^{23}_+}+\ov b\kb{\fii^{23}_+}{1}+c\kb{\fii^{23}_+}{\fii^{23}_+}
+d\kb{\fii^{23}_-}{\fii^{23}_-},
$$
where the complex numbers satisfy the following conditions: $a,\,c,\,d\ge 0$ ja $ac\ge|b|^2$.
\item Choose $g_{(1,1)}=(1)$, $g_{(2,2)}=(12)$ and $g_{(3,3)}=(13)$ for $\orb$ and
  $g_{(1,2)}=(1)$,
$g_{(2,1)}=(12)$,
$g_{(1,3)}=(23)$,
$g_{(3,1)}=(132)$,
$g_{(2,3)}=(123)$, and
$g_{(3,2)}=(13)$ for $\orb'$.
\item Finally, we normalize the following covariant POVM (where $a,\,c,\,d\ge 0$ and $ac\ge|b|^2$)
\begin{eqnarray*}
\Mo_{(1,1)}&=&K_{\orb}=a\kb11+
b\kb{1}{\fii^{23}_+}+\ov b\kb{\fii^{23}_+}{1}+c\kb{\fii^{23}_+}{\fii^{23}_+}
+d\kb{\fii^{23}_-}{\fii^{23}_-}, \\ 
\Mo_{(2,2)}&=&U(12) K_{\orb} U(12)^*=a\kb22+b\kb{2}{\fii^{13}_+}+\ov b\kb{\fii^{13}_+}{2}+c\kb{\fii^{13}_+}{\fii^{13}_+}
+d\kb{\fii^{13}_-}{\fii^{13}_-},\\
\Mo_{(3,3)}&=&U(13) K_{\orb}U(13)^*=a\kb33+b\kb{3}{\fii^{21}_+}+\ov b\kb{\fii^{21}_+}{3}+c\kb{\fii^{21}_+}{\fii^{21}_+}
+d\kb{\fii^{21}_-}{\fii^{21}_-}, \\
\Mo_{(1,2)}&=&K_{\orb'}\ge0 \\
\Mo_{(2,1)}&=&U(12) K_{\orb'} U(12)^*, \\
\Mo_{(1,3)}&=&U(23) K_{\orb'} U(23)^*, \\
\Mo_{(3,1)}&=&U(132) K_{\orb'} U(132)^*, \\
\Mo_{(2,3)}&=&U(123) K_{\orb'} U(123)^*, \\
\Mo_{(3,2)}&=&U(13) K_{\orb'} U(13)^*.
\end{eqnarray*}
If the operators $\Mo_{(n,m)}$ are linearly independent (resp.\ rank-1) then the normalized operators $K^{-1/2}\Mo_{(n,m)}K^{-1/2}$, $K=\sum_{n,m=1}^3\Mo_{(n,m)}$, are also linearly independent (resp.\ rank-1).
\end{enumerate}
Note that the matrices of the first three operators are
\begin{eqnarray*}
\Mo_{(1,1)}&=&
\begin{pmatrix}
a & b' & b' \\
\ov{b'} & c' & c' \\
\ov{b'} & c' & c' 
\end{pmatrix}
+d'
\begin{pmatrix}
0 & 0 & 0 \\
0 & 1 & -1 \\
0 & -1 & 1
\end{pmatrix}, \\ 
\Mo_{(2,2)}&=&
\begin{pmatrix}
c' & \ov{b'} & c'\\
b'  & a & b' \\
c' & \ov{b'} & c' 
\end{pmatrix}
+d'
\begin{pmatrix}
1 & 0 & -1 \\
0 & 0 & 0 \\
-1 & 0 & 1
\end{pmatrix}, \\ 
\Mo_{(3,3)}&=&
\begin{pmatrix}
c' & c' & \ov{b'} \\
c' & c' &  \ov{b'} \\
b' & b' & a 
\end{pmatrix}
+d'
\begin{pmatrix}
1 & -1 & 0 \\
-1 & 1 & 0 \\
0 & 0 & 0
\end{pmatrix},
\end{eqnarray*}
where $b'=2^{-1/2}b$, $c'=c/2$, and $d'=d/2$ (now $ac'\ge|b'|^2$).
\begin{itemize}
\item
$\Mo$ is rank-1 iff $K_{\orb}=\kb{d_\orb}{d_\orb}$ and $K_{\orb'}=\kb{d_{\orb'}}{d_{\orb'}}$. Now $K_{\orb}=\kb{d_\orb}{d_\orb}\ne 0$ iff either $ac=|b|^2\ne 0$ (i.e.\ $ac'=|b'|^2$) and $d=0$, or  $a=b=c=0$ and $d>0$.
\item $\Mo$ is a rank-1 PVM\footnote{Note that we can always choose the basis such that a rank-1 PVM is the corresponding diagonal `basis measurement.'} if  $a=1$ ja $b=c=d=0$ and $K_{\orb'}=0$ (i.e.\ $\Mo_{(n,m)}=\delta_{nm}\kb n n$).
\item A  rank-1 $\Mo$ is IC extreme (after normalization) iff the nine effects $\Mo_{(n,m)}$ are linearly independent. By direct calculation, this happens if we choose $K_\orb=\kb11$
and $K_{\orb'}=\kb{d_{\orb'}}{d_{\orb'}}$ where $d_{\orb'}=\alpha\big(e^{-i\pi/8}\ket1+e^{i\pi/8}\ket2\big)$, $\alpha>0$.  
For the properly normalized POVM, see Example \ref{ex:genSym}.
\end{itemize}
To conclude, we have a continuous ($\alpha$--indexed) family of covariant rank-1 IC extreme POVMs  whose ($\alpha=0$) end point is a covariant rank-1 PVM. The POVMs with $\alpha>0$ and $\alpha=0$ represent the two complementary optimality classes.
It is interesting to see that in the case $\alpha\approx0$ we get an IC POVM which is `almost' a PVM. \hfill $\triangle$ \end{example}

Using similar methods as above, we may extend an $(\Set,U,V)$--covariant instrument into an instrument whose values are described by $\Orb$ and $G$ and whose symmetries are simply described by permutations of the elements of $G$. Let $m_U$ (resp.\ $m_V$) be the multiplier associated with $U$ (resp.\ with $V$). In particular, through a similar group extension method, picking a (minimal) positive integer $p\leq\# G$ such that $m_U(g,h)^p=1=m_V(g,h)^p$ for all $g,\,h\in G$, we may essentially assume that $U$ and $V$ are ordinary unitary representations, i.e.,\ $m_U(g,h)=1=m_V(g,h)$ for all $g,\,h\in G$.

In what follows, we elaborate the description of $(\Set,U,V)$--covariant instruments and consider the convex set  of $(\Set,U,V)$--covariant instruments and its extreme points. In particular, we see that covariant instruments can be described by pointwise Kraus-operators given by a set of single-point Kraus operators of very particular form which we call $(\Set,U,V)$--intertwiners. After this, we consider the consequences of these results for covariant POVMs and channels. Motivated by the importance of the symmetric group, we give generalizations of the results of Example \ref{ex:Sym} for general symmetric groups and corresponding covariant POVMs in Example \ref{ex:genSym}. We will see that, in general we can determine a family of observables covariant w.r.t.\  the symmetric group in any finite-dimensional system where the disjoint optimality classes (a) and (b) are both represented and that representatives from both classes can be chosen arbitrarily close one another. After this, we generalize many of these results for measurements with continuous value spaces and possibly infinite-dimensional input and output systems.

\section{Instruments covariant with respect to a finite group}\label{sec:finInstr}

We fix Hilbert spaces $\hil$ (input system) and $\mc K$ (output system) and a finite set $\Set$ (measurement outcomes). We denote by $\mc L(\hil)$ the set of (bounded) linear operators on $\hil$ and by $\mc U(\hil)$ the group of unitary operators on $\hil$. We use the same notations for the output system Hilbert space $\mc K$ and, moreover, denote by $\mc L(\hil,\mc K)$ the set of (bounded) linear operators defined on $\hil$ and taking values in $\mc K$. As in this and a couple of the following sections we concentrate on finite dimensional systems, we can disregard the notion of boundedness for now. We assume $\Set$ to be a $G$-space for a finite group $G$, and retain the related notations fixed earlier. Let us fix an orbit $\orb\in\Orb$. We denote by $\hat{H}_{\orb}$ the representation dual of $H_{\orb}$, i.e.,\ the set of unitary equivalence classes of irreducible unitary representations of $H_{\orb}$. We pick a representative for every element of $\hat{H}_{\orb}$ and we denote these representatives typically by $\eta:H_{\orb}\to\mc U(\mc K_\eta)$ and the corresponding equivalence class we denote by $[\eta]$. This convention should cause no confusion. We denote, for any $[\eta]\in\hat{H}_{\orb}$, the dimension of $\mc K_\eta$ by $D_\eta\in\N:=\{1,2,3,\ldots\}$ and fix an orthonormal basis $\{e_{\eta,i}\}_{i=1}^{D_\eta}$ for $\mc K_\eta$. We denote, for any $[\eta]\in\hat{H}_{\orb}$,
$$
\eta_{i,j}(h):=\<e_{\eta,i}|\eta(h)e_{\eta,j}\>,\qquad i,\,j=1,\ldots,D_\eta,\quad h\in H_{\orb}.
$$
As we identify $\orb$ with $G/H_{\orb}$, we pick a section $s_{\orb}:\orb\to G$ (i.e.,\ $s_{\orb}(x)H_{\orb}$ corresponds to $x$ for any $x\in\orb$) such that $s_{\orb}(x_{\orb})=e$.\footnote{Note that we have used the notation $g_x$ for $s_\orb(x)$ for all $x\in\orb$ in Introduction, but this notation would be slightly cumbersome in the following discussion. Also recall that we have fixed a reference point $x_\orb\cong H_\orb=G_{x_\orb}$ for any orbit $\orb\cong G/H_\orb$.} Using these, we define, for all $[\eta]\in\hat{H}_{\orb}$, the cocycles $\zeta^\eta:G\times\orb\to\mc U(\mc K_\eta)$ through
$$
\zeta^\eta(g,x)=\eta\big(s_{\orb}(x)^{-1}g^{-1}s_{\orb}(gx)\big),\qquad g\in G,\quad x\in\orb,
$$
and define the cocycle $\zeta^\pi:G\times\orb\to\mc U(\hil_\pi)$ in exactly the same way whenever $\pi:H_{\orb}\to\mc U(\hil_\pi)$ is a unitary representation in some Hilbert space $\hil_\pi$. Note that the cocycle conditions
\begin{equation}\label{eq:cocycle}
\zeta^\pi(gh,x)=\zeta^\pi(h,x)\zeta^\pi(g,hx),\qquad\zeta^\pi(e,x)=\id_{\hil_\pi}
\end{equation}
hold for any $g,\,h\in G$ and $x\in\orb$. In addition, for any $h\in H_{\orb}$, $\zeta^\pi(h^{-1},x_{\orb})=\pi(h)$.
Finally, we denote by $\zeta^\eta_{i,j}:G\times\orb\to\C$ the matrix element functions of $\zeta^\eta$ in the basis $\{e_{\eta,i}\}_{i=1}^{D_\eta}$ 
for any $[\eta]\in\hat{H}_{\orb}$.

We say that a quadruple $(\mc M,\msf P,\overline{U},J)$ consisting of a Hilbert space $\mc M$, a PVM $\msf P=(\msf P_x)_{x\in\Set}$ in $\mc M$, a unitary representation $\overline{U}:G\to\mc U(\mc M)$, and an isometry $J:\hil\to\mc K\otimes\mc M$ is an {\it $(\Set,U,V)$--covariant minimal Stinespring dilation} for an $(\Set,U,V)$--covariant instrument $\mc I=(\mc I_x)_{x\in\Set}$ if
\begin{itemize}
\item[(i)] $\mc I_x^*(B)=J^*(B\otimes\msf P_x)J$ for all $x\in\Set$ and $B\in\mc L(\mc K)$, where $\mc I_x^*$ is the Heisenberg dual operation\footnote{That is, $\tr{\mc I_x^*(B)\rho}:=\tr{B\mc I_x(\rho)}$ for any $B\in\mc L(\mc K)$ and  input state $\rho$.} for $\mc I_x$,
\item[(ii)] $JU(g)=\big(V(g)\otimes\overline{U}(g)\big)J$ for all $g\in G$,
\item[(iii)] $\overline{U}(g)\msf P_x\overline{U}(g)^*=\msf P_{gx}$ for all $g\in G$ and $x\in\Set$, and
\item[(iv)] vectors $(B\otimes\msf P_x)J\fii$, $B\in\mc L(\mc K)$, $x\in\Set$, $\fii\in\hil$, span $\mc K\otimes\mc M$.
\end{itemize}
Recall that any instrument $\mc I$ has a [minimal] Stinespring dilation $(\mc M,\msf P,J)$ satisfying item (i) [and item (iv)] above. We construct the representation $\overline{U}$ satisfying items (ii) and (iii) for any covariant instrument in Appendix B for completeness. To elaborate Theorem \ref{theor:CovInstrBasic}, we present a useful definition. From now on, the paradoxical notation $m=1,\ldots,0$ means that the set of indices $m$ is empty, and sums of the form $\sum_{m=1}^0(\cdots)$ vanish.

\begin{definition}
Given, for any $\orb\in\Orb$ and $[\eta]\in\hat{H}_{\orb}$, a number $M_\eta\in\{0\}\cup\N$, we say that operators $L_{\eta,i,m}^{\orb}\in\mc L(\hil,\mc K)$ constitute a {\it minimal set of $(\Set,U,V)$--intertwiners} if, for any orbit $\orb\in\Orb$, the set
$$
\{L_{\eta,i,m}^{\orb}\,|\,m=1,\ldots,M_\eta,\ i=1,\ldots,D_\eta,\ [\eta]\in\hat{H}_{\orb}\}
$$
is linearly independent and, for all orbits $\orb$, $[\eta]\in\hat{H}_{\orb}$, $i=1,\ldots,D_\eta$, $m=1,\ldots,M_\eta$, and $h\in H_{\orb}$,
\begin{equation}\label{eq:Hinv}
L_{\eta,i,m}^{\orb}U(h)=\sum_{j=1}^{D_\eta}\eta_{i,j}(h)V(h)L_{\eta,j,m}^{\orb},
\end{equation}
and
\begin{equation}\label{eq:normitus}
\sum_{\orb\in\Orb}\sum_{g\in G}\sum_{[\eta]\in\hat{H}_{\orb}}\sum_{i=1}^{D_\eta}\sum_{m=1}^{M_\eta}\frac{1}{\# H_{\orb}}U(g)L_{\eta,i,m}^{\orb\,*}L_{\eta,i,m}^{\orb}U(g)^*=\id_\hil.
\end{equation}
If the initial linear independence condition is not satisfied, we say that the set $\{L_{\eta,i,m}^{\orb}\,|\,m=1,\ldots,M_\eta,\ i=1,\ldots,D_\eta,\ [\eta]\in\hat{H}_{\orb},\ \orb\in\Orb\}$ is a set of {\it $(\Set,U,V)$--intertwiners}.
\end{definition}

Note that, whenever $M_{\eta'}=0$ for some $[\eta']\in\hat{H}_{\orb}$, the set of intertwiners $L_{\eta,i,m}^{\orb}$ does not contain operators  where $\eta'$ appears as an index. The following theorem exhaustively determines the $(\Set,U,V)$--covariant instruments. It also gives a recipe for constructing covariant instruments and indicates that covariant instruments have the structure conjectured in \cite{Holevo98}. Later, in Theorem \ref{theor:genCovInstrStr} we see that the same structure can be found in covariant instruments even in quite general continuous cases (which is, in fact, the setting Ref.\ \cite{Holevo98} concentrates on).

\begin{theorem}\label{theor:CovInstrStructure}
For any $(\Set,U,V)$--covariant instrument $\mc I=(\mc I_x)_{x\in\Set}$, there is a minimal set
$$
\{L_{\eta,i,m}^{\orb}\,|\,m=1,\ldots,M_\eta,\ i=1,\ldots,D_\eta,\ [\eta]\in\hat{H}_{\orb},\ \orb\in\Orb\}
$$
of $(\Set,U,V)$--intertwiners, where $M_\eta\in\N\cup\{0\}$ for all $[\eta]\in\hat{H}_{\orb}$ and $\orb\in\Orb$, such that, for all $\orb\in\Orb$, $g\in G$, and input states $\rho$ on $\hil$,
\begin{equation}\label{eq:instrkaava}
\mc I_{gH_{\orb}}(\rho)=\sum_{[\eta]\in\hat{H}_{\orb}}\sum_{i=1}^{D_\eta}\sum_{m=1}^{M_\eta}V(g)L_{\eta,i,m}^{\orb}U(g)^*\rho U(g)L_{\eta,i,m}^{\orb\,*}V(g)^*.
\end{equation}
On the other hand, whenever $\{L_{\eta,i,m}^{\orb}\,|\,m=1,\ldots,M_\eta,\ i=1,\ldots,D_\eta,\ [\eta]\in\hat{H}_{\orb},\ \orb\in\Orb\}$, where $M_\eta\in\N\cup\{0\}$ for any $[\eta]\in\hat{H}_{\orb}$ and $\orb\in\Orb$, is a set of $(\Set,U,V)$--intertwiners, Equation \eqref{eq:instrkaava} determines an $(\Set,U,V)$--covariant instrument $\mc I=(\mc I_x)_{x\in\Set}$.
\end{theorem}

Note that, for the instrument $\mc I$ of Equation \eqref{eq:instrkaava}, and for any orbit $\orb\in\Orb$, the map $\Lambda_\orb$ of Theorem \ref{theor:CovInstrBasic} is given by $\Lambda_\orb(\rho)=\sum_{[\eta]\in\hat{H}_\orb}\sum_{i=1}^{D_\eta}\sum_{m=1}^{M_\eta}L^\orb_{\eta,i,m}\rho L^{\orb\,*}_{\eta,i,m}$ for any input state $\rho$.

\begin{proof}
Let us first fix an $(\Set,U,V)$--covariant instrument $\mc I=(\mc I_x)_{x\in\Set}$ and equip it with a minimal $(\Set,U,V)$--covariant Stinespring's dilation $(\mc M,\msf P,\overline{U},J)$ so that $(\overline{U},\msf P)$ is a system of imprimitivity. As in Appendix B, we represent this system of imprimitivity as a direct sum of the canonical systems $(\overline{U}^{\orb},\msf P^{\orb})$ of imprimitivity defined in Equations \eqref{eq:transitiveU} and \eqref{eq:transitiveP}.

Let us fix an orbit $\orb\in\Orb$. According to the Peter-Weyl theorem, for each $[\eta]\in\hat{H}_{\orb}$, there is a Hilbert space $\mc M_\eta$ such that $\mc M^{\orb}=\bigoplus_{[\eta]\in\hat{H}_{\orb}}\mc K_\eta\otimes\mc M_\eta$ and $\pi^{\orb}(g)=\bigoplus_{[\eta]\in\hat{H}_{\orb}} \eta(g)\otimes\id_{\mc M_\eta}$ for all $g\in G$. Denote the dimension of $\mc M_\eta$ by $M_\eta$ and pick an orthonormal basis $\{f_{\eta,m}\}_{m=1}^{M_\eta}\subset\mc M_\eta$. Let $\{\delta_x\}_{x\in\orb}$ be the natural basis of $\C^{\# \orb}$. Thus, $\{\delta_x\otimes e_{\eta,i}\otimes f_{\eta,m}\,|\,x\in\orb,\ [\eta]\in\hat{H}_{\orb},\ i=1,\ldots,D_\eta,\ m=1,\ldots,M_\eta\}$ is an orthonormal basis of $\mc M^{\orb}$ and the union of these bases over $\orb$ is an orthonormal basis for $\mc M$. Define, for $x\in\orb$, $[\eta]\in\hat{H}_{\orb}$, $i=1,\ldots,D_\eta$, and $m=1,\ldots,M_\eta$, the isometry $V_{x,\eta,i,m}:\mc K\to\mc K\otimes\mc M^{\orb}\subseteq\mc K\otimes\mc M$ through $V_{x,\eta,i,m}\psi=\psi\otimes\delta_x\otimes e_{\eta,i}\otimes f_{\eta,m}$ for all $\psi\in\mc K$. Clearly, $V_{x,\eta,i,m}BV_{x,\eta,i,m}^*=B\otimes|\delta_x\otimes e_{\eta,i}\otimes f_{\eta,m}\>\<\delta_x\otimes e_{\eta,i}\otimes f_{\eta,m}|$ for all $B\in\mc L(\mc K)$. Denoting $K_{x,\eta,i,m}:=V_{x,\eta,i,m}^*J$, we find, for all $x\in\orb$ and $B\in\mc L(\mc K)$,
\begin{align}
\mc I_x^*(B)&=J^*(B\otimes\msf P_x)J=\sum_{[\eta]\in\hat{H}_{\orb}}\sum_{i=1}^{D_\eta}\sum_{m=1}^{M_\eta}J^*(B\otimes |\delta_x\otimes e_{\eta,i}\otimes f_{\eta,m}\>\<\delta_x\otimes e_{\eta,i}\otimes f_{\eta,m}|)J\nonumber\\
&=\sum_{[\eta]\in\hat{H}_{\orb}}\sum_{i=1}^{D_\eta}\sum_{m=1}^{M_\eta}J^*V_{x,\eta,i,m}BV_{x,\eta,i,m}^*J=\sum_{[\eta]\in\hat{H}_{\orb}}\sum_{i=1}^{D_\eta}\sum_{m=1}^{M_\eta}K_{x,\eta,i,m}^*BK_{x,\eta,i,m}.\label{eq:KKraus}
\end{align}

Clearly, $\overline{U}^{\orb}(g)(\delta_x\otimes e_{\eta,i}\otimes f_{\eta,m})=\delta_{gx}\otimes\zeta^\eta(g^{-1},gx)e_{\eta_i}\otimes f_{\eta,m}$ for all $g\in G$, $x\in\orb$, $[\eta]\in\hat{H}_{\orb}$, $i=1,\ldots,D_\eta$, and $m=1,\ldots,M_\eta$. Using this and the intertwining properties of $J$, we find that, for all $\fii\in\hil$, $\psi\in\mc K$, $g\in G$, $x\in\orb$, $[\eta]\in\hat{H}_{\orb}$, $i=1,\ldots,D_\eta$, and $m=1,\ldots,M_\eta$,
\begin{align*}
&\<\psi|K_{x,\eta,i,m}U(g)\fii\>=\<V_{x,\eta,i,m}\psi|JU(g)\fii\>=\<V_{x,\eta,i,m}\psi|\big(V(g)\otimes\overline{U}(g)\big)J\fii\>\\
&=\<V(g)^*\psi\otimes\overline{U}(g)^*(\delta_x\otimes e_{\eta,i}\otimes f_{\eta,m})|J\fii\>=\<V(g)^*\psi\otimes\delta_{g^{-1}x}\otimes\zeta^\eta(g,g^{-1}x)e_{\eta,i}\otimes f_{\eta,m}|J\fii\>\\
&=\sum_{j=1}^{D_\eta}\<V(g)^*\psi\otimes\delta_{g^{-1}x}\otimes\zeta^\eta(g,g^{-1}x)e_{\eta,i}\otimes f_{\eta,m}|(\id_{\mc K}\otimes\id_{\C^{\# \orb}}\otimes|e_{\eta,j}\>\<e_{\eta,j}|\otimes\id_{\mc M_\eta})J\fii\>\\
&=\sum_{j=1}^{D_\eta}\overline{\zeta^\eta_{j,i}(g,g^{-1}x)}\<V(g)^*\psi\otimes\delta_{g^{-1}x}\otimes e_{\eta,i}\otimes f_{\eta,m}|J\fii\>=\sum_{j=1}^{D_\eta}\zeta^\eta_{i,j}(g^{-1},x)\<\psi|V(g)K_{g^{-1}x,\eta,j,m}\fii\>,
\end{align*}
where we have used the fact that $\zeta^\eta(g,g^{-1}x)^*=\zeta(g^{-1},x)$ which follows from the cocycle conditions. This means that
\begin{equation}\label{eq:Kinv}
K_{x,\eta,i,m}U(g)=\sum_{j=1}^{D_\eta}\zeta^\eta_{i,j}(g^{-1},x)V(g)K_{g^{-1}x,\eta,j,m}.
\end{equation}

As earlier, let $x_{\orb}$ be a representative for $\orb$ such that $H_{\orb}=G_{x_{\orb}}$, i.e.,\ $x_{\orb}=H_{\orb}$ in the identification $\orb=G/H_{\orb}$. For all $[\eta]\in\hat{H}_{\orb}$, $i=1,\ldots,D_\eta$, and $m=1,\ldots,M_\eta$, define $L^{\orb}_{\eta,i,m}:=K_{x_{\orb},\eta,i,m}$. Recall that, for all $h\in H_{\orb}$ and $[\eta]\in\hat{H}_{\orb}$, $\zeta^\eta(h^{-1},x_{\orb})=\eta(h)$. Using Equation \eqref{eq:Kinv}, we now have for all $[\eta]\in\hat{H}_{\orb}$, $i=1,\ldots,D_\eta$, $m=1,\ldots,M_\eta$, and $h\in H_{\orb}$,
$$
L_{\eta,i,m}^{\orb}U(h)=\sum_{j=1}^{D_\eta}\zeta^\eta_{i,j}(h^{-1},x_{\orb})V(h)K_{h^{-1}x_{\orb},\eta,j,m}=\sum_{j=1}^{D_\eta}\eta_{i,j}(h)V(h)L_{\eta,j,m}^{\orb}.
$$
Thus, we obtain Equation \eqref{eq:Hinv}.

Let us check that the operators $L_{\eta,i,m}^{\orb}$ are linearly independent. To show this, let us first note that vectors $(B\otimes\msf P_{x_{\orb}})J\fii$, $B\in\mc L(\mc K)$, $\fii\in\hil$, span $\mc K\otimes\msf P_{x_{\orb}}\mc M=\mc K\otimes\Big(\bigoplus_{[\eta]\in\hat{H}_{\orb}}\mc K_\eta\otimes\mc M_\eta\Big)$; this follows immediately from the minimality of $(\mc M,\msf P,J)$. Let $\beta_{\eta,i,m}\in\C$, $[\eta]\in\hat{H}_{\orb}$, $i=1,\ldots,D_\eta$, $m=1,\ldots,M_\eta$, and define $v:=\sum_{[\eta]\in\hat{H}_{\orb}}\sum_{i=1}^{D_\eta}\sum_{m=1}^{M_\eta}\beta_{\eta,i,m}e_{\eta,i}\otimes f_{\eta,m}\in\bigoplus_{[\eta]\in\hat{H}_{\orb}}\mc K_\eta\otimes\mc M_\eta$. Let us assume that $\sum_{[\eta]\in\hat{H}_{\orb}}\sum_{i=1}^{D_\eta}\sum_{m=1}^{M_\eta}\beta_{\eta,i,m}L_{\eta,i,m}^{\orb}=0$. Fix a non-zero $\psi_0\in\mc K$ so that, for all $\fii\in\hil$ and $B\in\mc L(\mc K)$,
\begin{align*}
0&=\sum_{[\eta]\in\hat{H}_{\orb}}\sum_{i=1}^{D_\eta}\sum_{m=1}^{M_\eta}\beta_{\eta,i,m}\<B^*\psi_0|L_{\eta,i,m}^{\orb}\fii\>=\sum_{[\eta]\in\hat{H}_{\orb}}\sum_{i=1}^{D_\eta}\sum_{m=1}^{M_\eta}\beta_{\eta,i,m}\<B^*\psi_0\otimes\delta_{x_{\orb}}\otimes e_{\eta,i}\otimes f_{\eta,m}|J\fii\>\\
&=\<B^*\psi_0\otimes\delta_{x_{\orb}}\otimes v|J\fii\>=\<\psi_0\otimes v|(B\otimes\msf P_{x_{\orb}})J\fii\>.
\end{align*}
According to the observation we made before picking the coefficients $\beta_{\eta,i,m}$, this means that $\psi_0\otimes v=0$ and, since $\psi_0\neq0$, we have $v=0$. This is equivalent with the vanishing of the coefficients $\beta_{\eta,i,m}$, proving the linear independence of $\{L_{\eta,i,m}^{\orb}\,|\,[\eta]\in\hat{H}_{\orb},\ i=1,\ldots,D_\eta,\ m=1,\ldots,M_\eta\}$.

Again identifying $\orb=G/H_{\orb}$ and $x_{\orb}=H_{\orb}$, from \eqref{eq:Kinv} we obtain
\begin{align}
K_{gH_{\orb},\eta,i,m}&=\sum_{j=1}^{D_\eta}\zeta^\eta_{i,j}(g^{-1},gH_{\orb})V(g)K_{H_{\orb},\eta,j,m}U(g)^*=\sum_{j=1}^{D_\eta}\zeta^\eta_{i,j}(g^{-1},gH_{\orb})V(g)L_{\eta,j,m}^{\orb}U(g)^*.\label{eq:apu1}
\end{align}
Indeed, it is easy to see directly that the RHS of Equation \eqref{eq:apu1} is invariant in substitutions $g\mapsto gh$ where $h\in H_{\orb}$. Using the Schr\"odinger version of Equation \eqref{eq:KKraus}, Equation \eqref{eq:apu1}, and the easily proven fact that, for any $[\eta]\in\hat{H}_{\orb}$, $g\in G$, and $j,\,k=1,\ldots,D_\eta$, $\sum_{i=1}^{D_\eta}\zeta^\eta_{i,j}(g^{-1},gH_{\orb})\overline{\zeta^\eta_{i,k}(g^{-1},gH_{\orb})}=\delta_{j,k}$, where $\delta_{j,k}$ is the Kronecker symbol (i.e.,\ $\delta_{j,k}=1$ if $j=k$ and, otherwise, $\delta_{j,k}=0$), we find, for all input states $\rho$ and $g\in G$,
\begin{align*}
\mc I_{gH_{\orb}}(\rho)&=\sum_{[\eta]\in\hat{H}_{\orb}}\sum_{i=1}^{D_\eta}\sum_{m=1}^{M_\eta}K_{gH_{\orb},\eta,i,m}\rho K_{gH_{\orb},\eta,i,m}^*\\
&=\sum_{[\eta]\in\hat{H}_{\orb}}\sum_{i,j,k=1}^{D_\eta}\sum_{m=1}^{M_\eta}\zeta^\eta_{i,j}(g^{-1},gH_{\orb})\overline{\zeta^\eta_{i,k}(g^{-1},gH_{\orb})}V(g)L_{\eta,j,m}^{\orb}U(g)^*\rho U(g)L_{\eta,k,m}^{\orb\,*}V(g)\\
&=\sum_{[\eta]\in\hat{H}_{\orb}}\sum_{i=1}^{D_\eta}\sum_{m=1}^{M_\eta}V(g)L_{\eta,i,m}^{\orb}U(g)^*\rho U(g)L_{\eta,i,m}^{\orb\,*}V(g)^*,
\end{align*}
implying Equation \eqref{eq:instrkaava}.

Let us move on to proving Equation \eqref{eq:normitus}. Let us first note that, for any orbit $\orb$, $[\eta]\in\hat{H}_{\orb}$, $m=1,\ldots,M_\eta$, and $h\in H_{\orb}$, we find, using the already established Equation \eqref{eq:Hinv},
\begin{align*}
\sum_{i=1}^{D_\eta}U(h)L_{\eta,i,m}^{\orb\,*}L_{\eta,i,m}^{\orb}U(h)^*&=\sum_{i,j,k=1}^{D_\eta}\overline{\eta_{i,j}(h^{-1})}\eta_{i,k}(h^{-1})L_{\eta,j,m}^{\orb\,*}V(h)V(h)^*L_{\eta,k,m}^{\orb}\\
&=\sum_{j,k=1}^{D_\eta}\<\eta(h)^*e_{\eta,j}| \eta(h)^*e_{\eta,k}\>L_{\eta,j,m}^{\orb\,*}L_{\eta,k,m}^{\orb}=\sum_{i=1}^{D_\eta}L_{\eta,i,m}^{\orb\,*}L_{\eta,i,m}^{\orb}.
\end{align*}
Using the above observation and the dual (Heisenberg) version of the already established Equation \eqref{eq:instrkaava}, we find
\begin{align*}
\id_\hil&=\sum_{x\in\Set}\mc I_x^*(\id_{\mc K})=\sum_{\orb\in\Orb}\sum_{x\in\orb}\mc I_{s^{\orb}(x)H_{\orb}}^*(\id_{\mc K})\\
&=\sum_{\orb\in\Orb}\sum_{x\in\orb}\sum_{[\eta]\in\hat{H}_{\orb}}\sum_{i=1}^{D_\eta}\sum_{m=1}^{M_\eta}U\big(s^{\orb}(x)\big)L_{\eta,i,m}^{\orb\,*}L_{\eta,i,m}^{\orb}U\big(s^{\orb}(x)\big)^*\\
&=\sum_{\orb\in\Orb}\sum_{x\in\orb}\sum_{h\in H_{\orb}}\sum_{[\eta]\in\hat{H}_{\orb}}\sum_{i=1}^{D_\eta}\sum_{m=1}^{M_\eta}\frac{1}{\# H_{\orb}}U\big(s^{\orb}(x)h\big)L_{\eta,i,m}^{\orb\,*}L_{\eta,i,m}^{\orb}U\big(s^{\orb}(x)h\big)^*\\
&=\sum_{\orb\in\Orb}\sum_{g\in G}\sum_{[\eta]\in\hat{H}_{\orb}}\sum_{i=1}^{D_\eta}\sum_{m=1}^{M_\eta}\frac{1}{\# H_{\orb}}U(g)L_{\eta,i,m}^{\orb\,*}L_{\eta,i,m}^{\orb}U(g)^*,
\end{align*}
implying Equation \eqref{eq:normitus}. The final converse claim follows from Theorem \ref{theor:CovInstrBasic} upon noting that the operation $\Lambda_\orb$ defined just after the statement of this theorem with a (minimal) set of $(\Set,U,V)$--intertwiners $L_{\eta,i,m}$ satisfies the conditions of Theorem \ref{theor:CovInstrBasic} by using Equations \eqref{eq:Hinv} and \eqref{eq:normitus}.
\end{proof}

\begin{remark}\label{normalization}
Suppose that, for any orbit $\orb\in\Orb$ and $[\eta]\in\hat{H}_{\orb}$, $M_\eta\in\{0\}\cup\N$ and $L_{\eta,i,m}^{\orb}\in\mc L(\hil,\mc K)$, $i=1,\ldots,D_\eta$, $m=1,\ldots,M_\eta$, are such that Equation \eqref{eq:Hinv} holds but
$$
K:=\sum_{\orb\in\Orb}\sum_{g\in G}\sum_{[\eta]\in\hat{H}_{\orb}}\sum_{i=1}^{D_\eta}\sum_{m=1}^{M_\eta}\frac{1}{\# H_{\orb}}U(g)L_{\eta,i,m}^{\orb\,*}L_{\eta,i,m}^{\orb}U(g)^*
$$
does not necessarily coincide with $\id_\hil$. Since, due to its definition, $K$ commutes with $U$, i.e.,\ $U(g)K=KU(g)$ for all $g\in G$, we may define, for any orbit $\orb$, $[\eta]\in\hat{H}_{\orb}$, $i=1,\ldots,D_\eta$, and $m=1,\ldots,M_\eta$, the new operator $\tilde{L}_{\eta,i,m}^{\orb}:=L_{\eta,i,m}^{\orb}K^{-1/2}$ (where $K^{-1/2}$ is the square root of the generalized inverse of $K$) which still satisfy Equation \eqref{eq:Hinv} (with $L_{\eta,i,m}^{\orb}$ replaced with $\tilde{L}_{\eta,i,m}^{\orb}$) and which now, additionally, satisfy
$$
\sum_{\orb\in\Orb}\sum_{g\in G}\sum_{[\eta]\in\hat{H}_{\mc O}}\sum_{i=1}^{D_\eta}\sum_{m=1}^{M_\eta}\frac{1}{\# H_{\orb}}U(g)\tilde{L}_{\eta,i,m}^{\orb\,*}\tilde{L}_{\eta,i,m}^{\orb}U(g)^*={\rm supp}\,K
$$
where ${\rm supp}\,K$ is the support projection of $K$. Thus we obtain an $(\Set,\tilde{U},V)$--covariant instrument through Equation \eqref{eq:instrkaava} (with $L_{\eta,i,m}^{\orb}$ replaced with $\tilde{L}_{\eta,i,m}^{\orb}$) for a possibly smaller input Hilbert space $({\rm supp}\,K)(\hil)=:\tilde{\hil}$ which is an invariant subspace for $U$ where the restriction of $U$ we denote by $\tilde{U}$. Naturally, if $U$ is irreducible, we have $K\in\C\id_\hil$ so that $\tilde{\hil}=\hil$ or $\tilde{\hil}=\{0\}$; the latter case is possible only in the highly reduced case where $L_{\eta,i,m}^{\orb}$ all vanish (which is hardly interesting). \hfill $\triangle$
\end{remark}

In the proof of Theorem \ref{theor:CovInstrStructure}, we saw that, from a minimal covariant Stinespring dilation of a covariant instrument $\mc I$, we obtain a minimal set of $(\Set,U,V)$--intertwiners defining $\mc I$ through Equation \eqref{eq:instrkaava}. The following lemma gives the converse result: a {\it minimal} set of intertwiners can be used to define a minimal covariant Stinespring dilation for a covariant instrument. This result will be very useful when giving extremality conditions for covariant instruments.

\begin{lemma}\label{lemma:minlemma}
Let $\mc I$ be an $(\Set,U,V)$--covariant instrument defined through Equation \eqref{eq:instrkaava} by a minimal set of $(\Set,U,V)$--intertwiners consisting of $L_{\eta,i,m}^{\orb}\in\mc L(\hil,\mc K)$ for all $\orb\in\Orb$, $[\eta]\in\hat{H}_{\orb}$, $i=1,\ldots,D_\eta$, and $m=1,\ldots,M_\eta$ where $M_\eta\in\{0\}\cup\N$. Defining
\begin{equation}\label{eq:apuapu1}
K_{gH_{\orb},\eta,i,m}:=\sum_{j=1}^{D_\eta}\zeta^\eta_{i,j}(g^{-1},gH_{\orb})V(g)L_{\eta,j,m}^{\orb}U(g)^*
\end{equation}
for all $\orb\in\Orb$, $g\in G$, $[\eta]\in\hat{H}_{\orb}$, $i=1,\ldots,D_\eta$, and $m=1,\ldots,M_\eta$ and setting
$$
\mc M:=\bigoplus_{\orb\in\Orb}\C^{\# \orb}\otimes\Big(\bigoplus_{[\eta]\in\hat{H}_{\orb}}\mc K_\eta\otimes\C^{M_\eta}\Big),
$$
the linear map $J:\hil\to\mc K\otimes\mc M$
$$
J\fii=\sum_{\orb\in\Orb}\sum_{x\in\orb}\sum_{[\eta]\in\hat{H}_{\orb}}\sum_{i=1}^{D_\eta}\sum_{m=1}^{M_\eta}K_{x,\eta,i,m}\fii\otimes\delta_x\otimes e_{\eta,i}\otimes f_{\eta,m},\qquad \fii\in\hil,
$$
where $\{\delta_x\}_{x\in\Set}$ is the natural basis for $\C^{\# \Set}\supseteq\C^{\# \orb}$ and $\{f_{\eta,m}\}_{m=1}^{M_\eta}$ is some orthonormal basis of $\C^{M_\eta}$, the PVM $\msf P=(\msf P_x)_{x\in\Set}$,
$$
\msf P_x=|\delta_x\>\<\delta_x|\otimes\bigg(\bigoplus_{[\eta]\in\hat{H}_{\orb}}\id_{\mc K_\eta}\otimes\id_{\C^{M_\eta}}\bigg),\qquad x\in\orb\in\Orb,
$$
and the unitary representation $\overline{U}:G\to\mc U(\mc M)$ through
$$
\overline{U}(g)(\delta_x\otimes e_{\eta,i}\otimes f_{\eta,m})=\delta_{gx}\otimes\zeta^\eta(g^{-1},gx)e_{\eta,i}\otimes f_{\eta,m}
$$
for all $g\in G$, $x\in\orb\in\Orb$, $[\eta]\in\hat{H}_{\orb}$, $i=1,\ldots,D_\eta$, and $m=1,\ldots,M_\eta$, the quadruple $(\mc M,\msf P,\overline{U},J)$ is a minimal $(\Set,U,V)$--covariant Stinespring dilation for $\mc I$.
\end{lemma}

\begin{proof}
Let us start by proving that $(\mc M,\msf P,J)$ is a minimal Stinespring dilation for $\mc I$. The fact that $\mc I_x^*(B)=J^*(B\otimes\msf P_x)J$ for all $x\in\Set$ and $B\in\mc L(\mc K)$ is proven through a simple direct calculation which we leave for the reader to check. Let us concentrate on the minimality claim. Let us first show that, for any $x\in\orb\in\Orb$, the set $\{K_{x,\eta,i,m}\,|\,[\eta]\in\hat{H}_{\orb},\ i=1,\ldots,D_\eta, m=1,\ldots,M_\eta\}$ is linearly independent. Let us fix an orbit $\orb\in\Orb$, and $g\in G$ and let $\beta_{\eta,i,m}\in\C$, $[\eta]\in\hat{H}_{\orb}$, $i=1,\ldots,D_\eta$, $m=1,\ldots,M_\eta$, be such that $\sum_{[\eta]\in\hat{H}_{\orb}}\sum_{i=1}^{D_\eta}\sum_{m=1}^{M_\eta}\beta_{\eta,i,m}K_{gH_{\orb},\eta,i,m}=0$. Using Equation \eqref{eq:apuapu1}, we obtain
\begin{align*}
0&=\sum_{[\eta]\in\hat{H}_{\orb}}\sum_{i=1}^{D_\eta}\sum_{m=1}^{M_\eta}\beta_{\eta,i,m}K_{gH_{\orb},\eta,i,m}=\sum_{[\eta]\in\hat{H}_{\orb}}\sum_{i,j=1}^{D_\eta}\sum_{m=1}^{M_\eta}\zeta^\eta_{i,j}(g^{-1},gH_{\orb})\beta_{\eta,i,m}V(g)L_{\eta,j,m}^{\orb}U(g)^*\\
&=V(g)\left[\sum_{[\eta]\in\hat{H}_{\orb}}\sum_{j=1}^{D_\eta}\sum_{m=1}^{M_\eta}\left(\sum_{i=1}^{D_\eta}\zeta^\eta_{i,j}(g^{-1},gH_{\orb})\beta_{\eta,i,m}\right)L_{\eta,j,m}^{\orb}\right]U(g)^*=0.
\end{align*}
Since $\{L_{\eta,i,m}^{\orb}\,|\,[\eta]\in\hat{H}_{\orb},\ i=1,\ldots,D_\eta,\ m=1,\ldots,M_\eta\}$ is linearly independent, it immediately follows that, for all $[\eta]\in\hat{H}_{\orb}$, $j=1,\ldots,D_\eta$, and $m=1,\ldots,M_\eta$, $\sum_{i=1}^{D_\eta}\zeta^\eta_{i,j}(g^{-1},gH_{\orb})\beta_{\eta,i,m}=0$. Thus, we obtain $0=\sum_{i,j=1}^{D_\eta}\overline{\zeta^\eta_{k,j}(g^{-1},gH_{\orb})}\zeta^\eta_{i,j}(g^{-1},gH_{\orb})\beta_{\eta,i,m}=\sum_{i=1}^{D_\eta}\delta_{i,k}\beta_{\eta,i,m}=\beta_{\eta,k,m}$ for any $[\eta]\in\hat{H}_{\orb}$, $k=1,\ldots,D_\eta$, and $m=1,\ldots,M_\eta$, proving that $\{K_{gH_{\orb},\eta,i,m}\,|\,[\eta]\in\hat{H}_{\orb},\ i=1,\ldots,D_\eta, m=1,\ldots,M_\eta\}$ is linearly independent.

Let us assume that $\Psi\in\mc K\otimes\mc M$ is such that $\<\Psi|(B\otimes\msf P_x)J\fii\>=0$ for all $B\in\mc L(\mc K)$, $x\in\Set$, and $\fii\in\hil$. For any $x\in\orb\in\Orb$, $[\eta]\in\hat{H}_{\orb}$, $i=1,\ldots,D_\eta$, and $m=1,\ldots,M_\eta$, there is $\psi_{x,\eta,i,m}\in\mc K$ such that $\Psi=\sum_{\orb\in\Orb}\sum_{x\in\orb}\sum_{[\eta]\in\hat{H}_{\orb}}\sum_{i=1}^{D_\eta}\sum_{m=1}^{M_\eta}\psi_{x,\eta,i,m}\otimes\delta_x\otimes e_{\eta,i}\otimes f_{\eta,m}$. Thus, we have, for all $B\in\mc L(\mc K)$, $x\in\orb\in\Orb$, and $\fii\in\hil$, $0=\<\Psi|(B\otimes\msf P_x)J\fii\>=\sum_{[\eta]\in\hat{H}_{\orb}}\sum_{i=1}^{D_\eta}\sum_{m=1}^{M_\eta}\<\psi_{x,\eta,i,m}|BK_{x,\eta,i,m}\fii\>$, implying, upon substituting $B=|\psi\>\<\psi'|$, that, for all $\psi,\,\psi'\in\mc K$, $x\in\orb\in\Orb$, and $\fii\in\hil$, $\sum_{[\eta]\in\hat{H}_{\orb}}\sum_{i=1}^{D_\eta}\sum_{m=1}^{M_\eta}\<\psi_{x,\eta,i,m}|\psi\>\<\psi'|K_{x,\eta,i,m}\fii\>=0$. Since $\{K_{x,\eta,i,m}\,|\,[\eta]\in\hat{H}_{\orb},\ i=1,\ldots,D_\eta, m=1,\ldots,M_\eta\}$ is linearly independent for any $x\in\orb\in\Orb$, this means that, for all $x\in\orb\in\Orb$, $[\eta]\in\hat{H}_{\orb}$, $i=1,\ldots,D_\eta$, $m=1,\ldots,M_\eta$, and $\psi\in\mc K$, $\<\psi_{x,\eta,i,m}|\psi\>=0$. This, of course, means that, for all $x\in\orb\in\Orb$, $[\eta]\in\hat{H}_{\orb}$, $i=1,\ldots,D_\eta$, and $m=1,\ldots,M_\eta$, $\psi_{x,\eta,i,m}=0$, i.e.,\ $\Psi=0$, proving the minimality.

As in the proof of Theorem \ref{theor:CovInstrStructure}, we can show that Equation \eqref{eq:Kinv} holds so that we have, for all $g\in G$ and $\fii\in\hil$
\begin{align*}
JU(g)\fii&=\sum_{\orb\in\Orb}\sum_{x\in\orb}\sum_{[\eta]\in\hat{H}_{\orb}}\sum_{i=1}^{D_\eta}\sum_{m=1}^{M_\eta}K_{x,\eta,i,m}U(g)\fii\otimes\delta_x\otimes e_{\eta,i}\otimes f_{\eta,m}\\
&=\sum_{\orb\in\Orb}\sum_{x\in\orb}\sum_{[\eta]\in\hat{H}_{\orb}}\sum_{i,j=1}^{D_\eta}\sum_{m=1}^{M_\eta}\zeta^\eta_{i,j}(g^{-1},x)V(g)K_{g^{-1}x,\eta,j,m}\fii\otimes\delta_x\otimes e_{\eta,i}\otimes f_{\eta,m}\\
&=\sum_{\orb\in\Orb}\sum_{x\in\orb}\sum_{[\eta]\in\hat{H}_{\orb}}\sum_{i,j=1}^{D_\eta}\sum_{m=1}^{M_\eta}\zeta^\eta_{i,j}(g^{-1},gx)V(g)K_{x,\eta,j,m}\fii\otimes\delta_{gx}\otimes e_{\eta,i}\otimes f_{\eta,m}\\
&=\sum_{\orb\in\Orb}\sum_{x\in\orb}\sum_{[\eta]\in\hat{H}_{\orb}}\sum_{j=1}^{D_\eta}\sum_{m=1}^{M_\eta}V(g)K_{x,\eta,j,m}\fii\otimes\delta_{gx}\otimes\zeta^\eta(g^{-1},gx)e_{\eta,j}\otimes f_{\eta,m}\\
&=\big(V(g)\otimes\overline{U}(g)\big)\sum_{\orb\in\Orb}\sum_{x\in\orb}\sum_{[\eta]\in\hat{H}_{\orb}}\sum_{j=1}^{D_\eta}\sum_{m=1}^{M_\eta}K_{x,\eta,j,m}\fii\otimes\delta_x\otimes e_{\eta,i}\otimes f_{\eta,m}=\big(V(g)\otimes\overline{U}(g)\big)J\fii.
\end{align*}
Proving that $\overline{U}(g)\msf P_x\overline{U}(g)^*=\msf P_{gx}$ for all $g\in G$ and $x\in\Set$ is straightforward.
\end{proof}

Using Theorem \ref{theor:CovInstrStructure} and Lemma \ref{lemma:minlemma}, we can also determine extremality conditions for $(\Set,U,V)$--covariant instruments. We say that an $(\Set,U,V)$--covariant instrument $\mc I=(\mc I_x)_{x\in\Set}$ is an {\it extreme instrument of the $(\Set,U,V)$--covariance structure} if it is an extreme point of the convex set of all $(\Set,U,V)$--covariant instruments.

\begin{theorem}\label{theor:ExtInstr}
Let $\mc I$ be an $(\Set,U,V)$--covariant instrument defined through Equation \eqref{eq:instrkaava} by a minimal set of $(\Set,U,V)$--intertwiners consisting of $L_{\eta,i,m}^{\orb}\in\mc L(\hil,\mc K)$ for all $\orb\in\Orb$, $[\eta]\in\hat{H}_{\orb}$, $i=1,\ldots,D_\eta$, and $m=1,\ldots,M_\eta$ where $M_\eta\in\{0\}\cup\N$. The instrument $\mc I$ is an extreme instrument of the $(\Set,U,V)$--covariance structure if and only if the set
$$
\left\{\sum_{g\in G}\sum_{i=1}^{D_\eta}U(g)L_{\eta,i,m}^{\orb\,*}L_{\eta,i,n}^{\orb}U(g)^*\,\bigg|\,m,\,n=1,\ldots,M_\eta,\ [\eta]\in\hat{H}_{\orb},\ \orb\in\Orb\right\}
$$
is linearly independent.
\end{theorem}

\begin{proof}
Let $(\mc M,\msf P,\overline{U},J)$ be the minimal $(\Set,U,V)$--covariant Stinespring dilation for $\mc I$ as defined in Lemma \ref{lemma:minlemma}. Denote, for brevity, for any orbit $\orb\in\Orb$,
$$
\mc M^{\orb}:=\bigoplus_{[\eta]\in\hat{H}_{\orb}}\mc K_\eta\otimes\C^{M_\eta}.
$$
According to the results of \cite{HaPe15}, $\mc I$ is an extreme observable of the $(\Set,U,V)$--covariance structure if and only if, for $E\in\mc L(\mc M)$ the conditions $E\msf P_x=\msf P_x E$ for all $x\in\Set$, $E\overline{U}(g)=\overline{U}(g)E$ for all $g\in G$, and $J^*(\id_{\mc K}\otimes E)J=0$ imply $E=0$; note that for this extremality characterization it is vital that the dilation is minimal. Let $E\in\mc L(\mc M)$ be such that $E\msf P_x=\msf P_x E$ for all $x\in\Set$ and $E\overline{U}(g)=\overline{U}(g)E$ for all $g\in G$. The first condition is equivalent with the existence of $E_x\in\mc L(\mc M^{\orb})$, $x\in\orb\in\Orb$, such that $E(\delta_x\otimes v)=\delta_x\otimes E_x v$ for all $v\in\mc M^{\orb}$. Denoting, for all $g\in G$ and $x\in\orb\in\Orb$, $\zeta^{\orb}(g,x):=\bigoplus_{[\eta]\in\hat{H}_{\orb}}\zeta^\eta(g,x)\otimes\id_{M_\eta}$, the second condition is easily seen to be equivalent with
\begin{equation}\label{eq:apuzeta}
\zeta^{\orb}(g^{-1},gx)E_x=E_{gx}\zeta^{\orb}(g^{-1},gx),\qquad x\in\orb\in\Orb,\quad g\in G.
\end{equation}
Identifying $\orb=G/H_{\orb}$, we obtain $E_{gH_{\orb}}=\zeta^{\orb}(g^{-1},gH_{\orb})E_{H_\orb}\zeta^{\orb}(g^{-1},gH_{\orb})^*$ for any orbit $\orb$. Note that, defining, for all orbits $\orb$ and $h\in H_{\orb}$, $\zeta^{\orb}(h^{-1},H_{\orb})=:\pi^{\orb}(h)$, we determine a unitary representation $\pi^{\orb}:H_{\orb}\to\mc U(\mc M^{\orb})$ such that
\begin{equation}\label{eq:ODecomp}
\pi^{\orb}(h)=\bigoplus_{[\eta]\in\hat{H}_{\orb}}\eta(h)\otimes\id_{M_\eta}.
\end{equation}
Using Equation \eqref{eq:apuzeta}, we have $\pi^{\orb}(h)E_{H_\orb}=\zeta^{\orb}(h^{-1},H_{\orb})E_{H_{\orb}}=E_{hH_{\orb}}\zeta^{\orb}(h^{-1},H_{\orb})=E_{H_\orb}\pi^{\orb}(h)$ for all $\orb\in\Orb$ and $h\in H_{\orb}$. The decomposition in Equation \eqref{eq:ODecomp} implies now that $E_{H_\orb}=\bigoplus_{[\eta]\in\hat{H}_{\orb}}\id_{\mc K_\eta}\otimes E_\eta$ for all $\orb\in\Orb$ where $E_\eta\in\mc L(\C^{M_\eta})$ for all $[\eta]\in\hat{H}_{\orb}$. We now have
$E_{gH_{\orb}}=\zeta^{\orb}(g^{-1},gH_{\orb})E_{H_\orb}\zeta^{\orb}(g^{-1},gH_{\orb})^*=\bigoplus_{[\eta]\in\hat{H}_{\orb}}\zeta^{\eta}(g^{-1},gH_{\orb})\zeta^\eta(g^{-1},gH_{\orb})^*\otimes E_\eta=\bigoplus_{[\eta]\in\hat{H}_{\orb}}\id_{\mc K_\eta}\otimes E_\eta=E_{H_\orb}$ for any orbit $\orb$ and $g\in G$. Thus,
\begin{equation}\label{eq:DHaj}
E=\sum_{\orb\in\Orb}\sum_{x\in\orb}|\delta_x\>\<\delta_x|\otimes\left(\bigoplus_{[\eta]\in\hat{H}_{\orb}}\id_{\mc K_\eta}\otimes E_\eta\right)
\end{equation}

In the same way as in the proof of Theorem \ref{theor:CovInstrStructure}, we see that, for any orbit $\orb\in\Orb$, $h\in H_{\orb}$, $[\eta]\in\hat{H}_{\orb}$, and $m,\,n=1,\ldots,M_\eta$, $\sum_{i=1}^{D_\eta}U(h)L_{\eta,i,m}^{\orb\,*}L_{\eta,i,n}^{\orb}U(h)^*=\sum_{i=1}^{D_\eta}L_{\eta,i,m}^{\orb\,*}L_{\eta,i,n}^{\orb}$. Recall the section $s_{\orb}:\orb\to G$ such that $s_{\orb}(x_{\orb})=e$. Using the above observation and Equation \eqref{eq:DHaj}, we have, for any $\fii\in\hil$,
\begin{align*}
&\<J\fii|(\id_{\mc K}\otimes E)J\fii\>=\sum_{\orb\in\Orb}\sum_{x\in\orb}\sum_{[\eta]\in\hat{H}_{\orb}}\<J\fii|(\id_{\mc K}\otimes|\delta_x\>\<\delta_x|\otimes\id_{\mc K_\eta}\otimes E_\eta)J\fii\>\\
=&\sum_{\orb\in\Orb}\sum_{x\in\orb}\sum_{[\eta]\in\hat{H}_{\orb}}\sum_{i=1}^{D_\eta}\sum_{m,n=1}^{M_\eta}\<K_{x,\eta,i,m}\fii\otimes f_{\eta,m}|K_{x,\eta,i,n}\fii\otimes E_\eta f_{\eta,n}\>\\
=&\sum_{\orb\in\Orb}\sum_{x\in\orb}\sum_{[\eta]\in\hat{H}_{\orb}}\sum_{i=1}^{D_\eta}\sum_{m,n=1}^{M_\eta}\<f_{\eta,m}|E_\eta f_{\eta,n}\>\<K_{s_{\orb}(x)H_{\orb},\eta,i,m}\fii|K_{s_{\orb}(x)H_{\orb},\eta,i,n}\fii\>\\
=&\sum_{\orb\in\Orb}\sum_{x\in\orb}\sum_{[\eta]\in\hat{H}_{\orb}}\sum_{i,j,k=1}^{D_\eta}\sum_{m,n=1}^{M_\eta}\<f_{\eta,m}|E_\eta f_{\eta,n}\>\overline{\zeta^\eta_{i,j}\big(s_{\orb}(x)^{-1},x\big)}\zeta^\eta_{i,k}\big(s_{\orb}(x)^{-1},x\big)\times\\
&\times\<L_{\eta,j,m}^{\orb}U\big(s_{\orb}(x)\big)^*\fii|L_{\eta,k,n}^{\orb}U\big(s_{\orb}(x)\big)^*\fii\>\\
=&\sum_{\orb\in\Orb}\sum_{x\in\orb}\sum_{[\eta]\in\hat{H}_{\orb}}\sum_{i=1}^{D_\eta}\sum_{m,n=1}^{M_\eta}\<f_{\eta,m}|E_\eta f_{\eta,n}\>\<L_{\eta,i,m}^{\orb}U\big(s_{\orb}(x)\big)^*\fii|L_{\eta,i,n}^{\orb}U\big(s_{\orb}(x)\big)^*\fii\>\\
=&\sum_{\orb\in\Orb}\sum_{x\in\orb}\sum_{h\in H_{\orb}}\sum_{[\eta]\in\hat{H}_{\orb}}\sum_{i=1}^{D_\eta}\sum_{m,n=1}^{M_\eta}\frac{1}{\# H_{\orb}}\<f_{\eta,m}|E_\eta f_{\eta,n}\>\<L_{\eta,i,m}^{\orb}U\big(s_{\orb}(x)h\big)^*\fii|L_{\eta,i,n}^{\orb}U\big(s_{\orb}(x)h\big)^*\fii\>\\
=&\sum_{\orb\in\Orb}\sum_{g\in G}\sum_{[\eta]\in\hat{H}_{\orb}}\sum_{i=1}^{D_\eta}\sum_{m,n=1}^{M_\eta}\frac{1}{\# H_{\orb}}\<f_{\eta,m}|E_\eta f_{\eta,n}\>\<L_{\eta,i,m}^{\orb}U(g)^*\fii|L_{\eta,i,n}^{\orb}U(g)^*\fii\>\\
=&\sum_{\orb\in\Orb}\sum_{[\eta]\in\hat{H}_{\orb}}\sum_{m,n=1}^{M_\eta}\beta_{\eta,m,n}^{\orb}\sum_{g\in G}\sum_{i=1}^{D_\eta}\<L_{\eta,i,m}^{\orb}U(g)^*\fii|L_{\eta,i,n}^{\orb}U(g)^*\fii\>,
\end{align*}
where we have denoted $\beta^{\orb}_{\eta,m,n}:=(\# H_{\orb})^{-1}\<f_{\eta,m}|E_\eta f_{\eta,n}\>$, for all orbits $\orb\in\Orb$, $[\eta]\in\hat{H}_{\orb}$, and $m,\,n=1,\ldots,M_\eta$. From this observation the claim immediately follows.
\end{proof}

Suppose now that $U$ is irreducible. Now for any minimal set of $(\Set,U,V)$--intertwiners $L_{\eta,i,m}^{\orb}$, $\orb\in\Orb$, $[\eta]\in\hat{H}_{\orb}$, $i=1,\ldots,D_\eta$, $m=1,\ldots,M_\eta$, where $M_\eta\in\{0\}\cup\N$ for all $[\eta]\in\hat{H}_{\orb}$ and any orbit $\orb\in\Orb$, we have
$$
\sum_{g\in G}\sum_{i=1}^{D_\eta}U(g)L_{\eta,i,m}^{\orb\,*}L_{\eta,i,n}^{\orb}U(g)^*=\beta^{\orb}_{\eta,m,n}\id_\hil
$$
with some $\beta^{\orb}_{\eta,m,n}\in\C$ for any orbit $\orb\in\Orb$, $[\eta]\in\hat{H}_{\mc O}$, and $m,\,n=1,\ldots,M_\eta$. Thus, the corresponding $(\Set,U,V)$--covariant instrument $\mc I$ is an extreme instrument in the $(\Set,U,V)$--covariance structure if and only if there is only one orbit $\orb_0$ and only one $[\eta_0]\in\hat{H}_{\orb_0}$ such that $L_{\eta,i,m}^{\orb_0}\neq0$ for some $i\in\{1,\ldots,D_{\eta_0}\}$ in which case $m_{\eta_0}=1$, i.e.,\ the only possibly non-zero minimal $(\Set,U,V)$--intertwiners are $L_{\eta_0,i,1}^{\orb_0}$, $i=1,\ldots,D_{\eta_0}$ with a unique orbit $\orb_0$ and a unique $[\eta_0]\in\hat{H}_{\orb_0}$. This means that the instrument $\mc I$ is supported totally on $\orb_0$. If we now equip $\mc I$ with the minimal $(\Set,U,V)$--covariant Stinespring dilation $(\mc M,\msf P,\overline{U},J)$ of Lemma \ref{lemma:minlemma}, the representation $\overline{U}$ only consists of the transitive part $\overline{U}^{\orb_0}$ (see Equation \eqref{eq:transitiveU} in Appendix B). Moreover the multiplicity $m_{\eta_0}$ of $[\eta_0]$ is 1 meaning that $\overline{U}$ is irreducible. This means that, when $U$ is irreducible and we give an $(\Set,U,V)$--covariant instrument $\mc I$ an $(\Set,U,V)$--covariant minimal dilation $(\mc M,\msf P,\overline{U},J)$, where $\mc M$, $\msf P$, and $\overline{U}$ have the decomposition of Equations \eqref{eq:transitiveU} and \eqref{eq:transitiveP} into transitive constituents over $\Orb$ where $\overline{U}^\orb$ is induced from $\pi^\orb:H_\orb\to\mc U(\hil^\orb)$ for each $\orb\in\Orb$, the instrument $\mc I$ is an extreme instrument of the $(\Set,U,V)$--covariance structure if and only if only one of these constituents, corresponding to a fixed $\orb_0\in\Orb$, is non-zero and the corresponding $\pi^\orb$ is irreducible. See Proposition \ref{prop:irrep} for a generalization of this fact in the single-orbit (transitive) case.

\begin{remark}\label{rem:ext}
We say that an instrument $\mc I=(\mc I_x)_{x\in\Set}$ is {\it extreme} if it is an extreme point of the convex set of all instruments with the value space $\Set$, input Hilbert space $\hil$, and output Hilbert space $\mc K$. This extremality property also depends on the minimal Stinespring dilation of the instrument and, if the instrument $\mc I$ is $(\Set,U,V)$--covariant, we can use the minimal dilation presented in Lemma \ref{lemma:minlemma}. It follows that the condition can be formulated as a property of the Kraus operators $K_{x,\eta,i,m}$ of the instrument obtained through Equation \eqref{eq:apuapu1} from the minimal $(\Set,U,V)$--intertwiners $L_{\eta,i,m}^{\orb}$, associated with the instrument $\mc I$: it follows that the instrument $\mc I$ is extreme if and only if the set of operators $K_{x,\eta,i,m}^*K_{x,\tj,j,n}$, $x\in\Set$, $[\eta],\,[\tj]\in\hat{H}_{Gx}$, $i=1,\ldots,D_\eta$, $j=1,\ldots,D_\tj,$ $m=1,\ldots,M_\eta,$ $n=1,\ldots,M_\tj$, is linearly independent. Naturally, an extreme instrument is also an extreme instrument of the $(\Set,U,V)$--covariance structure; in  Appendix C we see how this can be seen directly using the respective extremality characterizations. \hfill $\triangle$
\end{remark}

\section{Observables and channels covariant with respect to a finite group}\label{sec:finObsCh}

Let us retain the finite group $G$ and the $G$-space structure of the value space $\Set$ and the representation $U:\,G\to\mc U(\hil)$ of the preceding section. We say that an $(\Set,U)$--covariant observable $\msf M$ (i.e.,\ a POVM satisfying Equation \eqref{kovarianssiehto}) is an {\it extreme observable of the $(\Set,U)$--covariance structure} if $\msf M$ is an extreme point of the convex set of all $(\Set,U)$--covariant observables. We may view an $(\Set,U)$--covariant observable as a particular $(\Set,U,V)$--covariant instrument with the trivial output space $\C$ where $V$ is the trivial representation of $G$. Using this observation and Theorems \ref{theor:CovInstrStructure} and \ref{theor:ExtInstr}, we obtain the following result characterizing the $(\Set,U)$--covariant observables (and thus elaborating on Theorem \ref{theor:CovObsBasic}) and the extreme observables of the $(\Set,U)$--covariance structure. As the result is a direct corollary, we do not give a separate proof for it. Note that extreme points of sets of covariant observables have also been studied in \cite{ChiDA2006,HaPe11,HaPe15,HoPe2009}. Also the non-covariant results presented in \cite{Parthasarathy99} can be seen as corollaries of the following extremality characterization (in the case where every orbit is a singleton).

\begin{corollary}\label{cor:CovObs}
Let $\msf M=(\msf M_x)_{x\in\Set}$ be an $(\Set,U)$--covariant observable. For any orbit $\orb\in\Orb$, there is an operator $K_{\orb}\in\mc L(\hil)$ such that, for any $g\in G$,
\begin{equation}\label{eq:CovObsChar}
\msf M_{gH_{\orb}}=U(g)K_\orb U(g)^*.
\end{equation}
For any $\orb\in\Orb$, the above operator $K_\orb$ has the following structure: For all $[\eta]\in\hat{H}_{\orb}$ there is a number $M_\eta\in\{0\}\cup\N$ and a linearly independent set
$$
\{d_{\eta,i,m}^{\orb}\in\hil\,|\,[\eta]\in\hat{H}_{\orb},\ i=1,\ldots,D_\eta,\ m=1,\ldots,M_\eta\}
$$
such that, for any $[\eta]\in\hat{H}_{\orb}$, $i=1,\ldots,D_\eta$, $m=1,\ldots,M_\eta$, and $h\in H_{\orb}$,
\begin{equation}\label{eq:ObsHinv}
U(h)d_{\eta,i,m}^{\orb}=\sum_{j=1}^{D_\eta}\eta_{j,i}(h)d_{\eta,j,m}^{\orb}
\end{equation}
and
\begin{equation}\label{eq:CovObsKernel}
K_{\orb}=\sum_{[\eta]\in\hat{H}_{\orb}}\sum_{i=1}^{D_\eta}\sum_{m=1}^{M_\eta}|d_{\eta,i,m}^{\orb}\>\<d_{\eta,i,m}^{\orb}|.
\end{equation}
Furthermore,
\begin{equation}\label{eq:ObsNormitus}
\sum_{\orb\in\Orb}\sum_{g\in G}\sum_{[\eta]\in\hat{H}_{\orb}}\sum_{i=1}^{D_\eta}\sum_{m=1}^{M_\eta}\frac{1}{\# H_{\orb}}|U(g)d_{\eta,i,m}^{\orb}\>\<U(g)d_{\eta,i,m}^{\orb}|=\id_\hil.
\end{equation}
This observable is an extreme observable of the $(\Set,U)$--covariance structure if and only if the set
$$
\left\{\sum_{g\in G}\sum_{i=1}^{D_\eta}|U(g)d_{\eta,i,m}^{\orb}\>\<U(g)d_{\eta,i,n}^{\orb}|\,\bigg|\,m,\,n=1,\ldots,M_\eta,\ [\eta]\in\hat{H}_{\orb},\ \orb\in\Orb\right\}
$$
is linearly independent. Moreover, when $d_{\eta,i,m}^{\orb}$, $\orb\in\Orb$, $[\eta]\in\hat{H}_{\orb}$, $i=1,\ldots,D_\eta$, $m=1,\ldots,M_\eta$, where $M_\eta\in\{0\}\cup\N$, are vectors satisfying Equations \eqref{eq:ObsHinv} and \eqref{eq:ObsNormitus},\footnote{Sometimes the vectors $U(g)d_{\eta,i,m}^{\orb}$ are called {\it generalized coherent states.}} 
 Equations \eqref{eq:CovObsChar} and \eqref{eq:CovObsKernel} define an $(\Set,U)$--covariant observable.
\end{corollary}

\begin{example}\label{ex:genSym}

We continue to study the situation introduced in Example \ref{ex:Sym} and generalize it to a general finite dimension $D\ge 2$. The Hilbert space of our system is $\hil_D\simeq\C^D$, the symmetry group is the permutation group $S_D={\rm Sym}\big(\{1,2,\ldots,D\}\big)$ which operates in the value space $\Set_D^2=\{1,\ldots,D\}^2$ of our measurements through $S_D\times\Set_D^2\ni\big(\pi,(m,n)\big)\mapsto\pi(m,n)=\big(\pi(m),\pi(n)\big)\in\Set_D^2$ and in $\hil_D$ through the unitary representation $U:S_D\to\mc U(\hil_D)$ defined w.r.t.\ a fixed orthonormal basis $\{|n\>\}_{n=1}^D$ of $\hil_D$ via $U(\pi)|n\>=|\pi(n)\>$ for all $\pi\in S_D$ and $n=1,\ldots,D$. Note that $U$ is not irreducible as $\psi_0:=D^{-1/2}\big(|1\>+\cdots+|D\>\big)$ is invariant under $U$ and thus $U$ can be restricted to the orthogonal complement $\{\psi_0\}^\perp$. This restriction is irreducible and is called as the standard representation of $S_D$.

The set $\Set_D^2$ splits into two orbits, the diagonal $\orb=\{(1,1),(2,2),\ldots,(D,D)\}$ and the off-diagonal $\orb'=\Set_D^2\setminus\orb$. Picking the reference points $x_\orb=(1,1)$ and $x_{\orb'}=(1,2)$, the stability subgroup $H_\orb$ is easily seen to be the subgroup of those $\pi\in S_D$ such that $\pi(1)=1$ and the stability subgroup $H_{\orb'}$ is easily seen to consist of those $\pi'\in S_D$ such that $\pi'(1)=1$ and $\pi'(2)=2$. Hence, $H_\orb\simeq{\rm Sym}\big(\{2,3,\ldots,D\}\big)\simeq S_{D-1}$ and $H_{\orb'}\simeq{\rm Sym}\big(\{3,4,\ldots,D\}\big)\simeq S_{D-2}$; if $D=2$ then $H_{\orb'}$ is  the single-element group. It follows that, for any $(\Set_D^2,U)$--covariant observable $\msf M=(\msf M_{(m,n)})_{(m,n)\in\Set_D^2}$, there are positive kernels $K_\orb$ and $K_{\orb'}$ such that $U(\pi)K_\orb=K_\orb U(\pi)$ for all $\pi\in H_{\orb}$ and $U(\pi')K_{\orb'}=K_{\orb'}U(\pi')$ for all $\pi'\in H_{\orb'}$ 
and $\msf M_{(\pi(1),\pi(1))}=U(\pi)K_\orb U(\pi)^*$ for all $\pi\in S_D$ (defining the diagonal values) and $\msf M_{(\pi(1),\pi(2))}=U(\pi)K_{\orb'}U(\pi)^*$ for all $\pi\in S_D$ (defining the off-diagonal values). Furthermore, there are non-negative integers $M_\eta$ and $M_{\eta'}$, $[\eta]\in\hat{H}_{\orb}$, $[\eta']\in\hat{H}_{\orb'}$, and two linearly independent sets $\{d_{\eta,i,m}\,|\,[\eta]\in\hat{H}_{\orb},\ i=1,\ldots,D_\eta,\ m=1,\ldots,M_\eta\}$ and $\{d'_{\eta',k,r}\,|\,[\eta']\in\hat{H}_{\orb'},\ k=1,\ldots,D_{\eta'},\ r=1,\ldots,M_{\eta'}\}$ of vectors from $\hil_D$ such that $U(\pi)d_{\eta,i,m}=\sum_{j=1}^{D_\eta}\eta_{j,i}(\pi)d_{\eta,j,m}$ for all $\pi\in H_{\orb}$, $[\eta]\in\hat{H}_{\orb}$, $i=1,\ldots,D_\eta$, and $m=1,\ldots,M_\eta$ and $U(\pi')d'_{\eta',k,r}=\sum_{\ell=1}^{D_{\eta'}}\eta'_{\ell,k}(\pi')d'_{\eta',\ell,r}$ for all $\pi'\in H_{\orb'}$, $[\eta']\in\hat{H}_{\orb'}$, $k=1,\ldots,D_{\eta'}$, and $r=1,\ldots,M_{\eta'}$.

In exactly the same way as in Example \ref{ex:Sym}, we obtain (rank-1) PVMs when concentrating on the diagonal orbit. Let us construct a family of rank-1 informationally complete  extreme $(\Set_D^2,U)$--covariant POVMs. As we are interested in the rank-1 case, we only concentrate on the characters (1-dimensional irreducible representations) of the stability subgroups and unit multiplicities in the above framework. Let us make things simple by just assuming that the characters involved are just the trivial characters $\zeta_0\in\hat{H}_{\orb}$ and $\zeta'_0\in\hat{H}_{\orb'}$, i.e., $\<\pi,\zeta_0\>=1=\<\pi',\zeta'_0\>$ for all $\pi\in H_{\orb}$ and $\pi'\in H_{\orb'}$. It follows that we only have single vectors $d_{\orb}:=d_{\zeta_0,1,1}$ and $d_{\orb'}:=d'_{\zeta'_0,1,1}$ which satisfy $U(\pi)d_{\orb}=\<\pi,\zeta_0\>d_{\orb}=d_{\orb}$ for all $\pi\in H_{\orb}$ and $U(\pi')d_{\orb'}=\<\pi',\zeta'_0\>d_{\orb'}=d_{\orb'}$ for all $\pi'\in H_{\orb'}$.

Note that we do not have to overly worry about the normalization of the vectors $d_{\orb}$ and $d_{\orb'}$ for the moment as we can carry out the normalization afterwards according to Remark \ref{normalization}. Let us make the ansatz $d_{\orb}=|1\>$ and $d_{\orb'}=d_{\orb'}(\alpha):=\alpha(e^{-i\pi/8}|1\>+e^{i\pi/8}|2\>)$ where $\alpha\geq0$; indeed these are valid choices as they comply with the above necessary conditions. The case $\alpha=0$ corresponds to a rank-1 PVM supported by the diagonal $\orb$. For now, let us assume that $\alpha>0$. 
For any $m\le D$, we obtain $|m\>\<m|$ as an operator
$U(\pi)|1\>\<1|U(\pi)^*$ for some $\pi\in S_D$.
 Let $(m,n)\in\orb'$ and choose  $\pi\in S_D$ such that $\pi(1)=m$ and $\pi(2)=n$. It now follows that
$$
U(\pi)|d_{\orb'}(\alpha)\>\<d_{\orb'}(\alpha)|U(\pi)^*=\alpha^2\big(|m\>\<m|+e^{-i\pi/4}|m\>\<n|+e^{i\pi/4}|n\>\<m|+|n\>\<n|\big).
$$
Through linear combinations with operators from the diagonal, we now obtain the operators
$$
A:=e^{-i\pi/4}|m\>\<n|+e^{i\pi/4}|n\>\<m|,\qquad B:=e^{i\pi/4}|m\>\<n|+e^{-i\pi/4}|n\>\<m|,
$$
where $B$ is obtained by reversing the roles of $m$ and $n$, and ultimately $
2^{-3/2}\big[(A+B)+i(A-B)\big]=|m\>\<n|$. All in all, the operators $U(\pi)|1\>\<1|U(\pi)^*$ and $U(\pi)|d_{\orb'}(\alpha)\>\<d_{\orb'}(\alpha)|U(\pi)^*$, where $\pi\in S_D$, span the whole of $\mc L(\hil_D)$.
Following Remark \ref{normalization}, we may define
\begin{align*}
K(\alpha)&=\frac{1}{\# H_\orb}\sum_{\pi\in S_D}U(\pi)|1\>\<1|U(\pi)^*+\frac{1}{\# H_{\orb'}}\sum_{\pi\in S_D}U(\pi)|d_{\orb'}(\alpha)\>\<d_{\orb'}(\alpha)|U(\pi)^*\\
&=\Big[\big(2D-2-\sqrt{2}\big)\alpha^2+1\Big]\id+\sqrt{2}\alpha^2 D|\psi_0\>\<\psi_0|\\
&=\Big[\big(2D-2-\sqrt{2}\big)\alpha^2+1\Big]\big(\id-|\psi_0\>\<\psi_0|\big)+\Big[
\big(2+\sqrt{2}\big)\big(D-1\big)\alpha^2+1
\Big]
|\psi_0\>\<\psi_0|
\end{align*}
where the second equality is obtained through direct calculation and the final formula is the spectral resolution of $K(\alpha)$;\footnote{
Note that any operator commuting with $U$ has a spectral resolution like this recalling the decomposition of $U$ into the trivial character operating in the 1-dimensional subspace spanned by $\psi_0$ and to the standard representation operating in $\{\psi_0\}^\perp$.}
 recall the isotropic vector $\psi_0$ defined in the beginning of this example.
 Hence, we have the normalizer
 $$
 K(\alpha)^{-1/2}=
 \Big[\big(2D-2-\sqrt{2}\big)\alpha^2+1\Big]^{-1/2}\big(\id-|\psi_0\>\<\psi_0|\big)+\Big[
\big(2+\sqrt{2}\big)\big(D-1\big)\alpha^2+1
\Big]^{-1/2}
|\psi_0\>\<\psi_0|
 $$
 and we may define the $(\Set_D^2,U)$--covariant rank-1 POVM $\msf M^\alpha=(\msf M^\alpha_{(m,n)})_{(m,n)\in\Set_D^2}$ for all $\alpha\geq0$ through 
 \begin{eqnarray*}
\msf M^\alpha_{(m,m)}&=&K(\alpha)^{-1/2}|m\>\<m|K(\alpha)^{-1/2},\\
\msf M^\alpha_{(m,n)}&=&U(\pi)K(\alpha)^{-1/2}|d_{\orb'}(\alpha)\>\<d_{\orb'}(\alpha)|K(\alpha)^{-1/2}U(\pi)^* \\
&=&
\alpha^2
K(\alpha)^{-1/2}
\big(|m\>\<m|+e^{-i\pi/4}|m\>\<n|+e^{i\pi/4}|n\>\<m|+|n\>\<n|\big)
K(\alpha)^{-1/2}
\end{eqnarray*}
for all $m\ne n$ where $\pi\in S_D$ is such that $\pi(1)=m$ and $\pi(2)=n$. 
Whenever $\alpha>0$, using our observations just before introducing $K(\alpha)$ and the fact that $K(\alpha)^{-1/2}$ commutes with $U$,
the range of $\msf M^\alpha$ spans $\mc L(\hil_D)$ showing that $\msf M^\alpha$ is informationally complete. Since $\msf M^\alpha$ has $D^2$ non-zero outcomes when $\alpha>0$, this also implies that the set $\{\msf M^\alpha_{(m,n)}\,|\,(m,n)\in\Set_D^2\}$ is linearly independent. Hence, as a rank-1 POVM, $\msf M^\alpha$ is also extreme within the convex set of all observables with a finite outcome space and operating in $\hil_D$  \cite{HaHePe2012}. In the case $\alpha=0$, one gets the rank-1 PVM $\msf M_{(m,n)}^0=\delta_{m,n}|m\>\<m|$. To conclude, both of the mutually exclusive classes of optimal observables are represented within the $(\Set_D^2,U)$--covariance structure and they are arbitrarily close one another when $\alpha\approx0$.
It is easy to see that in the limit $\alpha\to\infty$, the diagonal effects of $\Mo^\alpha$ vanish so that the limit rank-1 POVM is not informationally complete. The limit POVM is a PVM only if $D=2$.
We observe that the margin\footnote{Defined by $\Ao^\alpha_m:=\sum_{n=1}^D\msf M^\alpha_{(m,n)}$ and $\sfb^\alpha_n:=\sum_{m=1}^D\msf M^\alpha_{(m,n)}$.} POVMs $(\Ao^\alpha_m)_{m=1}^D$ and $(\sfb^\alpha_n)_{n=1}^D$ are $(\Set_D,U)$--covariant (e.g.\ $U(\pi)\Ao_m^\alpha U(\pi)^*=\Ao_{\pi(m)}$)
but they are not of rank 1 except in the case $\alpha=0$ when they coincide with the basis measurement $\big(\kb m m\big)_{m=1}^D$ and $\Mo^0$ is their only possible joint measurement.

We notice that $K(\alpha)^{-1/2}$ is particularly simple when $\alpha=0$ or $\alpha=(2+\sqrt{2})^{-1/2}=:\alpha_0$. According to the above discussion, $\msf M^{\alpha_0}$ is an example of a rank-1 extreme informationally complete observable in the $(\Set_D^2,U)$--covariance structure. In a straightforward manner, we find that, for all $m,\,n=1,\ldots,D$, $\msf M^{\alpha_0}_{(m,n)}=|d_{m,n}\>\<d_{m,n}|$ where
$$
d_{m,n}=\frac{1}{\sqrt{2D}}\big(e^{-i\pi/8}|m\>+e^{i\pi/8}|n\>\big)-\frac{1}{D}\left(\sqrt{1+1/\sqrt{2}}-1\right)\psi_0.
$$
\hfill $\triangle$
\end{example}

In addition to a POVM, a quantum measurement associated with an instrument $\mc I=(\mc I_x)_{x\in\Set}$ with input Hilbert space $\hil$ and output space $\mc K$ also defines the total unconditioned state transformation $\sum_{x\in\Set}\mc I_x$ from the set $\mc S(\hil)$ of input states to the set of output states $\mc S(\mc K)$. This transformation is also known as a {\it channel}, a trace-preserving completely positive (affine) map. We immediately see that any channel can be viewed as an instrument with a single outcome. Let us again assume that $G$ is a finite group and that $U:G\to\mc U(\hil)$ and $V:G\to\mc U(\mc K)$ are unitary representations mediating the input and output symmetries. We say that a channel $\Phi:\mc S(\hil)\to\mc S(\mc K)$ is {\it $(U,V)$--covariant} if, for all $g\in G$ and $\rho\in\mc S(\hil)$,
$$
\Phi\big(U(g)\rho U(g)^*\big)=V(g)\Phi(\rho)V(g)^*.
$$
Furthermore, we say that a $(U,V)$--covariant channel $\Phi$ is an {\it extreme channel of the $(U,V)$--covariance structure} if $\Phi$ is an extreme point of the convex set of all $(U,V)$--covariant channels. Clearly, a $(U,V)$--covariant channel is an example of an $(\Set,U,V)$--covariant instrument where $\Set=\{x_0\}$ is a singleton where $G$ acts trivially. The following is again a direct corollary of Theorems \ref{theor:CovInstrStructure} and \ref{theor:ExtInstr} and the above observation.

\begin{corollary}\label{cor:CovCh}
Let $\Phi$ be a $(U,V)$--covariant channel. There is, for any $[\tj]\in\hat{G}$, a number $M_\tj\in\{0\}\cup\N$, and a linearly independent set
$$
\{L_{\tj,i,m}\in\mc L(\hil,\mc K)\,|\,[\tj]\in\hat{G},\ i=1,\ldots,D_\tj,\ m=1,\ldots,M_\tj\}
$$
of operators such that, for any $[\tj]\in\hat{G}$, $i=1,\ldots,D_\tj$, $m=1,\ldots,M_\tj$, and $g\in G$,
\begin{equation}\label{eq:ChGinv}
L_{\tj,i,m}U(g)=\sum_{j=1}^{D_\tj}\tj_{i,j}(g)V(g)L_{\tj,j,m},
\end{equation}
\begin{equation}\label{eq:ChNormitus}
\sum_{[\tj]\in\hat{G}}\sum_{i=1}^{D_\tj}\sum_{m=1}^{M_\tj}L_{\tj,i,m}^*L_{\tj,i,m}=\id_\hil,
\end{equation}
and, for any $\rho\in\mc S(\hil)$,
\begin{equation}\label{eq:CovChStr}
\Phi(\rho)=\sum_{[\tj]\in\hat{G}}\sum_{i=1}^{D_\tj}\sum_{m=1}^{M_\tj}L_{\tj,i,m}\rho L_{\tj,i,m}^*.
\end{equation}
This channel is an extreme channel of the $(U,V)$--covariance structure if and only if the set
$$
\left\{\sum_{i=1}^{D_\tj}L_{\tj,i,m}^*L_{\tj,i,n}\,\bigg|\,m,\,n=1,\ldots,M_\tj,\ [\tj]\in\hat{G}\right\}
$$
is linearly independent. Moreover, given a set of linear operators $L_{\tj,i,m}$, $[\tj]\in\hat{G}$, $i=1,\ldots,D_\tj$, $m=1,\ldots,M_\tj$, where $M_\tj\in\{0\}\cup\N$, satisfying Equations \eqref{eq:ChGinv} and \eqref{eq:ChNormitus}, Equation \eqref{eq:CovChStr} defines a $(U,V)$--covariant channel.
\end{corollary}

Suppose that $L_{\tj,i,m}:\hil\to\mc K$, $[\tj]\in\hat{G}$, $i=1,\ldots,D_\tj$, $m=1,\ldots,M_\tj$, where $M_\tj\in\{0\}\cup\N$, satisfy the condition of Equation \eqref{eq:ChGinv}. It easily follows that, for any $[\tj]\in\hat{G}$ and $m,\,n=1,\ldots,M_\tj$, the operator $\sum_{i=1}^{D_\tj}L_{\tj,i,m}^*L_{\tj,i,n}$ commutes with $U$. This is why we may omit the $G$-summations in the normalization condition of Equation \eqref{eq:ChNormitus}, the channel characterization of Equation \eqref{eq:CovChStr}, and the operators essential for the extremality characterization of Corollary \ref{cor:CovCh}.

\section{Covariant continuous instruments\\ associated with a compact stability subgroup}\label{sec:ContInstr}

We now concentrate on continuous quantum measurements possibly in infinite-dimensional systems and their symmetry properties. For this, we explicitly define $\mc L(\hil)$ as the algebra of bounded operators on the Hilbert space $\hil$ and $\mc T(\hil)$ as the trace class on $\hil$. Whenever $(\Set,\Sigma)$ is a measurable space (i.e.,\ $\Set\neq\emptyset$ and $\Sigma$ is a $\sigma$-algebra of subsets of $\Set$) and $\hil$ and $\mc K$ are Hilbert spaces, we say that a map $\mc I:\Sigma\times\mc T(\hil)\to\mc T(\mc K)$ is an instrument with the value space $(\Set,\Sigma)$, input space $\hil$, and output space $\mc K$ if, for any $X\in\Sigma$, $\mc I(X,\cdot)$ is an operation, $\mc I(\Omega,\cdot)$ is trace preserving, and, for any disjoint sequence $X_1,\,X_2,\ldots\in\Sigma$ and any $\rho\in\mc T(\hil)$,
$$
\mc I\big(\cup_{i=1}^\infty X_i,\rho\big)=\sum_{i=1}^\infty\mc I(X_i,\rho)
$$
where the sum converges w.r.t.\ the trace norm topology. For any instrument $\mc I:\Sigma\times\mc T(\hil)\to\mc T(\mc K)$, we also define the Heisenberg instrument $\mc I^*:\Sigma\times\mc L(\mc K)\to\mc L(\hil)$ through
$$
\tr{\rho\mc I^*(X,B)}=\tr{\mc I(X,\rho)B},\qquad\rho\in\mc T(\hil),\quad B\in\mc L(\mc K),\quad X\in\Sigma,
$$
i.e.,\ for all $X\in\Sigma$, $\mc I^*(X,\cdot)$ is the Heisenberg dual operation of $\mc I(X,\cdot)$. If $\Set$ is a topological space, we denote the corresponding Borel $\sigma$-algebra by $\mc B(\Set)$; there is never any ambiguity about which is the topology concerned, so the topology is not specifically indicated in this notation.

Let $G$ be a group. We say that a set $\Set$ is a [transitive] $G$-space if there is a map $G\times\Set\ni(g,x)\mapsto gx\in\Set$ such that $ex=x$ for all $x\in\Set$ and $(gh)x=g(hx)$ for all $g,\,h\in G$ and $x\in\Set$ [and, for any $x,\,y\in\Set$, there is $g\in G$ such that $gx=y$]. Suppose that $(\Set,\Sigma)$ is a measurable space where $\Set$ is a $G$-space and that, for any $g\in G$, the map $x\mapsto gx$ is measurable. Let $\hil$ and $\mc K$ be Hilbert spaces and $U:G\to\mc U(\hil)$ and $V:G\to\mc U(\mc K)$ be unitary representations. We say that an instrument $\mc I:\Sigma\times\mc T(\hil)\to\mc T(\mc K)$ is {\it $(\Sigma,U,V)$--covariant} if, for any $X\in\Sigma$, $\rho\in\mc T(\hil)$, and $g\in G$,
$$
\mc I\big(gX,U(g)\rho U(g)^*\big)=V(g)\mc I(X,\rho)V(g)^*.
$$
In the special case $\mc K=\C$, the set of $(\Sigma,U,V)$--covariant instruments simplifies to the set of $(\Sigma,U)$--covariant observables (POVMs), i.e.,\ weakly $\sigma$-additive maps $\msf M:\Sigma\to\mc L(\hil)$ such that $\msf M(\Omega)=\id_\hil$ (normalization) and
$$
U(g)\msf M(X)U(g)^*=\msf M(gX),\qquad g\in G,\quad X\in\Sigma.
$$

For any $(\Sigma,U,V)$--covariant instrument $\mc I$, there is a quadruple $(\mc M,\msf P,\overline{U},J)$ consisting of a Hilbert space $\mc M$, a projection-valued measure (PVM) $\msf P:\Sigma\to\mc L(\mc M)$ (a projection-valued set function which is weakly or, equivalently, strongly $\sigma$-additive, $\msf P(\emptyset)=0$, and $\msf P(G/H)=\id_{\mc M}$), a unitary representation $\overline{U}:G\to\mc U(\mc M)$, and an isometry $J:\hil\to\mc K\otimes\mc M$ so that
\begin{itemize}
\item[(i)] $\mc I^*(X,B)=J^*\big(B\otimes\msf P(X)\big)J$ for all $X\in\Sigma$ and $B\in\mc L(\mc K)$,
\item[(ii)] $JU(g)=\big(V(g)\otimes\overline{U}(g)\big)J$ for all $g\in G$,
\item[(iii)] $\overline{U}(g)\msf P(X)\overline{U}(g)^*=\msf P(gX)$ for all $g\in G$ and $X\in\Sigma$, and
\item[(iv)] the vectors $\big(B\otimes\msf P(X)\big)J\fii$, $B\in\mc L(\mc K)$, $X\in\Sigma$, $\fii\in\hil$, span a dense subspace of $\mc K\otimes\mc M$.
\end{itemize}
The existence of a triple $(\mc M,\msf P,J)$ satisfying items (i) and (iv) above is well known, and the existence of the unitary representation $\overline{U}$ satisfying items (ii) and (iii) is proven essentially in the same way as in the finite-outcome and finite-dimensional case which is studied in Appendix B.

Let $G$ be a locally compact second-countable group which is Hausdorff. If $\Omega$ is locally compact, second countable, and Hausdorff and $\Omega$ is a transitive $G$-space such that the map $G\times\Omega\ni(g,\omega)\mapsto g\omega\in\Omega$ is continuous, there is a closed subgroup $H\leq G$ such that $\Omega$ is homeomorphic with $G/H$ (space of left cosets) and, in this identification, the $G$-action is of the form
$$
g(g'H)=(gg')H,\qquad g,\,g'\in G.
$$ 

From now on, we assume that $\hil$ and $\mc K$ are separable Hilbert spaces, $G$ is a locally compact and second-countable group which is Hausdorff, $H\leq G$ is a closed subgroup, and $U:G\to\mc U(\hil)$ and $V:G\to\mc U(\mc K)$ are strongly continuous\footnote{That is, e.g.,\ $g\mapsto U(g)\fii$ is continuous for any $\fii\in\hil$.} unitary representations. We will concentrate on $\big(\mc B(G/H),U,V\big)$--covariant instruments and $\big(\mc B(G/H),U,V\big)$--covariant dilations which we will call, for short, {\it $(G/H,U,V)$--covariant}. In the same context, we call $\big(\mc B(G/H),U\big)$--covariant observables as {\it $(G/H,U)$--covariant}. Note that, we are now restricting to the transitive, i.e.,\ single-orbit case. We also fix a quasi-$G$--invariant measure $\mu:\mc B(G/H)\to[0,\infty]$ and a measurable section $s:G/H\to G$ for the factor projection $g\mapsto gH$ such that $s(H)=e$. It is well known \cite{varadarajan} that, fixing a left Haar measure $\mu_G$ for $G$, there is a $(\mu_G\times\mu)$-measurable function $\rho:G\times G/H\to(0,\infty)$ coinciding $(\mu_G\times\mu)$-a.e.\ with the function $(g,x)\mapsto(d\mu_g/d\mu)(x)$ where $\mu_g(X)=\mu(gX)$ for all $X\in\mc B(G/H)$ and $\rho(gh,x)=\rho(g,hx)\rho(h,x)$ for $(\mu_G\times\mu_G\times\mu)$-a.a.\ $(g,h,x)\in G\times G\times G/H$. As in Section \ref{sec:finInstr}, we define, for any unitary representation $\pi:H\to\mc U(\hil_\pi)$, the cocycle $\zeta^\pi:G\times G/H\to\mc U(\hil_\pi)$ through $\zeta^\pi(g,x)=\pi\big(s(x)^{-1}g^{-1}s(gx)\big)$ for all $g\in G$ and $x\in G/H$. The cocycle conditions \eqref{eq:cocycle} still hold. In this setting, a $(G/H,U,V)$--covariant instrument $\mc I$ has a very particular minimal $(G/H,U,V)$--covariant dilation $(\mc M,\msf P,\overline{U},J)$ \cite{CaHeTo2009,HaPe15}: There is a strongly continuous unitary representation $\pi:H\to\mc U(\hil_\pi)$ in some separable Hilbert space $\hil_\pi$ such that $\mc M=L^2_\mu\otimes\hil_\pi$ (which we identify with the Hilbert space of $\mu$-equivalence classes of $\mu$-square-integrable functions $F:G/H\to\hil_\pi$), $\msf P=\msf P_\pi^G$ defined through 
\begin{equation}\label{eq:genIndP}
\big(\msf P_\pi^G(X)F\big)(x)=\chi_X(x)F(x),\qquad X\in\mc B(G/H),\quad F\in L^2_\mu\otimes\hil_\pi,\quad x\in G/H,
\end{equation}
and $\overline{U}=U_\pi^G$ defined through
\begin{equation}\label{eq:genIndU}
\big(U_\pi^G(g)F\big)(x)=\sqrt{\rho(g^{-1},x)}\zeta^\pi(g^{-1},x)F(g^{-1}x),\qquad g\in G,\quad F\in L^2_\mu\otimes\hil_\pi,\quad x\in G/H.
\end{equation}
The representation $U_\pi^G$ is called as the representation {\it induced from $\pi$} and $(\msf P_\pi^G,U_\pi^G)$ is the {\it canonical system of imprimitivity associated to $\pi$}; note that $\msf P_\pi^G$ is a $(G/H,U_\pi^G)$--covariant PVM.

We additionally make the following more specific assumptions:
\begin{itemize}
\item[(a)] There is a dense subspace $\mc D$ of $\hil$ which is $U$--invariant, i.e.,\ $U(g)\mc D\subseteq\mc D$ for all $g\in G$.
\item[(b)] There is a norm $\|\cdot\|_1:\mc D\to[0,\infty)$ so that $(\mc D,\|\cdot\|_1)$ is a separable normed space. Moreover, for all $g\in G$ and $\fii\in\mc D$, $\|U(g)\fii\|_1=\|\fii\|_1$.
\item[(c)] For any $(G/H,U)$--covariant POVM $\msf M$, there is a strongly continuous unitary representation $\pi_0:H\to\hil_{\pi_0}$ in a separable Hilbert space $\hil_{\pi_0}$ and a linear operator $\Theta:\mc D\to\hil_{\pi_0}$ such that $\|\Theta\fii\|\leq\|\fii\|_1$ for all $\fii\in\mc D$, $\Theta U(h)=\pi_0(h)\Theta$ for all $h\in H$ and, defining the linear map $J:\hil\to\hil_{\pi_0}^G$ through $(J\fii)(x)=\pi_0\big(s(gH)^{-1}g\big)\Theta U(g)^*\fii$ for all $\fii\in\mc D$ and $g\in G$, $(L_\mu^2\otimes\hil_{\pi_0},\msf P_{\pi_0}^G,U_{\pi_0},J)$ is a minimal $(G/H,U)$--covariant Na\u{\i}mark dilation for $\msf M$ (i.e.,\ a minimal $(G/H,U,V_0)$--covariant Stinespring dilation for $\msf M$ when $\msf M$ is viewed as an instrument with the trivial output space $\C$ where the representation $V_0$ is chosen as trivial).
\end{itemize}

As an example, suppose that $G$ Abelian and define $\mc D\subseteq\hil$ and the norm $\|\cdot\|_1$ on $\mc D$ in the same way as in \cite{HaPe11} just after Proposition 3.1. It is quite easily seen that $(\mc D,\|\cdot\|_1)$ is a separable metric space and the results of \cite{HaPe11}, imply that items (a), (b), and (c) above hold when we set $\Theta=\mathfrak{W}$ where $\mathfrak{W}$ is the linear map defining a $(G/H,U)$--covariant POVM appearing in Theorem 3.1 of \cite{HaPe11}. As a second example, let $G$ unimodular and let $H$ be compact and assume that the decomposing measure $\mu_U$ for $U$ is absolutely continuous with respect to the Plancherel measure $\mu_{\hat{G}}$. Moreover, define the subspace $\mc D\subseteq\hil$ in the same way as in \cite{HaPe15} and, using the notations introduced therein, define the norm $\|\cdot\|_1$ through
$$
\|\zeta\star\xi\|_1=\int_{\hat{G}}\|\zeta(\gamma)\|\|\xi(\gamma)\|\,d\mu_{\hat{G}}(\gamma).
$$
It again follows quite easily that $(\mc D,\|\cdot\|_1)$ is separable and, perusing the proof of Theorem 3 of \cite{HaPe15}, one finds that $\|\Lambda\fii\|\leq\|\fii\|_1$ for all $\fii\in\mc D$ where $\Lambda$ is the operator of Theorem 3 of \cite{HaPe15} associated to a $(G/H,U)$--covariant observable. The same theorem states that items (a), (b), and (c) above hold upon setting, for each $(G/H,U)$--covariant POVM $\msf M$, $\Theta=\Lambda$ where $\Lambda$ is the linear map associated to $\msf M$ by this theorem. A third example where conditions (a), (b), and (c) hold is the case like that above, except that $U$ is square integrable, as the results of \cite{kiukas_etal2006} and Section 6.1 of \cite{HaPe15} show.


Using conditions (a), (b), and (c), one can prove a counterpart of Theorem 4 of \cite{HaPe15} using same methods as we will employ shortly. However, in order to obtain more interesting results, we have to assume that
\begin{itemize}
\item[(d)] $H\leq G$ is compact.
\end{itemize}
It hence follows that the dual $\hat{H}$ is countable. As earlier, we pick, for any $[\eta]\in\hat{H}$ a representative $\eta:H\to\mc U(\mc K_\eta)$ and denote by $D_\eta$ the dimension of $\mc K_\eta$ (which is finite). For any $[\eta]\in\hat{H}$, we also fix an orthonormal basis $\{e_{\eta,i}\}_{i=1}^{D_\eta}\subset\mc K_\eta$ and denote
$$
\eta_{i,j}(h):=\<e_{\eta,i}|\eta(h)e_{\eta,j}\>,\qquad i,\,j=1,\ldots,D_\eta,\quad h\in H.
$$
Moreover, for any $[\eta]\in\hat{H}$, we define the functions $\zeta^\eta_{i,j}:G\times G/H\to\C$ through the matrix elements of $\zeta^\eta$ in the basis $\{e_{\eta,i}\}_{i=1}^{D_\eta}$. Since $H$ is compact, $G/H$ allows an essentially unique regular $G$--invariant measure $\mu:\mc B(G/H)\to[0,\infty]$, i.e.,\ $\mu(gX)=\mu(X)$ for all $X\in\mc B(G/H)$. We keep this measure fixed in the sequel implying that we may assume $\rho\equiv1$ in the definition \eqref{eq:genIndU} of the induced representation.

Let us make a useful definition. Below, we say that, given a set $A$, a set $\{L_a\}_{a\in A}$ of linear operators $L_a:\mc D\to\mc K$ is {\it $(\mc K,\mc D)$-weakly independent} if, for $(\beta_a)_{a\in A}\in\ell^2_A$, the condition $\sum_{a\in A}\beta_a\<\psi|L_a\fii\>=0$ for any $\psi\in \mc K$ and any $\fii\in\mc D$ implies $\beta_a=0$ for all $a\in A$. Moreover, the notation $m=1,\ldots,M$ is to be taken as usual when $M\in\N$; if $M=0$, this means that the set of indices $m$ discussed is empty; and if $M=\infty$, the set of indices $m$ is the entire $\N$.

\begin{definition}
We say that, given $M_\eta\in\N\cup\{0,\infty\}$ for any $[\eta]\in\hat{H}$, a set
$$
\{L_{\eta,i,m}\,|\,[\eta]\in\hat{H},\ i=1,\ldots,D_\eta,\ m=1,\ldots,M_\eta\}
$$
of linear operators $L_{\eta,i,m}:\mc D\to\mc K$ is a {\it minimal set of $(G/H,U,V)$--intertwiners} if it is $(\mc K,\mc D)$-weakly independent,
\begin{equation}\label{eq:genHinv}
L_{\eta,i,m}U(h)=\sum_{j=1}^{D_\eta}\eta_{i,j}(h)V(h)L_{\eta,j,m}
\end{equation}
for all $[\eta]\in\hat{H}$, $i=1,\ldots,D_\eta$, $m=1,\ldots,M_\eta$, and $h\in H$, 
\begin{equation}\label{eq:DNormi}
\sum_{[\eta]\in\hat{H}}\sum_{i=1}^{D_\eta}\sum_{m=1}^{M_\eta}\|L_{\eta,i,m}\fii\|^2\leq\|\fii\|_1^2,\qquad\fii\in\mc D,
\end{equation}
and
\begin{equation}\label{eq:genNormitus}
\int_{G/H}\sum_{[\eta]\in\hat{H}}\sum_{i=1}^{D_\eta}\sum_{m=1}^{M_\eta}\|L_{\eta,i,m}U(g)^*\fii\|^2\,d\mu(gH)=\|\fii\|^2,\qquad\fii\in\mc D.
\end{equation}
\end{definition}

Note that, using Equation \eqref{eq:genHinv}, the integrand in Equation \eqref{eq:genNormitus} is found to be invariant in the replacement $g\to gh$ whenever $h\in H$ in exactly the same way as earlier in Section \ref{sec:finInstr}; this allows us to interpret the integrand as a function on $G/H$. The following theorem is a generalized version of Theorem \ref{theor:CovInstrStructure}.

\begin{theorem}\label{theor:genCovInstrStr}
Let $\mc I:\mc B(G/H)\times\mc T(\hil)\to\mc T(\mc K)$ be a $(G/H,U,V)$--covariant instrument. There are, for any $[\eta]\in\hat{H}$, $M_\eta\in\N\cup\{0,\infty\}$ and a minimal set
$$
\{L_{\eta,i,m}:\mc D\to\mc K\,|\,m=1,\ldots,M_\eta,\ i=1,\ldots,D_\eta,\ [\eta]\in\hat{H}\}
$$
of $(G/H,U,V)$--intertwiners such that, for all $X\in\mc B(G/H)$, $B\in\mc L(\mc K)$, and $\fii\in\mc D$,
\begin{equation}\label{eq:genCovInstrStr}
\<\fii|\mc I^*(X,B)\fii\>=\int_X\sum_{[\eta]\in\hat{H}}\sum_{i=1}^{D_\eta}\sum_{m=1}^{M_\eta}\<V(g)L_{\eta,i,m}U(g)^*\fii|BV(g)L_{\eta,i,m}U(g)^*\fii\>\,d\mu(gH).
\end{equation}
On the other hand, given $M_\eta\in\N\cup\{0,\infty\}$ for any $[\eta]\in\hat{H}$ and a minimal set of $(G/H,U,V)$--intertwiners consisting of $L_{\eta,i,m}:\mc D\to\mc K$, $[\eta]\in\hat{H}$, $i=1,\ldots,D_\eta$, $m=1,\ldots,M_\eta$, Equation \eqref{eq:genCovInstrStr} defines a $(G/H,U,V)$--covariant instrument.
\end{theorem}

\begin{proof}
Let $\mc I:\mc B(G/H)\times\mc T(\hil)\to\mc T(\mc K)$ be an $(\Set,U,V)$--covariant instrument and let $\msf M:\mc B(G/H)\to\mc L(\hil)$ be the (possibly continuous) quantum observable (POVM) measured by $\mc I$, i.e.,\ $\tr{\rho\msf M(X)}=\tr{\mc I(X,\rho)}$ for all $\rho\in\mc T(\hil)$ and $X\in\mc B(G/H)$. According to \cite{Cattaneo} (see also \cite{HaPe15}), there is a strongly continuous unitary representation $\pi_0:H\to\mc U(\hil_{\pi_0})$ in a separable Hilbert space $\hil_{\pi_0}$, and an isometry $W_0:\hil\to L^2_\mu\otimes\hil_{\pi_0}$, such that the quadruple $(L^2_\mu\otimes\hil_{\pi_0},\msf P_{\pi_0}^G,U_{\pi_0}^G,W_0)$ is a minimal $(G/H,U)$--covariant Na\u{\i}mark dilation for $\msf M$, i.e.,\ $\msf M(X)=W_0^*\msf P_{\pi_0}^G(X)W_0$ for all $X\in\mc B(G/H)$, $W_0U(g)=U_{\pi_0}^G(g)W_0$ for all $g\in G$, $U_{\pi_0}^G(g)\msf P_0(X)U_{\pi_0}^G(g)^*=\msf P_{\pi_0}^G(gX)$ for all $g\in G$ and $X\in\mc B(G/H)$ (a property which holds for a canonical system of imprimitivity as has been already stated), and the subspace spanned by the vectors $\msf P_{\pi_0}^G(X)W_0\fii$, $X\in\mc B(G/H)$, $\fii\in\hil$, is dense in $L_\mu^2\otimes\hil_{\pi_0}$. Moreover, according to Proposition 6 of \cite{HaPe15}, there is a strongly continuous unitary representation $\pi:H\to\mc U(\hil_\pi)$ in a separable Hilbert space $\hil_\pi$ and an isometry $\Lambda:\hil_{\pi_0}\to\mc K\otimes\hil_\pi$ with the property $\Lambda\pi_0(h)=\big(V(h)\otimes\pi(h)\big)\Lambda$ for all $h\in H$ such that, defining the decomposable isometry $W:L_\mu^2\otimes\hil_{\pi_0}\to \mc K\otimes L_\mu^2\otimes\hil_\pi$ through $(Wf)(x)=W(x)f(x)$ for all $f\in L_\mu^2\otimes\hil_{\pi_0}$ and $x\in G/H$, where
\begin{equation}\label{eq:WKomponentit}
W(gH)=\big(V(g)\otimes\zeta^\pi(g^{-1},gH)\big)\Lambda\zeta^{\pi_0}(g^{-1},gH)^*,\qquad g\in G,
\end{equation}
the vectors $(B\otimes\id_{L_\mu^2\otimes\hil_\pi})Wf$, $B\in\mc L(\mc K)$, $f\in L_\mu^2\otimes\hil_{\pi_0}$, span a dense subspace of $\mc K\otimes L_\mu^2\otimes\hil_\pi$ and $\mc I^*(X,B)=W_0^*\msf P_{\pi_0}^G(X)W^*(B\otimes\id_{L_\mu^2\otimes\hil_\pi})WW_0$ for all $X\in\mc B(G/H)$ and $B\in\mc L(\mc K)$. Noting that $WU_{\pi_0}^G(g)=U_\pi^G(g)W$ for all $g\in G$ and that $W\msf P_{\pi_0}^G(X)=\msf P_\pi^G(X)W$ for all $X\in\mc B(G/H)$, the quadruple $(L_\mu^2\otimes\hil_\pi,\msf P_\pi^G,U_\pi^G,J)$, where $J=WW_0$, is a minimal $(G/H,U,V)$--covariant Stinespring dilation for $\mc I$.

According to item (c) above, the isometry $W_0$ can be chosen so that there is a linear operator $\Theta:\mc D\to\hil_{\pi_0}$ with the property $\Theta U(h)=\pi_0(h)\Theta$ for all $h\in H$ so that $(W_0\fii)(gH)=\zeta^{\pi_0}(g^{-1},gH)\Theta U(g)^*\fii$ for all $g\in G$ and $\fii\in\mc D$. Moreover, $\|\Theta\fii\|\leq\|\fii\|_1$ for all $\fii\in\mc D$. We define the linear operators $J(gH):\mc D\to\mc K\otimes\hil_\pi$ through
$$
J(gH)=W(gH)\zeta^{\pi_0}(g^{-1},gH)\Theta U(g)^*=\big(V(g)\otimes\zeta^\pi(g^{-1},gH)\Lambda\Theta U(g)^*,\qquad g\in G,
$$
where we have used Equation \eqref{eq:WKomponentit}. According to the Peter-Weyl theorem, for each $[\eta]\in\hat{H}$, there is a separable Hilbert space $\mc M_\eta$ so that $\hil_\pi=\bigoplus_{[\eta]\in\hat{H}}\mc K_\eta\otimes\mc M_\eta$ and $\pi(h)=\bigoplus_{[\eta]\in\hat{H}}\eta(h)\otimes\id_{\mc M_\eta}$ for all $h\in H$. Denote, for each $[\eta]\in\hat{H}$, $M_\eta:={\rm dim}\,\mc M_\eta\in\{0,\infty\}\cup\N$, and let $\{f_{\eta,m}\}_{m=1}^{M_\eta}$ be an orthonormal basis for $\mc M_\eta$ for all $[\eta]\in\hat{H}$. Defining, for all $[\eta]\in\hat{H}$, $i=1,\ldots,D_\eta$, and $m=1,\ldots,M_\eta$, the isometry $V_{\eta,i,m}:\mc K\to\mc K\otimes\hil_\pi$ through $V_{\eta,i,m}\psi=\psi\otimes e_{\eta,i}\otimes f_{\eta,m}$ for all $\psi\in\mc K$, we denote $L_{\eta,i,m}:=V_{\eta,i,m}^*\Lambda\Theta$.

Proving that the set consisting of the operators $L_{\eta,i,m}$ is $(\mc K,\mc D)$-weakly independent is carried out in exactly the same way as the corresponding proof in Section \ref{sec:finInstr}. Pick $[\eta]\in\hat{H}$, $i=1,\ldots,D_\eta$, $m=1,\ldots,M_\eta$, and $h\in H$. We have
\begin{align*}
L_{\eta,i,m}U(h)&=V_{\eta,i,m}^*\Lambda\Theta U(h)=V_{\eta,i,m}^*\Lambda\pi_0(h)\Theta=V_{\eta,i,m}^*\big(V(h)\otimes\pi(h)\big)\Lambda\Theta\\
&=V(h)V_{\eta,i,m}^*\big(\id_{\mc K}\otimes\pi(h)\big)\Lambda\Theta=\sum_{j=1}^{D_\eta}\eta_{i,j}(h)V(h)L_{\eta,j,m}
\end{align*}
where we have used $\big(\id_{\mc K}\otimes\pi(h)\big)V_{\eta,i,m}=\sum_{j=1}^{D_\eta}\eta_{j,i}(h)V_{\eta,j,m}$ (which is easily proven) in the final equality, thus proving Equation \eqref{eq:genHinv}. Using the Pythagorean theorem and the fact that $\Lambda$ is an isometry, we find $\sum_{[\eta]\in\hat{H}}\sum_{i=1}^{D_\eta}\sum_{m=1}^{M_\eta}\|L_{\eta,i,m}\fii\|^2=\sum_{[\eta]\in\hat{H}}\sum_{i=1}^{D_\eta}\sum_{m=1}^{M_\eta}\|V_{\eta,i,m}^*\Lambda\Theta\fii\|^2=\|\Lambda\Theta\fii\|^2=\|\Theta\fii\|^2\leq\|\fii\|_1^2$ for all $\fii\in\mc D$, implying Inequality \eqref{eq:DNormi}. Using the fact that, for all $\fii\in\mc D$ and $g\in G$, $(J\fii)(gH)=J(gH)\fii$, we find, for all $\fii\in\mc D$, $X\in\mc B(G/H)$, and $B\in\mc L(\mc K)$,
\begin{align*}
&\<\fii|\mc I^*(X,B)\fii\>=\<J\fii|\big(B\otimes\msf P(X)\big)J\fii\>=\int_X\<J(gH)\fii|(B\otimes\id_{\hil_\pi})J(gH)\fii\>\,d\mu(gH)\\
&=\int_X\<\big(V(g)\otimes\zeta^\pi(g^{-1},gH)\big)\Lambda\Theta U(g)^*\fii|\big(BV(g)\otimes\zeta^\pi(g^{-1},gH)\big)\Lambda\Theta U(g)^*\fii\>\,d\mu(gH)\\
&=\int_X\<\Lambda\Theta U(g)^*\fii|\big(V(g)^*BV(g)\otimes\id_{\hil_\pi}\big)\Lambda\Theta U(g)^*\fii\>\,d\mu(gH)\\
&=\int_X\sum_{[\eta]\in\hat{H}}\sum_{i=1}^{D_\eta}\sum_{m=1}^{M_\eta}\<\Lambda\Theta U(g)^*\fii|V_{\eta,i,m}V(g)^*BV(g)V_{\eta,i,m}^*\Lambda\Theta U(g)^*\fii\>\,d\mu(gH)\\
&=\int_X\sum_{[\eta]\in\hat{H}}\sum_{i=1}^{D_\eta}\sum_{m=1}^{M_\eta}\<V(g)L_{\eta,i,m}U(g)^*\fii|BV(g)L_{\eta,i,m}U(g)^*\fii\>\,d\mu(gH),
\end{align*}
proving Equation \eqref{eq:genCovInstrStr}. The proof of the converse claim is straight-forward and is left for the reader; note that Equation \eqref{eq:genNormitus} corresponds to the normalization condition $\mc I^*(G/H,\id_{\mc K})=\id_\hil$.
\end{proof}

We again have the following elaboration for the final claim of Theorem \ref{theor:genCovInstrStr} stating that we may construct minimal covariant dilations from minimal sets of intertwiners.

\begin{lemma}\label{lemma:genminlemma}
Given, for each $[\eta]\in\hat{H}$, the number $M_\eta\in\{0,\infty\}\cup\N$, let
$$
\{L_{\eta,i,m}:\mc D\to\mc K\,|\,m=1,\ldots,M_\eta,\ i=1,\ldots,D_\eta,\ [\eta]\in\hat{H}\}
$$
be a minimal set of $(G/H,U,V)$--intertwiners and, for all $[\eta]\in\hat{H}$, $i=1,\ldots,D_\eta$, $m=1,\ldots,M_\eta$, and $g\in G$, define
\begin{equation}\label{eq:genapu1}
K_{\eta,i,m}(gH):=\sum_{j=1}^{D_\eta}\zeta^\eta_{i,j}(g^{-1},gH)V(g)L_{\eta,j,m}U(g)^*.
\end{equation}
For each $[\eta]\in\hat{H}$, let $\mc M_\eta$ be an $M_\eta$-dimensional Hilbert space with the orthonormal basis $\{f_{\eta,m}\}_{m=1}^{M_\eta}$ and define the strongly continuous unitary representation $\pi:H\to\mc U(\hil_\pi)$ where $\hil_\pi=\bigoplus_{[\eta]\in\hat{H}}\mc K_\eta\otimes\mc M_\eta$ and $\pi(h)=\bigoplus_{[\eta]\in\hat{H}}\eta(h)\otimes\id_{\mc M_\eta}$ for all $h\in H$ and the isometry $J:\hil\to\mc K\otimes L_\mu^2\otimes\hil_\pi$ such that, for all $\fii\in\mc D$ and $g\in G$,
$$
(J\fii)(gH)=\sum_{[\eta]\in\hat{H}}\sum_{i=1}^{D_\eta}\sum_{m=1}^{M_\eta}K_{\eta,i,m}(gH)\fii\otimes e_{\eta,i}\otimes f_{\eta,m},
$$
the quadruple $(L_\mu^2\otimes\hil_\pi,\msf P_\pi^G,U_\pi^G,J)$ is a minimal $(G/H,U,V)$--covariant Stinespring dilation for the $(G/H,U,V)$--covariant instrument $\mc I$ defined through Equation \eqref{eq:genCovInstrStr}.
\end{lemma}

\begin{proof}
Let $M_\eta\in\{0,\infty\}\cup\N$ for each $[\eta]\in\hat{H}$ and suppose that operators $L_{\eta,i,m}:\mc D\to\mc K$, $[\eta]\in\hat{H}$, $i=1,\ldots,D_\eta$, $m=1,\ldots,M_\eta$, constitute a minimal set of $(G/H,U,V)$--intertwiners and define the representation $\pi$ and the linear map $J$ as in the claim. Direct calculation utilizing Equation \eqref{eq:genNormitus} shows that $\|J\fii\|=\|\fii\|$ for all $\fii\in\mc D$. Since $\mc D$ is a dense subspace of $\hil$, this means that $J$ indeed can be extended into an isometry $J:\hil\to\mc K\otimes L_\mu^2\otimes\hil_\pi$. Thus, equation $\mc I^*(X,B)=J^*\big(B\otimes\msf P(X)\big)J$ for all $X\in\mc B(G/H)$ and $B\in\mc L(\mc K)$ defines an instrument $\mc I:\mc B(G/H)\times\mc T(\hil)\to\mc T(\mc K)$. Checking $JU(g)=\big(V(g)\otimes U_\pi^G(g)\big)J$ for all $g\in G$ and $X\in\mc B(G/H)$ is straight-forward and is left for the reader. Let us concentrate on showing that the vectors $\big(B\otimes\msf P_\pi^G(X)\big)J\fii$, $B\in\mc L(\mc K)$, $X\in\mc B(G/H)$, $\fii\in\hil$, span a dense subspace of $\mc K\otimes L_\mu^2\otimes\hil_\pi$.

Proving that the set $\{K_{\eta,i,m}(x)\,|\,m=1,\ldots,M_\eta,\ i=1,\ldots,D_\eta,\ [\eta]\in\hat{H}\}$ is $(\mc K,\mc D)$-weakly independent for any $x\in G/H$ is carried out in essentially the same way as in the proof of Lemma \ref{lemma:minlemma}. Let $\Psi\in\mc K\otimes L_\mu^2\otimes\hil_\pi$ be such that $\<\Psi|\big(B\otimes\msf P_\pi^G(X)\big)J\fii\>=0$ for all $B\in\mc L(\mc K)$, $X\in\mc B(G/H)$, and $\fii\in\hil$. We may assume that, for any $[\eta]\in\hat{H}$, $i=1,\ldots,D_\eta$, and $m=1,\ldots,M_\eta$, there is a field $G/H\ni x\mapsto\psi_{\eta,i,m}(x)\in\mc K$ such that $\Psi(x)=\sum_{[\eta]\in\hat{H}}\sum_{i=1}^{D_\eta}\sum_{m=1}^{M_\eta}\psi_{\eta,i,m}(x)\otimes e_{\eta,i}\otimes f_{\eta,m}$ for all $x\in G/H$, so that we may assume that $\sum_{[\eta]\in\hat{H}}\sum_{i=1}^{D_\eta}\sum_{m=1}^{M_\eta}\|\psi_{\eta,i,m}(x)\|^2=\|\Psi(x)\|^2<\infty$ for all $x\in G/H$. Essentially in the same way as in the proof of Lemma \ref{lemma:minlemma}, we find that, for any $B\in\mc L(\mc K)$, $X\in\mc B(G/H)$, and $\fii\in\mc D$,
$$
0=\<\Psi|\big(B\otimes\msf P_\pi^G(X)\big)J\fii\>=\int_X\sum_{[\eta]\in\hat{H}}\sum_{i=1}^{D_\eta}\sum_{m=1}^{M_\eta}\<\psi_{\eta,i,m}(x)|BK_{\eta,i,m}(x)\fii\>\,d\mu(x).
$$
Substituting above $B=|\psi\>\<\psi'|$ where $\psi,\,\psi'\in\mc K$ and varying $X\in\mc B(G/H)$, we find that, for any $\fii\in\mc D$ and $\psi,\,\psi'\in\mc K$, there is a $\mu$-measurable set $N_{\fii,\psi,\psi'}\subset G/H$ such that $\mu(N_{\fii,\psi\psi'})=0$ and
\begin{equation}\label{eq:minapu}
\sum_{[\eta]\in\hat{H}}\sum_{i=1}^{D_\eta}\sum_{m=1}^{M_\eta}\<\psi_{\eta,i,m}(x)|\psi\>\<\psi'|K_{\eta,i,m}(x)\fii\>=0
\end{equation}
for all $x\in(G/H)\setminus N_{\fii,\psi,\psi'}$. Let $C_{\mc D}\subset\mc D$ be a countable set which is dense in $\mc D$ w.r.t.\ the 1-norm (recall the assumption (b) we made in the beginning of this section) and $C_{\mc K}\subset\mc K$ be a countable set dense in $\mc K$ w.r.t.\ the usual Hilbert space topology. Define $N:=\bigcup\{N_{\fii,\psi,\psi'}\,|\,\fii\in C_{\mc D},\ \psi,\,\psi'\in C_{\mc K}\}$. Clearly, $\mu(N)=0$. Pick $\fii\in\mc D$, $\psi,\,\psi'\in\mc K$ and let $(\fii_r)_{r=1}^\infty\subset C_{\mc D}$, $(\psi_r)_{r=1}^\infty\subset C_{\mc K}$, and $(\psi'_r)_{r=1}^\infty\subset C_{\mc K}$ be sequences such that $\lim_{r\to\infty}\|\fii-\fii_r\|_1=\lim_{r\to\infty}\|\psi-\psi_r\|=\lim_{r\to\infty}\|\psi'-\psi'_r\|=0$. Using the Pythagorean theorem, the Cauchy-Schwarz inequality and Equation \eqref{eq:DNormi}, we may easily evaluate, for any $x\in G/H$,
\begin{align*}
&\left|\sum_{[\eta]\in\hat{H}}\sum_{i=1}^{D_\eta}\sum_{m=1}^{M_\eta}\left(\<\psi_{\eta,i,m}(x)|\psi\>\<\psi'|K_{\eta,i,m}(x)\fii\>-\<\psi_{\eta,i,m}(x)|\psi_r\>\<\psi'_r|K_{\eta,i,m}(x)\fii_r\>\right)\right|\\
\leq&\|\Psi(x)\|\|\psi-\psi_r\|\|\psi'\|\|\fii\|_1+\|\Psi(x)\|\|\psi_r\|\|\psi'-\psi'_r\|\|\fii\|_1+\|\Psi(x)\|\|\psi_r\|\|\psi'_r\|\|\fii-\fii_r\|_1\\
\leq&\|\Psi(x)\|\|\psi-\psi_r\|\|\psi'\|\|\fii\|_1+\|\Psi(x)\|\big(\|\psi-\psi_r\|+\|\psi\|\big)\|\psi'-\psi'_r\|\|\fii\|_1\\
&+\|\Psi(x)\|\big(\|\psi-\psi_r\|+\|\psi\|\big)\big(\|\psi'-\psi'_r\|+\|\psi'\|\big)\|\fii-\fii_r\|_1\overset{r\to\infty}{\rightarrow}0,
\end{align*}
From this, it immediately follows that, Equation \eqref{eq:minapu} holds for all $\fii\in\mc D$, $\psi,\,\psi'\in\mc K$, and $x\in(G/H)\setminus N$. Since, for any $x\in G/H$ and $\psi\in\mc K$, the sequence $\big(\<\psi|\psi_{\eta,i,m}(x)\>\,\big|\,i=1,\ldots,D_\eta,\ m=1,\ldots,M_\eta,\ [\eta]\in\hat{H}\big)$ is square summable, and the set of operators $K_{\eta,i,m}(x)$, $[\eta]\in\hat{H}$, $i=1,\ldots,D_\eta$, $m=1,\ldots,M_\eta$, is $(\mc K,\mc D)$-weakly independent, it follows that $\<\psi|\psi_{\eta,i,m}(x)\>=0$ for all $[\eta]\in\hat{H}$, $i=1,\ldots,D_\eta$, $m=1,\ldots,M_\eta$, $\psi\in\mc K$, and $x\in(G/H)\setminus N$. This means that $\psi_{\eta,i,m}(x)=0$ for all $[\eta]\in\hat{H}$, $i=1,\ldots,D_\eta$, $m=1,\ldots,M_\eta$, and $x\in(G/H)\setminus N$. Hence, $\Psi(x)=0$ for $\mu$-a.a.\ $x\in G/H$, i.e.,\ $\Psi=0$ finalizing the proof.
\end{proof}

As before, we say that a $(G/H,U,V)$--covariant instrument $\mc I$ is an {\it extreme instrument of the $(G/H,U,V)$--covariance structure} if $\mc I$ is an extreme point of the convex set of all $(G/H,U,V)$--covariant instruments. Using Theorem \ref{theor:genCovInstrStr} and Lemma \ref{lemma:genminlemma} and earlier extremality characterizations from \cite{HaPe15}, we may describe all the extreme instruments of the $(G/H,U,V)$--covariance structure. For this, we make a couple of technical definitions. Pick, for all $[\eta]\in\hat{H}$, $M_\eta\in\{0,\infty\}\cup\N$, and let $L_{\eta,i,m}:\mc D\to\mc K$, $[\eta]\in\hat{H}$, $i=1,\ldots,D_\eta$, $m=1,\ldots,M_\eta$, constitute a minimal set of $(G/H,U,V)$--intertwiners. Using the Cauchy-Schwarz inequality (in its different forms) and Equation \eqref{eq:genNormitus}, we have, for any $\fii\in\mc D$, $[\eta]\in\hat{H}$, $m,\,n=1,\ldots,M_\eta$,
\begin{align*}
&\left|\int_{G/H}\sum_{i=1}^{D_\eta}\<L_{\eta,i,m}U(g)^*\fii|L_{\eta,i,n}U(g)^*\fii\>\,d\mu(gH)\right|\\
\leq&\int_{G/H}\sum_{i=1}^{D_\eta}\|L_{\eta,i,m}U(g)^*\fii\|\|L_{\eta,i,n}U(g)^*\fii\|\,d\mu(gH)\\
\leq&\int_{G/H}\sqrt{\sum_{i=1}^{D_\eta}\|L_{\eta,i,m}U(g)^*\fii\|^2\sum_{j=1}^{D_\eta}\|L_{\eta,j,n}U(g)^*\fii\|^2}\,d\mu(gH)\\
\leq&\sqrt{\int_{G/H}\sum_{i=1}^{D_\eta}\|L_{\eta,i,m}U(g)^*\fii\|^2\,d\mu(gH)\,\int_{G/H}\sum_{j=1}^{D_\eta}\|L_{\eta,j,n}U(g)^*\fii\|^2\,d\mu(gH)}\\
\leq&\int_{G/H}\sum_{[\eta]\in\hat{H}}\sum_{i=1}^{D_\eta}\sum_{r=1}^{M_\eta}\|L_{\eta,i,r}U(g)^*\fii\|^2\,d\mu(gH)=\|\fii\|^2.
\end{align*}
This means that the map $\mc D^2\ni(\fii,\psi)\mapsto\int_{G/H}\sum_{i=1}^{D_\eta}\<L_{\eta,i,m}U(g)^*\fii|L_{\eta,i,n}U(g)^*\psi\>\,d\mu(gH)\in\C$ is a bounded sesquilinear form for all $[\eta]\in\hat{H}$ and $m,\,n=1,\ldots,M_\eta$ (and, thus, extends to $\hil^2$); we denote the corresponding bounded linear operator as
$$
\int_{G/H}\sum_{i=1}^{D_\eta}U(g)L_{\eta,i,m}^*L_{\eta,i,n}U(g)^*\,d\mu(gH)\in\mc L(\hil).
$$
Moreover, given sets $A\neq\emptyset$ and $B_a\neq\emptyset$ for any $a\in A$, we say that a set consisting of $B_{a,b,c}\in\mc L(\hil)$, $b,\,c\in B_a$, $a\in A$, is strongly independent if, for any decomposable bounded operator $\bigoplus_{a\in A}(\beta_{a,b,c})_{b,c\in B_a}\in\bigoplus_{a\in A}\mc L(\ell_{B_a}^2)\subset\mc L\left(\bigoplus_{a\in B_a}\ell_{B_a}^2\right)$, the condition $\sum_{a\in A}\sum_{b,c\in B_a}\beta_{a,b,c}B_{a,b,c}=0$ (where the series is required to converge strongly) implies $\beta_{a,b,c}=0$ for all $a\in A$ and $b,\,c\in B_a$.

\begin{theorem}\label{theor:genExtInstr}
Let $\mc I$ be a $(G/H,U,V)$--covariant instrument defined through Equation \eqref{eq:genCovInstrStr} by a minimal set of $(G/H,U,V)$--intertwiners consisiting of $L_{\eta,i,m}:\mc D\to\mc K$, $[\eta]\in\hat{H}$, $i=1,\ldots,D_\eta$, $m=1,\ldots,M_\eta$, where $M_\eta\in\{0,\infty\}\cup\N$. This instrument is an extreme instrument of the $(G/H,U,V)$--covariance structure if and only if the set
$$
\left\{\int_{G/H}\sum_{i=1}^{D_\eta}U(g)L_{\eta,i,m}^*L_{\eta,i,n}U(g)^*\,d\mu(gH)\,\bigg|\,m,\,n=1,\ldots,M_\eta,\ [\eta]\in\hat{H}\right\}
$$
is strongly independent.
\end{theorem}

\begin{proof}
Let $(L_\mu^2\otimes\hil_\pi,\msf P_\pi^G,U_\pi^G,J)$ be the minimal $(G/H,U,V)$--covariant Stinespring dilation for $\mc I$ defined by $L_{\eta,i,m}$, $[\eta]\in\hat{H}$, $i=1,\ldots,D_\eta$, $m=1,\ldots,M_\eta$ as in Lemma \ref{lemma:genminlemma}. According to \cite{HaPe15}, $\mc I$ is an extreme instrument of the $(G/H,U,V)$--covariance structure if and only if, for $E\in\mc L(L_\mu^2\otimes\hil_\pi)$, the conditions $E\msf P_\pi^G(X)=\msf P_\pi^G(X)E$ for all $X\in\mc B(G/H)$, $EU_\pi^G(g)=U_\pi^G(g)E$ for all $g\in G$, and $J^*(\id_{\mc K}\otimes E)J=0$ imply $E=0$. This is why we next focus on characterizing the intersection of the commutant of the range of $\msf P_\pi^G$ and that of the range of $U_\pi^G$.

Suppose that $E\in\mc L(L_\mu^2\otimes\hil_\pi)$ commutes with $\msf P_\pi^G$ and $U_\pi^G$. The former condition implies that there is a (weakly) $\mu$-measurable field $G/H\ni x\mapsto E(x)\in\mc L(\hil_\pi)$ such that $(EF)(x)=E(x)F(x)$ for all $F\in L_\mu^2\otimes\hil_\pi$ and $x\in G/H$. Fix a left Haar measure $\mu_G$ for $G$. Requiring that $EU_\pi^G(g)=U_\pi^G(g)E$ for all $g\in G$ easily yields that, for all $g\in G$, there is $N_g\in\mc B(G/H)$ such that $\mu(N_g)=0$ and
\begin{equation}\label{eq:commutantintersectionapu}
E(x)\zeta^\pi(g,x)=\zeta^\pi(g,x)E(gx)
\end{equation}
for all $x\in(G/H)\setminus N_g$.

Denote by $N$ the set of those $(g,x)\in G\times G/H$ such that Equation \eqref{eq:commutantintersectionapu} does not hold. Since $\hil_\pi$ is separable, this is easily seen to be a Borel set. Using the Fubini theorem, we get
$$
(\mu_G\times\mu)(N)=\int_N d(\mu_G\times\mu)=\int_G\underbrace{\int_{G/H}\chi_N(g,x)\,d\mu(x)}_{=0}\,d\mu_G(g)=0,
$$
implying that Equation \eqref{eq:commutantintersectionapu} holds for $(\mu_G\times\mu)$-a.a.\ $(g,x)\in G\times G/H$. Using the Fubini theorem for a second time, we find
$
0=(\mu_G\times\mu)(N)=\int_N d(\mu_G\times\mu)=\int_{G/H}\int_G\chi_N(g,x)\,d\mu_G(g)\,d\mu(x)
$
and, since $\int_G\chi_N(g,x)\,d\mu_G(g)\geq0$ for all $x\in G/H$, this means that $\int_G\chi_N(g,x)\,d\mu_G(g)=0$ for $\mu$-a.a.\ $x\in G/H$. This means that we may pick $x_0\in G/H$ with the property $\chi_N(g,x_0)=0$ for $\mu_G$-a.a.\ $g\in G$. This means that, for $\mu_G$-a.a.\ $g\in G$,
\begin{equation}\label{eq:commutantsint2}
E(x_0)\zeta^\pi(g,x_0)=\zeta^\pi(g,x_0)E(gx_0).
\end{equation}
Since $G$ is locally compact and second countable, we may assume that the set $Y\in\mc B(G)$ of those $g\in G$ such that Equation \eqref{eq:commutantsint2} holds (and whose complement is $\mu_G$-null) is a countable union of compact sets, implying that $X:=\{gH\,|\,g\in Y\}$ is a Borel-measurable subset of $G/H$. The pre-image of $(G/H)\setminus X$ under the factor projection $g\mapsto gH$ is contained within the $\mu_G$-null $G\setminus Y$. Since, according to Corollary V.5.16 of \cite{varadarajan}, a set $Z\in\mc B(G/H)$ is $\mu$-null if and only if its pre-image under the factor projection is $\mu_G$-null, we have that $\mu\big((G/H)\setminus X\big)=0$. It now follows from the above and Equation \eqref{eq:commutantsint2}, for all $g\in G$ such that $gs(x_0)^{-1}\in Y$, i.e.,\ for $\mu_G$-a.a.\ $g\in G$,
\begin{align}
E(gH)&=E(gs(x_0)^{-1}x_0)=\zeta^\pi\big(gs(x_0)^{-1},x_0\big)^*E(x_0)\zeta^\pi\big(gs(x_0)^{-1},x_0\big)\nonumber\\
&=\Big(\zeta^\pi\big(s(x_0)^{-1},x_0\big)\zeta^\pi(g,H)\Big)^*E(x_0)\zeta^\pi\big(s(x_0)^{-1},x_0\big)\zeta^\pi(g,H)\nonumber\\
&=\zeta^\pi(g,H)^*E_0\pi(g,H)=\pi\big(g^{-1}s(gH)\big)^*E_0\pi\big(g^{-1}s(gH)\big)\label{eq:Dkentta}
\end{align}
where we have denoted $E_0:=\zeta^\pi\big(s(x_0)^{^{-1}},x_0\big)^*E(x_0)\zeta^\pi\big(s(x_0)^{^{-1}},x_0\big)$.

Denote by $N_1$ the $\mu_G$-measurable subset of those $g\in G$ such that Equation \eqref{eq:commutantsint2} does not hold. Since we have, for every $f\in L^1(G)$, $\int_G f\,d\mu_G=\int_{G/H}\int_H f(gh)\,d\mu_H(h)\,d\mu(gH)$, where $\mu_H$ is the essentially unique left Haar measure on $H$, we have
$$
0=\mu_G(N_1)=\int_G\chi_{N_1}\,d\mu_G=\int_{G/H}\underbrace{\int_H\chi_{N_1}(gh)\,d\mu_H(h)}_{\geq0}\,d\mu(gH),
$$
implying that, for $\mu_G$-a.a.\ $g\in G$ (i.e.,\ for $\mu$-a.a.\ $gH\in G/H$) $\int_H\chi_{N_1}(gh)\,d\mu_H(h)=0$. It follows that there is $g_0\in G$ such that $\chi_{N_1}(g_0h)=0$ for $\mu_H$-a.a.\ $h\in H$. Since $\mu_G(N_1)=0$, we may assume that $g_0\in G\setminus N_1$. Thus, we find that, for $\mu_H$-a.a.\ $h\in H$, $\pi(h)E_0\pi(h)^*=\pi\big(g_0^{-1}s(g_0H)\big)\pi\big(h^{-1}g_0^{-1}s(g_0H)\big)^*E_0\pi\big(h^{-1}g_0^{-1}s(g_0H)\big)\pi\big(g_0^{-1}s(g_0H)\big)^*
=\pi\big(g_0^{-1}s(g_0H)\big)E(g_0H)\times\\
\times\pi\big(g_0^{-1}s(g_0H)\big)^*=E_0$ where we have used the fact that $g_0\in G\setminus N_1$ in the final equality. Using the strong continuity of $\pi$, this means that $E_0\pi(h)=\pi(h)E_0$ for all $h\in H$. Using Equation \eqref{eq:Dkentta}, this means that $E(x)=E_0$ for $\mu$-a.a.\ $x\in G/H$. Thus the intersection of the commutant of the range of $\msf P_\pi^G$ and that of the range of $U_\pi^G$ is included within the set of those operators $E\in\mc L(L_\mu^2\otimes\hil_\pi)$ defined by some $E_0\in\mc L(\hil_\pi)$ commuting with the range of $\pi$ through $(EF)(x)=E_0F(x)$ for all $F\in L_\mu^2\otimes\hil_\pi$ and $x\in G/H$. The converse inclusion is immediate. Thus the intersection we are studying corresponds to the commutant of the range of $\pi$.

Let $E\in\mc L(L_\mu^2\otimes\hil_\pi)$ commute with $\msf P_\pi^G$ and $U_\pi^G$ and let $E_0$ be the corresponding operator in the commutant of $\pi$. Using the definition $\pi(h)=\bigoplus_{[\eta]\in\hat{H}}\eta(h)\otimes\id_{\mc M_\eta}$ for all $h\in H$, we find that there is a bounded sequence $\hat{H}\ni[\eta]\mapsto E_\eta\in\mc L(\mc M_\eta)$ such that $E(x)=E_0=\bigoplus_{[\eta]\in\hat{H}}\id_{\mc K_\eta}\otimes E_\eta$ for $\mu$-a.a.\ $x\in G/H$. Define, for all $[\eta]\in\hat{H}$, $i=1,\ldots,D_\eta$, and $m=1,\ldots,M_\eta$, the isometry $V_{\eta,i,m}:\mc K\to\mc K\otimes\hil_\pi$ as earlier. Denoting $\beta_{\eta,m,n}:=\<f_{\eta,m}|E_\eta f_{\eta,n}\>$ for all $[\eta]\in\hat{H}$ and $m,\,n=1,\ldots,M_\eta$, we find that, for any $\fii\in\mc D$,
\begin{align*}
&\<J\fii|(\id_{\mc K}\otimes E)J\fii\>=\int_{G/H}\<(J\fii)(x)|(\id_{\mc K}\otimes E_0)(J\fii)(x)\>\,d\mu(x)\\
=&\int_{G/H}\sum_{[\eta],[\tj]\in\hat{H}}\sum_{i=1}^{D_\eta}\sum_{j=1}^{D_\tj}\sum_{m=1}^{M_\eta}\sum_{n=1}^{M_\tj}\<V_{\eta,i,m}^*(J\fii)(x)|\underbrace{V_{\eta,i,m}^*(\id_{\mc K}\otimes E_0)V_{\tj,j,n}}_{\begin{scriptsize}
=\<f_{\eta,m}|E_\eta f_{\eta,n}\>\delta_{[\eta],[\tj]}\delta_{i,j}
\end{scriptsize}}V_{\tj,j,n}^*(J\fii)(x)\>\,d\mu(x)\\
=&\sum_{[\eta]\in\hat{H}}\sum_{m,n=1}^{M_\eta}\<f_{\eta,m}|E_\eta f_{\eta,n}\>\int_{G/H}\sum_{i=1}^{D_\eta}\<K_{\eta,i,m}(x)\fii|K_{\eta,i,n}(x)\fii\>\,d\mu(x)\\
=&\sum_{[\eta]\in\hat{H}}\sum_{m,n=1}^{M_\eta}\beta_{\eta,m,n}\int_{G/H}\sum_{j,k=1}^{D_\eta}\underbrace{\sum_{i=1}^{D_\eta}\overline{\zeta^\eta_{i,j}(g^{-1},gH)}\zeta^\eta_{i,k}(g^{-1},gH)}_{\begin{scriptsize}
\delta_{j,k}
\end{scriptsize}}\<L_{\eta,j,m}U(g)^*\fii|L_{\eta,k,n}U(g)^*\fii\>\,d\mu(gH)\\
=&\sum_{[\eta]\in\hat{H}}\sum_{m,n=1}^{M_\eta}\beta_{\eta,m,n}\int_{G/H}\sum_{i=1}^{D_\eta}\<L_{\eta,i,m}U(g)^*\fii|L_{\eta,i,n}U(g)^*\fii\>\,d\mu(gH).
\end{align*}
Noting that the set of decomposable bounded operators in $\bigoplus_{[\eta]\in\hat{H}}\ell_{B_\eta}^2$, where $B_\eta$ is the set of indices $m=1,\ldots,M_\eta$ for any $[\eta]\in\hat{H}$, coincides with the set of $\bigoplus_{[\eta]\in\hat{H}}\big(\<f_{\eta,m}|E_\eta f_{\eta,n}\>\big)_{m,n=1}^{M_\eta}$, where $\hat{H}\ni[\eta]\mapsto E_\eta\in\mc L(\mc M_\eta)$ is a bounded sequence, the claim now follows from the extremality characterization stated at the beginning of this proof.
\end{proof}

\begin{remark}\label{rem:genext}
Given a measurable space $(\Omega,\Sigma)$, we say that an instrument $\mc I:\Sigma\times\mc T(\hil)\to\mc T(\mc K)$ is an {\it extreme instrument} if it is a convex extreme point of the convex set of all instruments with outcome space $(\Omega,\Sigma)$, input Hilbert space $\hil$, and output Hilbert space $\mc K$. Let $\mc I$ be a $(G/H,U,V)$--covariant instrument defined through Equation \eqref{eq:genCovInstrStr} by a minimal set of $(G/H,U,V)$--intertwiners $L_{\eta,i,m}$, $[\eta]\in\hat{H}$, $i=1,\ldots,D_\eta$, $m=1,\ldots,M_\eta$. For brevity, let us denote the set of indices $(\eta,i,m)$, where $[\eta]\in\hat{H}$, $i=1,\ldots,D_\eta$, and $m=1,\ldots,M_\eta$, by $B$.  
Using the minimal Stinespring dilation of Lemma \ref{lemma:genminlemma} and an earlier extremality characterization given in \cite{InstrutI} and recalling the section $s:G/H\to G$, we find that the above $\mc I$ is an extreme instrument if and only if, for a family $\{f^\beta_\gamma\}_{\beta,\gamma\in B}\subset L^\infty_\mu$ such that $G/H\ni x\mapsto\big(f^\beta_\gamma(x)\big)_{\beta,\gamma\in B}\in\mc L(\ell_B^2)$ is $\mu$-essentially bounded, the condition
$$
\int_{G/H}\sum_{\beta,\gamma\in B}f^\beta_\gamma(x)\<L_\beta(U\circ s)(x)^*\fii|L_\gamma(U\circ s)(x)^*\fii\>\,d\mu(x)=0
$$
for all $\fii\in\mc D$ implies $f^\beta_\gamma(x)=0$ for all $\beta,\,\gamma\in B$ and $\mu$-a.a.\ $x\in G/H$. This fact is proven in Appendix D.
\hfill $\triangle$
\end{remark}

Let us give an extremality condition which is particularly convenient when the input representation $U$ is irreducible. We formulate this result, not using minimal intertwiners, but using a particular minimal covariant dilation of a $(G/H,U,V)$--covariant instrument into a canonical system of imprimitivity. Note that we do not have assume that $H$ is compact.

\begin{proposition}\label{prop:irrep}
Let $\mc I$ be a $(G/H,U,V)$--covariant instrument and let $\pi:H\to\mc U(\hil_\pi)$ be a strongly continuous unitary representation, where $\hil_\pi$ is separable, and $J:\hil\to\mc K\otimes L^2_\mu\otimes\hil_\pi$ be an isometry such that $(L^2_\mu\otimes\hil_\pi,\msf P_\pi,U_\pi,J)$ is a minimal $(G/H,U,V)$--covariant Stinespring dilation for $\mc I$. If $\pi$ is irreducible, then $\mc I$ is an extreme instrument of the $(G/H,U,V)$--covariance structure. If $U$ is irreducible, also the converse claim holds.
\end{proposition}

\begin{proof}
For the duration of this proof, define the map $\mc L(\hil_\pi)\ni E\mapsto E^\bullet\in\mc L(L^2_\mu\otimes\hil_\pi)$ through $(E^\bullet f)(x)=Ef(x)$ for all $E\in\mc L(\hil_\pi)$, $f\in L^2_\mu\otimes\hil_\pi$, and $x\in G/H$. Suppose first that $\pi$ is irreducible. This means that the commutant $({\rm ran}\,\pi)'$ of the range of $\pi$ is $\C\id_{\hil_\pi}$. The commutant $({\rm ran}\,U_\pi)'$ of the range of $U_\pi$ is, according to the proof of Theorem \ref{theor:genExtInstr}, the image of $({\rm ran}\,\pi)'$ under the map $E\mapsto E^\bullet$. Clearly, this means that $({\rm ran}\,U_\pi)'=\C\id_{L^2_\mu\otimes\hil_\pi}$. (This just means that, when $\pi$ is irreducible, then also $U_\pi$ is irreducible which is well known.) Obviously, the map $({\rm ran}\,U_\pi)'\ni D\mapsto J^*(\id_{\mc K}\otimes D)J\in\mc L(\hil)$ is now injective, meaning that $\mc I$ is an extreme instrument of the $(G/H,U,V)$--covariance structure.

Suppose then that $U$ is irreducible and $\mc I$ is an extreme instrument of the $(G/H,U,V)$--covariance structure. Using the intertwining property $JU(g)=\big(V(g)\otimes U_\pi(g)\big)J$ for all $g\in G$ and the fact that $({\rm ran}\,U_\pi)'$ is the image of $({\rm ran}\,\pi)'$ under $E\mapsto E^\bullet$, it easily follows that $U(g)J^*(\id_{\mc K}\otimes E^\bullet)J=J^*(\id_{\mc K}\otimes E^\bullet)JU(g)$ for all $g\in G$ and $E\in({\rm ran}\,\pi)'$, implying that, for all $E\in({\rm ran}\,\pi)'$, there is $z(E)\in\C$ such that $J^*(\id_{\mc K}\otimes E^\bullet)J=z(E)\id_\hil$, i.e.,\ $0=J^*\id_{\mc K}\otimes\big(E^\bullet-z(E)\id_{L^2_\mu\otimes\hil_\pi}\big)J=J^*\id_{\mc K}\otimes\big(E-z(E)\id_{\hil_\pi}\big)^\bullet J$. Since $\mc I$ is an extreme instrument of the $(G/H,U,V)$--covariance structure, the extremality condition given in \cite{HaPe15} (which has also appeared in the proof of Theorem \ref{theor:genExtInstr}) implies that $E=z(E)\id_{\hil_\pi}$ for all $E\in({\rm ran}\,\pi)'$, i.e.,\ $({\rm ran}\,\pi)'=\C\id_{\hil_\pi}$ meaning that $\pi$ is irreducible.
\end{proof}

\begin{remark}\label{rem:kiet}
Let, for each $[\eta]\in\hat{H}$, $M_\eta\in\{0,\infty\}\cup\N$, and let $L_{\eta,i,m}:\mc D\to\mc K$, $[\eta]\in\hat{H}$, $i=1,\ldots,D_\eta$, $m=1,\ldots,M_\eta$, constitute a minimal set of $(G/H,U,V)$--intertwiners. Define, for all $[\eta]\in\hat{H}$ and $i=1,\ldots,D_\eta$, the isometry $V_{\eta,i}:\mc K\to\mc K_\eta\otimes\eta$ through $V_{\eta,i}\psi=e_{\eta,i}\otimes\psi$ for all $\psi\in\mc K$. This allows us to define the operators $B_{\eta,m}:\mc D\to\mc K_\eta\otimes\mc K$ for all $[\eta]\in\hat{H}$ and $m=1,\ldots,M_\eta$ through
$$
B_{\eta,m}=\sum_{i=1}^{D_\eta}V_{\eta,i}L_{\eta,i,m}.
$$
Thus, $L_{\eta,i,m}=V_{\eta,i}^*B_{\eta,m}$ for all $[\eta]\in\hat{H}$, $i=1,\ldots,D_\eta$, and $m=1,\ldots,M_\eta$ and one easily finds that
\begin{equation}\label{eq:BIntertwine}
B_{\eta,m}U(h)=\big(\eta(h)\otimes V(h)\big)B_{\eta,m},\qquad [\eta]\in\hat{H},\quad m=1,\ldots,M_\eta,\quad h\in H.
\end{equation}
This intertwining property can often be easier to verify that the property of Equation \eqref{eq:genHinv} using Clebsch-Gordan methods.

The instrument defined by the intertwiners $L_{\eta,i,m}$, $[\eta]\in\hat{H}$, $i=1,\ldots,D_\eta$, $m=1,\ldots,M_\eta$, is an extreme instrument of the $(G/H,U,V)$--covariance structure if and only if the set
$$
\left\{\int_G U(g)B_{\eta,m}^*B_{\eta,n}U(g)^*\,d\mu_G(g)\,\bigg|\,m,\,n=1,\ldots,M_\eta,\ [\eta]\in\hat{H}\right\}
$$
is strongly independent. The operators
$$
\int_G U(g)B_{\eta,m}^*B_{\eta,n}U(g)^*\,d\mu_G(g)=\int_{G/H} U(g)B_{\eta,m}^*B_{\eta,n}U(g)^*\,d\mu(gH)
$$
are defined in the same way as the integrated operators in the claim of Theorem \ref{theor:genExtInstr}. The above equality follows from Equation \eqref{eq:BIntertwine} upon choosing $\mu$ so that the associated left Haar measure $\mu_H$ of $H$ (i.e.,\ the left Haar measure of $H$ such that $\int_G f\,d\mu_G=\int_{G/H}\int_H f(gh)\,d\mu_H(h)\,d\mu(gH)$ for all $f\in L^1(G)$) is normalized, i.e.,\ $\mu_H(H)=1$. Similarly, we have
$$
\int_{G/H}\sum_{i=1}^{D_\eta}U(g)L_{\eta,i,m}^*L_{\eta,i,n}U(g)^*\,d\mu(gH)=\int_G\sum_{i=1}^{D_\eta}U(g)L_{\eta,i,m}^*L_{\eta,i,n}U(g)^*\,d\mu_G(g)
$$
for all $[\eta]\in\hat{H}$ and $m,\,n=1,\ldots,M_\eta$ which can be substituted in the claim of Theorem \ref{theor:genExtInstr}. In particular, these operators commute with the representation $U$.\hfill $\triangle$
\end{remark}

\begin{example}
We finally study the case of covariant phase space measurements and the corresponding instruments. The pre-measurement system is a quantum system with $N$ degrees of freedom and associated with the Hilbert space $L^2(\R^N)$ and the post-measurement system has $N'$ degrees of freedom and is associated with the Hilbert space $L^2(\R^{N'})$ in the position representation. The position shifts act on the states by shifting the argument of a state vector associated with a pure state, i.e.,\ through the unitary representation $U_N:\R^N\to\mc U\big(L^2(\R^N)\big)$, $\big(U_N(\vec{q})\fii\big)(\vec{x})=\fii(\vec{x}-\vec{q})$ for all $\vec{q}\in\R^N$, $\fii\in L^2(\R^N)$, and a.a.\ $\vec{x}\in\R^N$. The momentum boosts are hence associated with the unitary representation $V_N:\R^N\to\mc U\big(L^2(\R^N)\big)$, $V_N(\vec{p})=\mc F^*U_N(\vec{p})\mc F$ for all $\vec{p}\in\R^N$, where $\mc F$ is the unitary Fourier transform operator, i.e.,\ for all $\vec{p}\in\R^N$, $\fii\in L^2(\R^N)$, and a.a.\ $\vec{x}\in\R^N$, $\big(V_N(\vec{p})\fii\big)(\vec{x})=e^{i\vec{x}^T\vec{p}}\fii(\vec{x})$. By defining
$$
W_N(\vec{q},\vec{p}):=e^{\frac{i}{2}\vec{q}^T\vec{p}}U_N(\vec{q})V_N(\vec{p}),\qquad\vec{q},\,\vec{p}\in\R^N,
$$
we are able to encapsulate position shifts and momentum boosts into phase space translations giving rise to a projective unitary representation $W_N:\R^N\times\R^N\to\mc U\big(L^2(\R^N)\big)$. Indeed, one easily checks that, upon defining the $(2N\times 2N)$-matrix
$$
S_N:=\left(\begin{array}{cc}
0&\id_N\\
-\id_N&0
\end{array}\right)
$$
in the block form and denoting the phase space points by $\vec{z}=(\vec{q},\vec{p})\in\R^{2N}$, we have
\begin{equation}\label{eq:Wmult}
W_N(\vec{z}+\vec{w})=e^{\frac{i}{2}\vec{z}^T S_N\vec{w}}W_N(\vec{z})W_N(\vec{w}),\qquad\vec{z},\,\vec{w}\in\R^N.
\end{equation}
This projective representation is called as the {\it Weyl representation}. In quantum optics literature, the operators $W_N(\vec{z})$, $\vec{z}\in\R^N$, are associated to the {\it displacement operators}.

Let us next introduce the {\it Weyl-Heisenberg group} $H_N$ which coincides, as a set, with $\R^{2N}\times\T$ and whose group law is given by
$$
(\vec{z},s)(\vec{w},t)=(\vec{z}+\vec{w},ste^{-i\vec{z}^T S_N\vec{w}}),\qquad\vec{z},\,\vec{w}\in\R^{2N},\quad s,\,t\in\T.
$$
Let us also define the map $D_N:H_N\to\mc U\big(L^2(\R^N)\big)$ through
$$
D_N(\vec{z},s)=\overline{s}W_N(\vec{z}),\qquad\vec{z}\in\R^{2N},\quad s\in\T.
$$
Using Equation \eqref{eq:Wmult}, one easily sees that $D_N$ is an ordinary strongly continuous unitary representation. In fact, $H_N$ can be seen as a central extension of the additive group $\R^{2N}$ by the multiplier $(\vec{z},\vec{w})\mapsto e^{-i\vec{z}^T S_N\vec{w}}$ and $D_N$ as the lifting of the Weyl representation $W_N$ to $H_N$.

Let $Y$ be a real $(2N\times 2N')$--matrix such that $Y^T S_{N'} Y=S_N$. We let $U=D_N$ and define $V:H_N\to\mc U\big(L^2(\R^{N'})\big)$ through $V(\vec{w},s):=D_{N'}(Y\vec{w},s)$ for all $\vec{w}\in\R^{2N}$ and $s\in\T$. One may easily check that $V$ is an ordinary unitary representation as well. The value space of the measurements we are interested in is $\R^{2N}$, so that the stability subgroup is $H:=\{0\}\times\T$. Since the restrictions $U|_H$ and $V|_H$ coincide and have values in the respective centres of $\mc L\big(L^2(\R^N)\big)$ and $\mc L\big(L^2(\R^{N'})\big)$, the intertwining property of Equation \eqref{eq:Hinv} becomes irrelevant. Moreover, there is only one $\eta\in\hat{H}$ (the trivial one) appearing in this scenario. This means that the relevant (minimal) sets of $(\R^{2N},U,V)$--intertwiners are (weakly independent) sets $\{L_m\}_{m=1}^M\subset\mc L\big(L^2(\R^N),L^2(\R^{N'})\big)$, with $M\in\N\cup\{\infty\}$, of Hilbert-Schmidt operators such that
$$
\sum_{m=1}^M \tr{L_m^*L_m}=\frac{1}{\pi^N}.
$$
Indeed, perusing \cite{kiukas_etal2006} and Section 6.1 of \cite{HaPe15}, we see that, in item (a) of the beginning of this section, we may choose $\mc D=L^2(\R^N)$ and, in item (b), $\|\cdot\|_1$ can be chosen as the ordinary Hilbert norm so that the intertwiners are simply bounded operators. The Hilbert-Schmidt property follows from the square-integrability of $U$, i.e.,\ for all unit vectors $\fii,\,\psi\in L^2(\R^N)$,
$$
\int_\T\int_{\R^{2N}}|\<\fii|U(\vec{z},s)\psi\>|^2\,d\vec{z}\,ds=\int_{\R^{2N}}|\<\fii|W_N(\vec{z})\psi\>|^2\,d\vec{z}=\pi^N
$$
which, in turn, implies, according to Lemma 2 of \cite{kiukas_etal2006} that, for positive $A\in\mc L\big(L^2(\R^N)\big)$ and $T\in\mc T\big(L^2(\R^N)\big)$, the function $\R^{2N}\ni\vec{z}\mapsto\tr{W_N(\vec{z})TW_N(\vec{z})^*A}$ is Lebesgue-integrable if and only if $A\in\mc T\big(L^2(\R^N)\big)$ in which case $\int_{\R^{2N}}\tr{W_N(\vec{z})TW_N(\vec{z})^*A}\,d\vec{z}=\pi^N\tr{T}\tr{A}$.

We say that an instrument $\mc I:{\mc B}(\R^{2N})\times\mc T\big(L^2(\R^N)\big)\to\mc T\big(L^2(\R^{N'})\big)$ is a {\it covariant phase space instrument} if it is $(\R^{2N},U,V)$--covariant, i.e.,\ for all $\vec{z}\in\R^{2N}$, $X\in{\mc B}(\R^{2N})$, and $\rho\in\mc S\big(L^2(\R^N)\big)$,
$$
\mc I\big(X+\vec{z},W_N(\vec{z})\rho W_N(\vec{z})^*\big)=W_{N'}(Y\vec{z})\mc I(X,\rho)W_{N'}(Y\vec{z})^*.
$$
For any covariant phase space instrument $\mc I$ there is $M\in\N\cup\{\infty\}$ and a minimal set $\{L_m\}_{m=1}^M$ of $(\R^{2N},U,V)$--intertwiners like those above such that
$$
\mc I(X,\rho)=\int_X\sum_{m=1}^M W_{N'}(Y\vec{z})L_m W_N(\vec{z})^*\rho W_N(\vec{z})L_m^*W_{N'}(Y\vec{z})^*\,d\vec{z}
$$
for all $X\in\mc B(\R^{2N})$ and $\rho\in\mc S\big(L^2(\R^N)\big)$. The observable measured by $\mc I$ is easily seen to coincide with $\msf M_S$,
$$
\msf M_S(X)=\frac{1}{\pi^N}\int_X W_N(\vec{z})SW_N(\vec{z})^*\,d\vec{z},\qquad X\subseteq\R^{2N}\ {\rm (measurable)},
$$
defined by $S=\pi^N\sum_{m=1}^M L_m^*L_m\in\mc S\big(L^2(\R^N)\big)$. Moreover, this covariant phase space instrument $\mc I$ is an extreme point of the $(\R^{2N},U,V)$--covariance structure if and only if $M=1$. Indeed, if $M=1$, extremality follows immediately from Theorem \ref{theor:genExtInstr}. If, on the other hand, $M>1$, then, using Lemma 2 of \cite{kiukas_etal2006}, we have that $\int_{\R^{2N}}W_N(\vec{z})L_m^*L_n W_N(\vec{z})\,d\vec{z}$ is a multiple of the identity for $1\leq m,\,n\leq M$. According to Theorem \ref{theor:genExtInstr}, $\mc I$ cannot be an extreme instrument of the $(\R^{2N},U,V)$--covariance structure. 

According to Remark \ref{rem:genext}, a covariant phase space instrument $\mc I$ associated with the intertwiners $L_m$, $m=1,\ldots,M\in\N\cup\{\infty\}$ is an extreme instrument if and only if, for $\{f_{m,n}\}_{m,n=1}^M\subset L^\infty(\R^{2N})$ such that $\R^{2N}\ni\vec{z}\mapsto\big(f_{m,n}(\vec{z})\big)_{m,n=1}^M\in\mc L(\ell^2_{\N_M})$ (where $\N_M$ is the set of indices $m=1,\ldots,\,M$) is an essentially bounded field, the condition
$$
\int_{\R^{2N}}\sum_{m,n=1}^M f_{m,n}(\vec{z})W_N(\vec{z})L_m^*L_nW_N(\vec{z})^*\,d\vec{z}=0
$$
implies $f_{m,n}=0$ for all $m,\,n=1,\ldots,\,M$. However, this extremality characterization is greatly simplified recalling that an extreme instrument is also an extreme instrument of the convex subset of covariant phase space instruments and thus only has one intertwiner, i.e.,\ $M=1$. This can also be proven directly: Assume that the covariant phase space instrument associated with the minimal set $\{L_m\}_{m=1}^M$ of intertwiners is an extreme instrument. We make the counter assumption that $M\geq2$, so that $L_1$ and $L_2$ are non-zero, implying that $\|L_1\|_{HS}\neq0\neq\|L_2\|_{HS}$ where $\|K\|_{HS}=\sqrt{\tr{K^*K}}$ is the Hilbert-Schmidt norm of the Hilbert-Schmidt operator $K$. Let us define the constant functions $f_{1,1}\equiv\|L_1\|_{HS}^{-2}$, $f_{2,2}\equiv-\|L_2\|_{HS}^{-2}$, and $f_{m,n}\equiv0$ otherwise for $m,\,n=1,\ldots,M$. Using Lemma 2 of \cite{kiukas_etal2006}, it easily follows that
$$
\int_{\R^{2N}}\sum_{m,n=1}^M f_{m,n}(\vec{z})W_N(\vec{z})L_m^*L_nW_N(\vec{z})^*\,d\vec{z}=0\quad\Longrightarrow\quad(f_{m,n})_{m,n=1}^M\equiv0,
$$
where the final implication following from the extremality characterization clearly does not hold. Thus, $M=1$. It finally follows that a covariant phase space instrument $\mc I$ is an extreme instrument if and only if (any) minimal set of intertwiners associated with $\mc I$ is a singleton $\{L\}$ and, for any $f\in L^\infty(\R^{2N})$,
$$
\int_{\R^{2N}}f(\vec{z})W_N(\vec{z})L^*LW_N(\vec{z})^*\,d\vec{z}=0\quad\Longrightarrow\quad f\equiv0.
$$
We note that a covariant phase space instrument is an extreme instrument if and only if its pointwise Kraus rank \cite{InstrutI} is 1 and the covariant phase space observable it measures is an extreme POVM \cite{OptObs}. 
\hfill $\triangle$
\end{example}

\section{Conclusions}

In this work we have presented a comprehensive study of covariant quantum measurements studied in the form of POVMs and instruments. We have given exhaustive characterizations for these covariant measurement devices and for their extremality properties. In particular, in Examples \ref{ex:Sym} and \ref{ex:genSym}, we have introduced a parametrized family $\{\msf M^\alpha\}_{\alpha\geq0}$ of POVMs covariant w.r.t.\ the symmetric group $S_D$ in dimension $D$ where $\msf M^0$ is a rank-1 PVM and, whenever $\alpha>0$, $\msf M^\alpha$ is extreme (within the set of all POVMs) rank-1 informationally complete POVM. Since being a rank-1 PVM and a rank-1 extreme informationally complete POVM are complementary properties for optimal quantum observables according to \cite{OptObs}, we observe the remarkable fact that these complementary classes are just a `small deviation' away from each other in the sense that even a small positive value of $\alpha$ produces a POVM in the second optimality class whereas $\msf M^0$ is firmly in the first class.

There are several questions that remain to be studied in the field of symmetric quantum measurements. Post-processing is a method of producing a new POVM from another one using only classical data processing. In the discrete case, this processing is described by probability (Markov) matrices $(p_{x|y})$: a POVM $\Mo=(\Mo_x)_{x\in\Set}$ is post processed from a POVM $\msf N=(\msf N_y)_{y\in\mathbb Y}$ if there exist conditional probabilites $0\le p_{x|y}\le 1$ such that $\sum_{x\in\Set}p_{x|y}=1$ and $\Mo_x=\sum_{y\in\mathbb Y}p_{x|y}\msf N_y$; we denote this pre-ordering by  $\msf M\leq_{\rm p.p.}\msf N$.
  The post-processing-maximal POVMs, i.e.,\ those POVMs $\msf M$ such that $\msf M\leq_{\rm p.p.}\msf N$ for some POVM $\msf N$ implies $\msf N\leq_{\rm p.p.}\msf M$, have been identified as exactly the rank-1 POVMs \cite{OptObs}. Since it might happen that there is no rank-1 covariant POVM, it is reasonable to study the maximality w.r.t.\ the post-processing pre-order restricted to the class of $(\Set,U)$--covariant POVMs where the $G$-space $\Set$ may vary. Without restricting generality, we may assume that the probability matrices involved are $G$-equivariant.\footnote{Suppose that $\Set$ and $\mathbb Y$ are $G$-spaces and $\Mo$ [resp.\ $\msf N$] is a $(\Set,U)$--covariant [resp.\ $(\mathbb Y,U)$--covariant] POVM such that 
$\Mo_x=\sum_{y\in\mathbb Y}p'_{x|y}\msf N_y$ for some probability matrix $(p'_{x|y})$.
Define the probability matrix $p_{x|y}:=(\# G)^{-1}\sum_{g\in G} p'_{gx|gy}$ which is equivariant: $p_{x|gy}=p_{g^{-1}x|y}$. Since $\Mo_{x}=U(g)^*\Mo_{gx} U(g)=\sum_{y\in\mathbb Y}p'_{gx|y}U(g)^*\msf N_yU(g)=
\sum_{y'\in\mathbb Y}p'_{gx|gy'}U(g)^*\msf N_{gy'}U(g)=\sum_{y'\in\mathbb Y}p'_{gx|gy'}\msf N_{y'}$
one gets $\sum_{y\in\mathbb Y}p_{x|y}\msf N_y=(\#G)^{-1}\sum_{g\in G}\sum_{y'\in\mathbb Y}p'_{gx|gy'}\msf N_{y'}=\Mo_x$.  }
Another important problem arises in the case where there are no rank-1 covariant POVMs:
Might it happen that the only covariant instruments measuring  a covariant POVM $\Mo$ are nuclear (i.e.\ determine the future) although $\Mo$ is not of rank 1?
 Without the requirement of covariance, an observable determines the future if and only if it is of rank 1, implying that post-processing maximality and determination of the future are identical properties. Whether this result also holds for the respective optimality properties restricted to covariance structures is still an open problem.

Determination of the past, i.e.\ informational completeness, is often closely tied to covariance. Indeed, most of the relevant informationally complete POVMs, e.g.\ the covariant phase space observable generated by the vacuum, arise from covariance structures. However, it remains to be determined under which conditions does a covariance structure contain informationally complete observables. Similarly, whether a covariance structure allows a PVM is an interesting question which, however, has been solved in the case of an Abelian symmetry group \cite{HaPe11,Holevo83}.

An observable $\msf M$ determines its values if, for any outcome $x$ (or, in the continuous case, for any set of outcomes) and $\varepsilon>0$ there is an input state $\rho$ such that $p_\rho^{\msf M}(x)=\tr{\rho\msf M_x}>1-\varepsilon$. It easily follows that $\msf M$ determines its values if $\|\msf M_x\|=1$ for all outcomes $x$; this is called as the norm-1 property. Value determination within covariance structures is a further valid avenue of research. In \cite{OptObs}, it was shown that value determination is related to (although not exactly the same as) pre-processing purity: an observable $\msf M=(\msf M_x)_x$ is pre-processing pure if and only if, from $\msf M_x=\Phi^*(\msf N_x)$ for some POVM $\msf N=(\msf N_x)_x$, some channel $\Phi$, and all $x$, it follows that $\msf N_x=\Psi^*(\msf M_x)$ for some channel $\Psi$ and all $x$. This means that $\msf M$ cannot be realized by adding `quantum noise' in the form of a channel to the pre-measurement state and then measuring a genuinely `cleaner' POVM. Such a scenario is called as pre-processing. Within a covariance structure, we cas restrict the quantum noise into covariant channels.\footnote{If $\msf M_x=\Phi^*(\msf N_x)$ where $\msf M$ [resp.\ $\msf N$] is $(\Set,U)$--covariant [resp.\ $(\Set,V)$--covariant] then $\msf M_x=\tilde \Phi^*(\msf N_x)$ where the covariant channel $\tilde \Phi$ is defined by $\tilde\Phi(\rho)=(\#G)^{-1}\sum_{g\in G}V(g)^*\Phi\big(U(g)\rho U(g)^*\big)V(g)$.} In absence of covariance, pre-processing purity was shown in \cite{OptObs} to correspond to the observable being essentially a direct sum of a PVM and some other POVM. How the presence of symmetries affects this characterization is left as a future research problem.

\subsection*{Acknowledgements}
E.H.\ has received funding from the National Natural Science Foundation of China (grant no.\ 11875110).

\section*{Appendix A}

Fix a finite group $G$ and let $m:\,G\times G\to\T$ be a 2-cocycle, i.e.\ it satisfies the cocycle condition
$m(g,hk)m(h,k)\equiv m(gh,k)m(g,h)$. Define a function
$$
t(g):=\prod_{h\in G} m(g,h)\in\T
$$
so that, for all $g,\,h\in G$,
$$
\frac{t(g)t(h)}{t(gh)}=\prod_{k\in G}\frac{m(g,hk)m(h,k)}{m(gh,k)}=m(g,h)^{\#G}.
$$
Hence, we have the least positive integer $p\le\#G$ such that $m(g,h)^p\equiv t'(g)t'(h)/t'(gh)$ for some function $t':G\to\T$. Write $t'(g)=e^{ip\fii(g)}$ where $\fii$ is real valued and define a new 2-cocycle $m'$ via $m'(g,h):=e^{i\fii(gh)}e^{-i\fii(g)}e^{-i\fii(h)}m(g,h)$. Hence, $m'(g,h)^p \equiv 1$. By defining a 2-cocycle $m''(g,h):=m'(g,h)/m'(e,e)$ we also have
$m''(g,h)^p \equiv 1$ and, in addition, $m''(e,e)=1$.

One can replace the projective unitary representation $g\mapsto U(g)$ with the new projective unitary representation $U'(g):=m'(e,e)e^{i\fii(g)}U(g)$. Indeed, $U(gh)=m(g,h)U(g)U(h)$ implies $U'(gh)=m''(g,h)U'(g)U'(h)$. Furthermore, the covariance condition $\Mo_{gx}=U(g) \Mo_x U(g)^*$ equals with $\Mo_{gx}=U'(g) \Mo_x U'(g)^*$ so that, without restricting generality, we may assume that
the multiplier $m$ of $U$ satisfies $m(e,e)=1$ and $m(g,h)^p\equiv1$ for some (minimal) integer $p>0$.

\section*{Appendix B}

Let us make the same assumptions as in Section \ref{sec:finInstr} and fix an $(\Set,U,V)$--covariant instrument $\mc I=(\mc I_x)_{x\in\Set}$ and a minimal Stinespring dilation $(\mc M,\msf P,J)$ for $\mc I$. We first show that there is a unitary representation $\overline{U}:G\to\mc U(\mc M)$ such that $JU(g)=\big(V(g)\otimes\overline{U}(g)\big)J$ for all $g\in G$. In the sequel, we denote, for all $Y\subseteq\Set$, $\mc I_Y:=\sum_{x\in Y}\mc I_x$. Let us pick  $n\in\N$, $B_1,\ldots,B_n\in\mc L(\mc K)$, $x_1,\ldots,x_n\in\Set$, and $\fii_1,\ldots,\fii_n\in\hil$ and define $\xi:=\sum_{i=1}^n(B_i\otimes\msf P_{x_i})J\fii_i$ and $\xi_g:=\sum_{i=1}^n(B_iV(g)^*\otimes\msf P_{gx_i})JU(g)\fii_i$ for all $g\in G$. Using the $(\Set,U,V)$--covariance, we have
\begin{align*}
\|\xi_g\|^2&=\sum_{i,j=1}^n\<JU(g)\fii_i|\big(V(g)B_i^*B_jV(g)^*\otimes\msf P_{gx_i}\msf P_{gx_j}\big)JU(g)\fii_j\>\\
&=\sum_{i,j=1}^n\<U(g)\fii_i|\mc I_{\{gx_i\}\cap\{gx_j\}}^*\big(V(g)B_i^*B_jV(g)^*\big)U(g)\fii_j\>\\
&=\sum_{i,j=1}^n\<\fii_i|\mc I_{\{x_i\}\cap\{x_j\}}^*(B_i^*B_j)\fii_j\>=\|\xi\|^2
\end{align*}
for all $g\in G$. The minimality of $(\mc M,\msf P,J)$ implies that we may define, for each $g\in G$, a unique isometry $\tilde{U}(g)\in\mc L(\mc K\otimes\mc M)$ such that $\tilde{U}(g)(B\otimes\msf P_x)J=\big(BV(g)^*\otimes\msf P_{gx}\big)JU(g)$ for all $B\in\mc L(\mc K)$ and $x\in\Set$. It is easily checked (using again the minimality) that $\tilde{U}(gh)=\tilde{U}(g)\tilde{U}(h)$ for all $g,\,h\in G$ from whence it easily follows that $\tilde{U}:G\to\mc U(\mc K\otimes\mc M)$ is a unitary representation.

Let $\xi\in\mc K\otimes\mc M$ be as above and pick $g\in G$ and $B\in\mc L(\mc K)$. Using covariance, we get
\begin{align*}
\<\xi|\tilde{U}(g)(B\otimes\id_{\mc M})\xi\>&=\sum_{i,j=1}^n\<(B_i\otimes\msf P_{x_i})J\fii_i|(BB_jV(g)^*\otimes\msf P_{gx_j})JU(g)\fii_j\>\\
&=\sum_{i,j=1}^n\<\fii_i|\mc I_{\{x_i\}\cap\{gx_j\}}^*\big(B_i^*BB_jV(g)^*\big)U(g)\fii_j\>\\
&=\sum_{i,j=1}^n\<\fii_i|\mc I_{\{gg^{-1}x_i\}\cap\{gx_j\}}^*\big(V(g)V(g)^*B_i^*BB_jV(g)^*\big)U(g)\fii_j\>\\
&=\sum_{i,j=1}^n\<U(g)^*\fii_i|\mc I_{\{g^{-1}x_i\}\cap\{x_j\}}^*\big(V(g)^*B_i^*BB_j\big)\fii_j\>\\
&=\sum_{i,j=1}^n\<\big(B^*B_iV(g)\otimes\msf P_{g^{-1}x_i}\big)JU(g)^*\fii_i|(B_j\otimes\msf P_{x_j})J\fii_j\>\\
&=\sum_{i,j=1}^n\<\tilde{U}(g)^*(B^*B_i\otimes\msf P_{x_i})J\fii_i|(B_j\otimes\msf P_{x_j})J\fii_j\>=\<\xi|(B\otimes\id_{\mc M})\tilde{U}(g)\xi\>
\end{align*}
which, together with the minimality, implies that $\tilde{U}(g)(B\otimes\id_{\mc M})=(B\otimes\id_{\mc M})\tilde{U}(g)$ for all $g\in G$ and $B\in\mc L(\mc K)$. This means that there is a unique unitary representation $\overline{U}:G\to\mc U(\mc M)$ such that $\tilde{U}(g)=\id_{\mc K}\otimes\overline{U}(g)$ for all $g\in G$. Furthermore, for any $g\in G$, $x\in\Set$, and $\xi$ as above,
\begin{align*}
\big(\id_{\mc K}\otimes\overline{U}(g)\msf P_x\overline{U}(g)^*\big)\xi&=\sum_{i=1}^n\tilde{U}(g)(\id_{\mc K}\otimes\msf P_x)\tilde{U}(g)^*(B_i\otimes\msf P_{x_i})J\fii_i\\
&=\sum_{i=1}^n\tilde{U}(g)\big(B_iV(g)\otimes\msf P_x\msf P_{g^{-1}x_i})JU(g)^*\fii_i\\
&=\sum_{i=1}^n\tilde{U}(g)\big(B_iV(g)\otimes\msf P_{\{x\}\cap\{g^{-1}x_i\}})JU(g)^*\fii_i\\
&=\sum_{i=1}^n(B_i\otimes\msf P_{\{gx\}\cap\{x_i\}})J\fii_i=(\id_{\mc K}\otimes\msf P_{gx})\xi.
\end{align*}
Minimality again implies that $\overline{U}(g)\msf P_x\overline{U}(g)^*=\msf P_{gx}$ for all $g\in G$ and $x\in\Set$.

It follows that the pair $(\overline{U},\msf P)$ is an example of an imprimitivity system. Let us define, for each orbit $\orb\in\Orb$, the Hilbert space $\mc M^{\orb}:=\big(\sum_{x\in\orb}\msf P_x\big)\mc M$ the map $\overline{U}^{\orb}:G\to\mc U(\mc M^{\orb})$, $\overline{U}^{\orb}(g)=\sum_{x\in\orb}\msf P_x\overline{U}(g)|_{\mc M^{\orb}}$ for all $g\in G$, and the PVM $\msf P^{\orb}=(\msf P^{\orb}_x)_{x\in\orb}:=(\msf P_x)_{x\in\orb}$ in $\mc M^{\orb}$. It easily follows that $\overline{U}^{\orb}$ is still a unitary representation and $\overline{U}^{\orb}(g)\msf P^{\orb}_x\overline{U}^{\orb}(g)^*=\msf P^{\orb}_{gx}$ for all $g\in G$ and $x\in\orb$. This means that, for any orbit $\orb$, $(\overline{U}^{\orb},\msf P^{\orb})$ is a transitive system of imprimitivity as $G$ acts transitively in any orbit. Mackey's imprimitivity theorem tells us that, for any orbit $\orb$, we may assume (possibly by tweaking the isometry $J$) that there is a (finite-dimensional) Hilbert space $\hil^{\orb}$ and a unitary representation $\pi^{\orb}:H_{\orb}\to\mc U(\hil^{\orb})$ such that $\mc M^{\orb}=\C^{\# \orb}\otimes\hil^{\orb}$,
\begin{equation}\label{eq:transitiveU}
\big(\overline{U}^{\orb}(g)f\big)(x)=\zeta^{\orb}(g^{-1},x)f(g^{-1}x),\qquad g\in G,\quad f\in\mc M^{\orb},\quad x\in\orb,
\end{equation}
where $\zeta^{\orb}:=\zeta^{\pi^{\orb}}$ is the cocycle associated with $\pi^{\orb}$, and
\begin{equation}\label{eq:transitiveP}
\msf P^{\orb}_x f=f(x),\qquad x\in\orb,\quad f\in\mc M^{\orb}.
\end{equation}
Note that we identify $\mc M^{\orb}$ with the Hilbert space of functions $f:\orb\to\hil^{\orb}$. In total, $(\overline{U},\msf P)$ is a direct sum of these {\it canonical systems of imprimitivity} $(\overline{U}^{\orb},\msf P^{\orb})$ over $\orb\in\Orb$.

\section*{Appendix C}

Let us now directly see how the extremality characterization within the set of all instruments presented in Remark \ref{rem:ext} implies the extremality within the set of $(\Set,U,V)$--covariant instruments. We continue to use the notations fixed in Section \ref{sec:finInstr}. Let us assume that an $(\Set,U,V)$--covariant instrument $\mc I=(\mc I_x)_{x\in\Set}$ is an extreme instrument. Let
$$
\{L_{\eta,i,m}^{\orb}\,|\,m=1,\ldots,M_\eta,\ i=1,\ldots,D_\eta,\ [\eta]\in\hat{H}_{\orb},\ \orb\in\Orb\}
$$
be a minimal set of $(\Set,U,V)$--intertwiners, where $M_\eta\in\{0\}\cup\N$ for all $\orb\in\Orb$ and $[\eta]\in\hat{H}_{\orb}$. Let $\beta_{\eta,m,n}^{\orb}\in\C$, $\orb\in\Orb$, $[\eta]\in\hat{H}_{\orb}$, $m,\,n=1,\ldots,M_\eta$, be such that
$$
\sum_{\orb\in\Orb}\sum_{g\in G}\sum_{[\eta]\in\hat{H}_{\orb}}\sum_{i=1}^{D_\eta}\sum_{m,n=1}^{M_\eta}\beta_{\eta,m,n}^{\orb}L_{\eta,i,m}^{\orb\,*}L_{\eta,i,n}^{\orb}=0.
$$
Denote $\gamma_{x,\eta,\tj,i,j,m,n}=(\# H_{Gx})\beta_{\eta,m,n}^{Gx}$ for all $x\in\Set$ whenever $[\eta]=[\tj]\in\hat{H}_{Gx}$, $i=j\in\{1,\ldots,D_\eta\}$, and $m,\,n=1,\ldots,M_\eta$. Otherwise, $\gamma_{x,\eta,\tj,i,j,m,n}=0$. Using similar tricks as earlier (and denoting by $\delta_{j,k}$ the Kronecker symbol, i.e.,\ $\delta_{j,k}=1$ if $j=k$ and, otherwise, $\delta_{j,k}=0$), we find
\begin{align*}
&\sum_{\orb\in\Orb}\sum_{x\in\orb}\sum_{[\eta],[\tj]\in\hat{H}_{\orb}}\sum_{i=1}^{D_\eta}\sum_{j=1}^{D_\tj}\sum_{m=1}^{M_\eta}\sum_{n=1}^{M_\tj}\gamma_{x,\eta,\tj,i,j,m,n}K_{x,\eta,i,m}^*K_{x,\tj,j,n}\\
=&\sum_{\orb\in\Orb}\sum_{x\in\orb}\sum_{[\eta]\in\hat{H}_{\orb}}\sum_{i=1}^{D_\eta}\sum_{m,n=1}^{M_\eta}(\# H_{\orb})\beta_{\eta,m,n}^{\orb}K_{x,\eta,i,m}^*K_{x,\eta,i,n}\\
=&\sum_{\orb\in\Orb}\sum_{x\in\orb}\sum_{[\eta]\in\hat{H}_{\orb}}\sum_{j,k=1}^{D_\eta}\sum_{m,n=1}^{M_\eta}(\# H_{\orb})\underbrace{\sum_{i=1}^{D_\eta}\overline{\zeta^\eta_{i,j}\big(s_{\orb}(x)^{-1},x\big)}\zeta^\eta_{i,k}\big(s_{\orb}(x)^{-1},x\big)}_{=\delta_{j,k}}\times\\
&\times\beta_{\eta,m,n}^{\orb}U\big(s_{\orb}(x)\big)L_{\eta,j,m}^{\orb\,*}L_{\eta,k,n}^{\orb}U\big(s_{\orb}(x)\big)^*\\
=&\sum_{\orb\in\Orb}\sum_{x\in\orb}\sum_{[\eta]\in\hat{H}_{\orb}}\sum_{i=1}^{D_\eta}\sum_{m,n=1}^{M_\eta}(\# H_{\orb})\beta_{\eta,m,n}^{\orb}U\big(s_{\orb}(x)\big)L_{\eta,i,m}^{\orb\,*}L_{\eta,i,n}^{\orb}U\big(s_{\orb}(x)\big)^*\\
=&\sum_{\orb\in\Orb}\sum_{x\in\orb}\sum_{h\in H_{\orb}}\sum_{[\eta]\in\hat{H}_{\orb}}\sum_{i=1}^{D_\eta}\sum_{m,n=1}^{M_\eta}\beta_{\eta,m,n}^{\orb}U\big(s_{\orb}(x)h\big)L_{\eta,i,m}^{\orb\,*}L_{\eta,i,n}^{\orb}U\big(s_{\orb}(x)h\big)^*\\
=&\sum_{\orb\in\Orb}\sum_{g\in G}\sum_{[\eta]\in\hat{H}_{\orb}}\sum_{i=1}^{D_\eta}\sum_{m,n=1}^{M_\eta}\beta_{\eta,m,n}^{\orb}U(g)L_{\eta,i,m}^{\orb\,*}L_{\eta,i,n}^{\orb}U(g)^*=0.
\end{align*}
Using the extremality of $\mc I$, we now find that $\gamma_{x,\eta,\tj,i,j,m,n}=0$ for all orbits $\orb\in\Orb$, $x\in\orb$, $[\eta],\,[\tj]\in\hat{H}_{\orb}$, $i=1,\ldots,D_\eta$, $j=1,\ldots,D_\tj$, $m=1,\ldots,M_\eta$, and $n=1,\ldots,M_\tj$, implying that $\beta_{\eta,m,n}^{\orb}=0$ for all $\orb\in\Orb$, $[\eta]\in\hat{H}_{\orb}$, and $m,\,n=1,\ldots,M_\eta$. Thus, $\mc I$ is also an extreme instrument of the $(\Set,U,V)$--covariance structure.

\section*{Appendix D}

We now prove the extremality characterization of Remark \ref{rem:genext}. We fix the $(G/H,U,V)$--covariant instrument $\mc I$ of said Remark and retain the notation and definitions therein. Let $(L^2_\mu\otimes\hil_\pi,\msf P_\pi^G,U_\pi^G,J)$ be the minimal $(G/H,U,V)$--covariant Stinespring dilation for $\mc I$ constructed in Lemma \ref{lemma:genminlemma}. According to \cite{InstrutI}, $\mc I$ is extreme if and only if, for $E\in\mc L(L_\mu^2\otimes\hil_\pi)$ such that $\msf P_\pi^G(X)E=E\msf P_\pi^G(X)$ for all $X\in\mc B(G/H)$, the condition $J^*(\id_{\mc K}\otimes E)J=0$ implies $E=0$. Let us fix $E\in\mc L(L_\mu^2\otimes\hil_\pi)$ such that $\msf P_\pi^G(X)E=E\msf P_\pi^G(X)$ for all $X\in\mc B(G/H)$. It follows that there is a $\mu$-measurable field $G/H\ni x\mapsto E(x)\in\mc L(\hil_\pi)$ such that $(DF)(x)=D(x)F(x)$ for all $F\in L^2_\mu\otimes\hil_\pi$ and $x\in G/H$. We define $f^\beta_\gamma\in L^\infty_\mu$ through $f^{\eta,i,m}_{\tj,j,n}(x)=\<e_{\eta,i}\otimes f_{\eta,m}|E(x)(e_{\tj,j}\otimes f_{\tj,n})\>$ for all $x\in G/H$ and $(\eta,i,m),\,(\tj,j,n)\in B$. Using Equation \eqref{eq:genHinv}, we have, for all $(\eta,i,m)\in B$ and $g\in G$,
$$
\sum_{k=1}^{D_\eta}\zeta^\eta_{i,k}(g^{-1},gH)V(g)L_{\eta,k,m}U(g)^*=(V\circ s)(gH)L_{\eta,i,m}(U\circ s)(gH)^*.
$$
Using this and the definitions of Lemma \ref{lemma:genminlemma}, we get, for all $\fii\in\mc D$,
\begin{align*}
&\<J\fii|(\id_{\mc K}\otimes E)J\fii\>=\int_{G/H}\big\<(J\fii)(x)\big|\big(\id_{\mc K}\otimes E(x)\big)(J\fii)(x)\big\>\,d\mu(x)\\
=&\int_{G/H}\sum_{[\eta],[\tj]\in\hat{H}}\sum_{i,k=1}^{D_\eta}\sum_{j,l=1}^{D_\tj}\sum_{m=1}^{M_\eta}\sum_{n=1}^{M_\tj}\overline{\zeta^\eta_{i,k}(g^{-1},gH)}\zeta^\tj_{j,l}(g^{-1},gH)\times\\
&\times\<V(g)L_{\eta,k,m}U(g)^*\fii|V(g)L_{\tj,l,n}U(g)^*\fii\>\<e_{\eta,i}\otimes f_{\eta,m}|E(gH)(e_{\tj,j}\otimes f_{\tj,n})\>\,d\mu(gH)\\
=&\int_{G/H}\sum_{\beta,\gamma\in B}f^\beta_\gamma(x)\<L_\beta(U\circ s)(x)^*\fii|L_\gamma(U\circ s)(x)^*\fii\>\,d\mu(x).
\end{align*}
Noticing that $G/H\ni x\mapsto\big(f^\beta_\gamma(x)\big)_{\beta,\gamma\in B}\in\mc L(\ell^2_B)$ is $\mu$-essentially bounded and that any family $\{f^\beta_\gamma\}_{\beta,\gamma\in B}\subset L^\infty_\mu$ with this property can be reached with a $\mu$-essentially bounded $\mu$-measurable field $G/H\ni x\mapsto E(x)\in\mc L(\hil_\pi)$ through $f^{\eta,i,m}_{\tj,j,n}(x)=\<e_{\eta,i}\otimes f_{\eta,m}|E(x)(e_{\tj,j}\otimes f_{\tj,n})\>$ for all $x\in G/H$ and $(\eta,i,m),\,(\tj,j,n)\in B$ and using the fact that such bounded fields of operators exactly correspond to bounded operators commuting with $\msf P_\pi^G$, we obtain the desired extremality characterization. Also note that, using familiar countability arguments, $E(x)=0$ for $\mu$-a.a.\ $x\in G/H$ for a $\mu$-essentially bounded $\mu$-measurable field $G/H\ni x\mapsto E(x)\in\mc L(\hil_\pi)$ is equivalent with $f^\beta_\gamma(x)=0$ for $\mu$-a.a.\ $x\in G/H$ and all $\beta,\,\gamma\in B$ where $\{f^\beta_\gamma\}_{\beta,\gamma\in B}\subset L^\infty_\mu$ is defined as above and the $\mu$-null set of those $x\in G/H$ for which $f^\beta_\gamma(x)\neq0$ does not have to depend on $\beta,\,\gamma\in B$.

\end{document}